\documentclass[a4paper,UKenglish,cleveref, autoref, thm-restate]{lipics-v2021}
\usepackage[utf8]{inputenc}

\usepackage{subcaption}
\usepackage{mathpartir}
\usepackage{tikz-cd}

\newcounter{vi}
\newcounter{ai}
\newcounter{ti}

\newcommand\lambdaHut{\lambda_{\vartriangleleft}}
\newcommand\T{\mathsf{Top}}
\newcommand\V{\mathrm{V}}

\newcommand\myrel[1]{\mathop{\stackrel{#1}{\longrightarrow}}}
\newcommand\mrelTwo[1]{\mathop{\stackrel{#1}
    {\longrightarrow\hspace{-1.2em}\rightarrow\hspace{.2em}}}}

\newcommand\cas[3]{#1[#2\backslash#3]}

\newcommand\myequiv{\,{\equiv}\,}
\newcommand\myleq{\,{\leq}\,}
\newcommand\myvartriangleleft{\,{\vartriangleleft}\,}

\newcommand\NF{\mathsf{NF}}

\newcommand\Nat{\mathit{Nat}}

\newcommand\Unit{\mathit{Unit}}
\newcommand\Single{\mathit{single}}
\newcommand\UnitExpr{\lambda y \myleq \T. \lambda x \myleq y. y}
\newcommand\SingleExpr{\lambda y \myleq \T. \lambda x \myleq y. x}
\newcommand\IdA{\mathit{id}_A}
\newcommand\IdExprA{\lambda x \myleq A. x}

\newcommand\Co{\mathsf{C}}
\newcommand\Po{\mathsf{Co}}

\newcommand\prevalid{\mathrm{prevalid}}

\newcommand\dom{\mathrm{dom}}
\newcommand\fv{\mathrm{fv}}
\newcommand\wef{\mathrm{wf}}
\newcommand\nil{\mathrm{nil}}

\newcommand\ie{i.e.}

\nolinenumbers

\title{Towards the type safety of Pure Subtype Systems (Full version)}

\author{Valentin Pasquale}{CEA List, Université Paris-Saclay, France}{valentin.pasquale@cea.fr}{https://orcid.org/0009-0009-3422-476X}{}

\author{Álvaro García-Pérez}{CEA List, Université Paris-Saclay, France}{alvaro.garciaperez@cea.fr}{https://orcid.org/0000-0002-9558-6037}{}

\authorrunning{V. Pasquale and A. García-Pérez}

\Copyright{Valentin Pasquale and Álvaro García-Pérez}



\EventEditors{Stefano Guerrini and Barbara K\"{o}nig}
\EventNoEds{2}
\EventLongTitle{34th EACSL Annual Conference on Computer Science Logic (CSL 2026)}
\EventShortTitle{CSL 2026}
\EventAcronym{CSL}
\EventYear{2026}
\EventDate{February 23--28, 2026}
\EventLocation{Paris, France}
\EventLogo{}
\SeriesVolume{363}
\ArticleNo{9}
\begin{document}

\maketitle

\begin{abstract}
  Hutchins' Pure Subtype Systems (PSS) offer a unified framework for types and terms,
  promising significant advancements in language design for features like dependent types
  and higher-order subtyping. However, the theory has been hampered by a critical gap: a
  proof of type safety has remained an open problem for over a decade. The original
  attempt to prove this property relied on the conjectured commutativity of two
  fundamental reduction relations, equivalence and subtyping. Proving transitivity
  elimination, however, requires this commutativity, a property that is notoriously
  difficult to establish for higher-order subtyping systems.

  In this paper, we address this issue by introducing Machine-Based PSS (MPSS), a novel
  reformulation of the original system. MPSS integrates a continuation stack mechanism,
  reminiscent of the Krivine Abstract Machine, to keep track of arguments that are passed
  during function application, enabling more fine-grained reductions. This architectural
  change exposes crucial intermediate reduction steps that were absent in the original
  PSS. The primary contribution of our work is a direct proof that the equivalence and
  subtyping reductions in MPSS commute. This result formally establishes transitivity
  elimination, which is the cornerstone of the inversion lemma required for type safety.
  We conclude by outlining a pathway from our foundational result to a complete, type-safe
  system, thereby paving the way for the practical realization of PSS-based languages.
\end{abstract}

\section{Introduction}
\label{sec:introduction}

Hutchins' Pure Subtype Systems (PSS for short) \cite{Hut2010,hutchinsthese} have been
proposed as a novel approach to type theory that enables a number of advanced language
features for extensibility and genericity. Among other benefits, PSS (hereafter in the
singular) harmoniously solves both the expression problem \cite{Wad98,ZO05} and the
tag-elimination problem \cite{hutchinsthese}. PSS's expressivity allows one to implement
extensible modules based on deep mixin composition, which is a key result for a type
system \cite{hutchinsthese}. By blurring the distinction between types and terms, PSS
naturally combines dependent types and higher-order subtyping, which makes the approach
very promising as a basis for generic, modular, and extensible programming. However, a
proof of type safety for PSS is conspicuously lacking. This crucial property has only
been shown to hold under the conjecture that two key reduction relations in the system,
equivalence and subtyping, commute \cite{hutchinsthese}. The commutativity of these two
reductions is an instance of a recurrent problem in higher-order subtyping known as
\emph{transitivity elimination}.  Consequently, the type safety of PSS has remained an
open problem for more than a decade.

In this paper, we resolve this long-standing open problem by introducing a novel
reformulation of Hutchins’ PSS. Our system is based on a mechanism reminiscent of the
Krivine Abstract Machine, it allows more fine-grained reductions and exposes crucial
intermediate reduction steps that were absent in the original PSS. This architectural
change allows for a direct and elegant proof of transitivity elimination.

Contrary to traditional type systems, PSS harmoniously mixes terms and types. Consider the
natural number $3$ and the type $\Nat$ of natural numbers. In PSS one can have a term that
encodes $\Nat + 3$ with the naive Church encoding of naturals, which type-checks to
$\Nat$. Terms can also be used instead of types: for instance the function
$\lambda x \myleq 3. x$, whose subtype annotation is $3$ (we write $\myleq$ for the
subtyping relation, that is, the substitute of typing in our calculus), is a valid
expression. The latter example is an instance of \emph{bounded quantification}
\cite{pierce1997higher} where the term $3$ is the singleton type $\{3\}$ and the subtyping
relation is akin to the subset relation.

By replacing the typing relation with a more general subtyping relation, PSS greatly
increases expressivity. Since terms and types belong to the same syntactic kind, PSS
naturally subsumes dependent types and higher-order subtyping. This unification enables a
host of advanced language features for extensibility, genericity, and efficiency,
including virtual types, recursive types, deep mixin composition, feature-oriented
programming, bounded quantification, and partial evaluation
\cite{Hut2010,ZO05,hutchinsthese}.

In the original work on PSS, Hutchins introduces a \emph{declarative} system with two
relations: one for equivalence (akin to $\beta$-equivalence) and another for subtyping,
both of which are defined with transitive closure. The proof of type safety requires a
\emph{type preservation} property, which states that the type of a program is preserved
under evaluation. This property, in turn, hinges on a critical inversion lemma: any two
functions in the subtyping relation must have equivalent subtype annotations. This lemma
cannot be proven directly in the declarative system due to the presence of transitivity,
which introduces intermediate terms that may not be structurally related to the functions
being compared. Proving type preservation thus amounts to showing that transitivity is an
admissible rule in the system.

In order to prove transitivity elimination, Hutchins proposes an equivalent
syntax-directed, \emph{algorithmic} system with two different notions of reduction: an
\emph{equivalence reduction} that models $\beta$-reduction, and a \emph{subtyping
  reduction} that models small-step subtyping, where subtyping reduction subsumes
equivalence reduction. For short, we may write ``reduction'' for the $\beta$-reduction,
and ``promotion'' for the subtyping reduction. The declarative subtyping relation is shown
to be equivalent to the combination of subtyping and equivalence reduction sequences as
depicted below. %
\begin{center}
  $u\leq t\quad\textup{iff}\quad\textup{exists}~v~\textup{such that}$\quad
  \begin{tikzcd}
    t \arrow[r, two heads, "\equiv"] & v \\
    & u \arrow[u, two heads, "\leq"]
  \end{tikzcd}
\end{center}
In the algorithmic setting, transitivity can be shown to be admissible if these two
notions of reduction commute. Diagrammatically, commutativity ensures that a transitive
derivation which forms a ``stair'' of reduction and promotion steps can be flattened into
a single angle.

Despite his efforts, Hutchins failed to prove that commutativity holds in his algorithmic
system, which prevented him from proving transitivity elimination and type safety.  The
culprit of this failure is exemplified by a very elementary case (depicted below), which
involves the reduction of a redex and the promotion of the formal parameter in the redex's
abstraction to its annotation. %
\begin{center}
  \begin{tikzcd}
    (\lambda x\myleq t.t)\,v \arrow[r, "\equiv"] & t \\
    (\lambda x\myleq t.x)\,v \arrow[u, "\leq"] \arrow[r, "\equiv"] & v \arrow[u, dash,
    dashed, "?"]
  \end{tikzcd}
\end{center}
The redex $(\lambda x\myleq t.x)v$ on the bottom-left corner of the diagram is reduced (in
the horizontal) to $v$, and promoted (in the vertical) to $(\lambda x\myleq t.t)v$. It is
then unclear how to complete the right edge of the diagram in a single step.

To tackle this problem, Hutchins resorted to simultaneous reduction and to the
\emph{decreasing diagrams} technique of \cite{van1994confluence},
but he failed to assign depths to each reduction step so as to show
that reduction decreases in the way prescribed by the decreasing diagrams technique,
because $\beta$-reduction increases this depth.
Hutchins himself explains this failure in detail in his PhD thesis
(Section~2.7: Confluence and commutativity)\cite{hutchinsthese}.

We reformulate the PSS theory by providing an alternative version of the algorithmic
system with a more fine-grained notion of subtyping, in which we can now prove the
commutativity result. Our version, which we dub \emph{Machine-Based PSS} (MPSS for short),
uses a continuation stack mechanism, reminiscent of the Krivine Abstract Machine (KAM)
\cite{Kri07,ABDM03} to track operands passed to abstractions. This mechanism enables a
direct proof of commutativity. Indeed, by explicitly managing operands on a stack, the
subtyping relation in MPSS exposes all the intermediate terms that are needed to complete
the commutativity diagram directly, without resorting to complex techniques like
decreasing diagrams.  For instance, in our system, the subtyping relation
$(\lambda x\myleq t.x)v \myrel{\leq} (\lambda x\myleq t.t)v$ doesn't exist, instead the
variable $x$ must first be promoted to the operand $v$, which later may be promoted to the
annotation $t$. Exposing this intermediate step, essentially disallowing \emph{premature
  promotion} that exists in Hutchins' system, allows us to complete the above
commutativity diagram:
\begin{center}
  \begin{tikzcd}
    (\lambda x\myleq t.v)\,v \arrow[r, "\equiv"] & v \\
    (\lambda x\myleq t.x)\,v \arrow[u, "\leq"] \arrow[r, "\equiv"] & v \arrow[u, "\leq"]
  \end{tikzcd}
\end{center}

Our detailed contributions are the following:
\begin{itemize}
\item We introduce MPSS, a reformulation of Hutchins' PSS that, among other changes, keeps
  track of the operands that have been passed at a particular scope by using a mechanism
  reminiscent of the continuation stack of the KAM \cite{Kri07,ABDM03}.
  MPSS is a more fine-grained system in the sense that the stack mechanism
  exposes intermediate terms that were absent in Hutchins' system.
\item We prove the commutativity of the subtyping and equivalence reductions of MPSS,
  which enables transitivity elimination and the inversion lemma. Our proof of
  commutativity is direct and proceeds by simple structural induction on terms.
\item We sketch how a type system can be derived from our alternative system,
  thanks to transitivity being now admissible.
\end{itemize}

The remainder of this paper is organised as
follows. Section~\ref{sec:machine-based-system} introduces MPSS. Our system is a
reformulation of Hutchins’ algorithmic framework equipped with a stack mechanism inspired
by the Krivine Abstract Machine (KAM) to track operands, alongside two reductions for
subtyping and equivalence. Section~\ref{sec:commutation} proves the commutativity of the
two reduction relations in MPSS. Section~\ref{sec:type-safety} sketches the design of a
type system built upon MPSS and discusses the path to proving its type safety.
Section~\ref{sec:related} and~\ref{sec:future-work} discuss related and future work, and
Section~\ref{sec:conclusions} concludes.

All the formal proofs of the lemmas and theorems stated in the paper are collected in the
appendix.

\section{Machine-Based Pure Subtype Systems}
\label{sec:machine-based-system}

Machine-Based PSS (MPSS for short) is a reformulation of PSS where the algorithmic
reductions are instrumented with a stack mechanism reminiscent of the Krivine Abstract
Machine (KAM) \cite{Kri07,ABDM03}.

MPSS is an extension of the pure lambda calculus with a constant symbol $\T$ for the most
general supertype and with subtype annotations in the formal parameters of abstractions.
Assuming a countably infinite set of variables $\V$ ranged over by $x, y, \dots$, the
syntax of terms, contexts, and stacks is defined as follows:
\begin{displaymath}
  \begin{array}{rcl}
    t,u,v,\alpha,\ldots \in \Lambda &::=&x \mid \T \mid (\lambda x\myleq t.u) \mid (u\,v)\\
    \Gamma &::=& \varepsilon \mid \Gamma, x\myleq t \mid \Gamma, x\myequiv \alpha\\
    s &::=& \nil \mid \alpha :: s
  \end{array}
\end{displaymath}
A \emph{variable} $x \in \V$ is a term.  The term $\T$ is the \emph{most general
  supertype}, such that any other term is its subtype. An \emph{abstraction}
$(\lambda x\myleq t.u)$ has a \emph{parameter} $x$, a \emph{subtype annotation} $t$, and a
\emph{body} $u$. An \emph{application} $(u\,v)$ has an \emph{operator} $u$ and an
\emph{operand} $v$. The Greek letter $\alpha$ is a metavariable for terms that originate
as operands from the stack.

We assume the usual notions of \emph{free} and \emph{bound variables}, and the usual
\emph{capture-avoiding substitution} function, denoted by $\cas{u}{x}{v}$, that replaces
the free occurrences of variable $x$ in $u$ by $v$, while avoiding the capture of any
bound variable in $u$. We consider terms that result after the renaming of bound variables
to be identical. When needed, we assume that $\alpha$-equivalence is applied at will to
avoid clashing of free variables.

An application of the form $(\lambda x\myleq t.u)v$ is a \emph{redex}.  Contraction of a
redex is modelled by $\beta$-reduction, which rewrites $(\lambda x\myleq t.u)v$ into
$\cas{u}{x}{v}$. Consider the following example of a redex involving the term
$\lambda x\myleq\T.x$ for \emph{universal identity}. Since universal identity is a subtype
of $\T$, the self-application of universal identity, \ie\ the redex
$(\lambda x\myleq\T.x)(\lambda x\myleq\T.x)$, reduces to
$\cas{x}{x}{\lambda x\myleq\T.x}$, which is equal to universal identity by the definition
of the substitution function. We assume conventional parenthesis-dropping conventions
where abstractions associate to the right and applications to the left, and applications
bind stronger than abstractions. We define normal forms as terms that are either Top,
functions that have a subtype annotation and a body in normal form, or a variable applied
to any number of normal forms.

\begin{figure}[h!]
  \textsf{\textbf{Prevalidity of logical contexts}}\hfill\boxed{\text{$\Gamma~\prevalid$}}
  \begin{mathpar}
    \inferrule*[right=Pv-Emp]
    { }
    {\varepsilon~\prevalid}
    \and
    \inferrule*[right=Pv-Ctx]
    {\Gamma~\prevalid\and x\not\in\dom(\Gamma)
      \and \fv(t)\subseteq\dom(\Gamma)}
    {\Gamma,x\myleq t~\prevalid}
    \and
    \inferrule*[right=Pv-EqA]
    {\Gamma~\prevalid\and x\not\in\dom(\Gamma)
      \and \fv(\alpha)\subseteq\dom(\Gamma)}
    {\Gamma,x\myequiv \alpha~\prevalid}
  \end{mathpar}
  \textsf{\textbf{Prevalidity of extended contexts}}\hfill\boxed{\text{$\Gamma;s~\prevalid$}}
  \begin{mathpar}
    \inferrule*[right=Pv-Nil]
    {\Gamma~\prevalid }
    {\Gamma;\nil~\prevalid}
    \and
    \inferrule*[right=Pv-Sta]
    {\Gamma;s~\prevalid\and \fv(\alpha)\subseteq\dom(\Gamma)}
    {\Gamma;\alpha::s~\prevalid}
  \end{mathpar}
  \caption{Prevalidity in MPSS}
  \label{fig:prevalidity-rules}
\end{figure}

Our system uses \emph{logical contexts} (\emph{contexts} for short) $\Gamma$ to store
annotations from abstractions encountered within a given scope. A context $\Gamma$ is a
sequence of \emph{subtype annotations} $x\myleq t$ and \emph{equivalence annotations}
$x\myequiv \alpha$. We write $x\myvartriangleleft t$ to denote either kind of annotation,
and throughout this paper the metavariable $\myvartriangleleft$ designates either $\myleq$
or $\myequiv$. We use $\varepsilon$ for the \emph{empty context} and we write
$\Gamma, x\myvartriangleleft t$ for the operation that appends annotation
$x\vartriangleleft t$ to the subtype annotations of $\Gamma$. By abuse of notation, we
write $\Gamma_1,\Gamma_2$ for the concatenation of contexts $\Gamma_1$ and $\Gamma_2$. We
say that $x$ \emph{has subtype $t$ in} $\Gamma$ (written $x \myleq t \in \Gamma$) if
$x\myleq t$ is the rightmost annotation for $x$ in $\Gamma$. Respectively, we say that $x$
\emph{is equivalent to $\alpha$ in} $\Gamma$ (written $x \myequiv \alpha \in \Gamma$)
under similar conditions.

\begin{figure}[h!]
  \textsf{\textbf{Equivalence
      reduction}}\hfill\boxed{\text{$\Gamma; s \vdash u \myrel{\equiv} v$}}
  \begin{mathpar}
    \inferrule*[right={Me-Pro}]
    {\Gamma; s ~\prevalid \and x \myequiv \alpha \in \Gamma \and \Gamma; s \vdash
      \alpha \myrel{\equiv} \alpha'}
    {\Gamma; s \vdash x \myrel{\equiv} \alpha'}

    \inferrule*[right=Me-Bet]
    {\Gamma; s \vdash u \myrel{\equiv} u' \and \Gamma; \nil \vdash v \myrel{\equiv} v'}
    {\Gamma; s \vdash (\lambda x \myleq t . u) \, v \myrel{\equiv} \cas{u'}{x}{v'}}

    \inferrule*[right=Me-Top]
    {\Gamma; s ~\prevalid}
    {\Gamma; s \vdash \T \myrel{\equiv} \T}

    \inferrule*[right=Me-App]
    {\Gamma; v :: s \vdash u \myrel{\equiv} u' \and \Gamma; \nil \vdash v \myrel{\equiv} v'}
    {\Gamma; s \vdash u \, v \myrel{\equiv} u' \, v'}

    \inferrule*[right=Me-Var]
    {\Gamma; s ~\prevalid}
    {\Gamma; s \vdash x \myrel{\equiv} x}

    \inferrule*[right=Me-Fun]
    {\Gamma; \nil \vdash t \myrel{\equiv} t' \and
      \Gamma, x \myleq t; \nil \vdash u \myrel{\equiv} u'}
    {\Gamma; \nil \vdash \lambda x \myleq t . u \myrel{\equiv} \lambda x \myleq t' . u'}

    \inferrule*[right=Me-TAp]
    {\Gamma; s ~\prevalid}
    {\Gamma; s \vdash \T \, u \myrel{\equiv} \T}

    \inferrule*[right=Me-FOp]
    {\Gamma; \nil \vdash t \myrel{\equiv} t' \and
      \Gamma, x \myequiv \alpha; s \vdash u \myrel{\equiv} u'}
    {\Gamma; \alpha :: s \vdash \lambda x \myleq t . u \myrel{\equiv}
      \lambda x \myleq t' . u'}
  \end{mathpar}

  \textsf{\textbf{Subtyping reduction}}\hfill
  \boxed{\text{$\Gamma;s \vdash u \myrel{\leq} v$}}
  \begin{mathpar}
    \inferrule*[right=Ms-Pro]
    {\Gamma;s~\prevalid \and x \myleq t \in \Gamma}
    {\Gamma;s \vdash x \myrel{\leq} t}

    \inferrule*[right=Ms-Top]
    {\Gamma;s~\prevalid}
    {\Gamma;s \vdash u \myrel{\leq} \T}

    \inferrule*[right=Ms-Equ]
    {\Gamma; s ~\prevalid \and \Gamma; s \vdash u \myrel{\equiv} v}
    {\Gamma;s \vdash u \myrel{\leq} v}

    \inferrule*[right=Ms-App]
    {\Gamma;v :: s \vdash u \myrel{\leq} u'}
    {\Gamma;s \vdash u \, v \myrel{\leq} u' \, v}

    \inferrule*[right=Ms-Fun]
    {\Gamma,x\myleq t; \nil \vdash u \myrel{\leq} u'}
    {\Gamma;\nil \vdash \lambda x \myleq t . u \myrel{\leq} \lambda x \myleq t . u'}

    \inferrule*[right=Ms-FOp]
    {\Gamma,x\myequiv\alpha; s \vdash u \myrel{\leq} u'}
    {\Gamma;\alpha :: s \vdash \lambda x \myleq t . u \myrel{\leq}
      \lambda x\myleq t . u'}
  \end{mathpar}

  \caption{Equivalence and subtyping reduction in MPSS}
  \label{fig:subtyping-equivalence-rules}
\end{figure}

In order to accommodate the stack mechanism, we introduce \emph{extended contexts}
$\Gamma;s$, which are logical contexts $\Gamma$ coupled with a \emph{continuation stack}
$s$. Extended contexts adapt the algorithmic reduction relations to the stack mechanism of
MPSS.

Figure~\ref{fig:prevalidity-rules} introduces the \emph{prevalidity} condition on contexts
that ensures each annotation in a context mentions a different variable, to avoid clashing
of variable names.  We say that $x$ \emph{is in the domain of} $\Gamma$ (written
$x\in\dom(\Gamma)$) if and only if an annotation for variable $x$ occurs in $\Gamma$, \ie,
there exists a term $t$ (or a term $\alpha$) such that $x\myleq t\in\Gamma$ (or
$x\myequiv \alpha\in\Gamma$). The empty context $\varepsilon$ is prevalid
(Rule~\textsc{Pv-Emp} in Figure~\ref{fig:prevalidity-rules}). Adding a subtype annotation
$x\myleq t$, or an equivalence annotation $x\myequiv \alpha$ to a prevalid context
$\Gamma$ yields a prevalid context if the free variables of $t$ (respectively, $\alpha$)
are in the domain of $\Gamma$, and if $x$ is not already in the domain of $\Gamma$
(Rules~\textsc{Pv-Ctx} and \textsc{Pv-EqA}). \emph{Prevalidity for extended contexts}
(also defined in Figure~\ref{fig:prevalidity-rules}) stipulates that a prevalid context
coupled with an empty stack is prevalid (Rule~\textsc{Pv-Nil}), and that an extended
context $\Gamma;\alpha::s$ is prevalid if the free variables of $\alpha$ are in the domain
of $\Gamma$ (Rule~\textsc{Pv-Sta}). We now present our reduction relations.

Figure~\ref{fig:subtyping-equivalence-rules} presents our equivalence ($\myrel{\equiv}$)
and subtyping ($\myrel{\leq}$) reductions. We do away with a declarative system like
Hutchins' system, and our MPSS resembles his algorithmic system (p.~58 of
\cite{hutchinsthese}) where we adopt simultaneous reduction upfront in order to simplify
the commutativity proofs (resembling p.~74 of \cite{hutchinsthese}), and where contexts
are coupled with a stack in order to accommodate the stack mechanism. The equivalence
reduction $\myrel{\equiv}$ models a reflexive, small-step, simultaneous
$\beta$-reduction. Rule~\textsc{Me-Bet} performs $\beta$-reduction, while
Rule~\textsc{Me-Pro} replaces a variable by its equivalent operand from the context.
Rules~\textsc{Me-Var} and \textsc{Me-Top} allow the reduction relation to be reflexive.
Rule~\textsc{Me-TAp} is a technical device from Hutchins' formulation in
\cite{hutchinsthese}, whose objective is to accommodate terms with shape $\T \, u$, which
may pop up in reductions, and which cannot be $\beta$-reduced by Rule~\textsc{Me-Bet} (see
p. 59 of \cite{hutchinsthese} for a further discussion on this issue). \textsc{Me-App}
propagates promotion through unapplied applications by pushing the operand onto the stack
and proceeding with the operator, mimicking KAM’s stack management. Rules~\textsc{Me-Fun}
and \textsc{Me-FOp} distinguish between unapplied and applied abstractions:
\textsc{Me-Fun} propagates promotion by enlarging the context with the abstraction's
annotation and proceeding with the body, and \textsc{Me-FOp} promotes the body of an
applied abstraction by taking the operand from the stack and enlarging the context with an
equivalence annotation.

The subtyping reduction $\myrel{\leq}$ defines promotion in a stack-aware manner.
Rule~\textsc{Ms-Pro} promotes a variable to its subtype annotation, and
Rule~\textsc{Ms-Top} promotes any term to $\T$. Rule~\textsc{Ms-Equ} subsumes equivalence
reduction.  Rules~\textsc{Ms-App}, \textsc{Ms-Fun} and \textsc{Ms-FOp} are the
counterparts of the equivalence reduction rules \textsc{Me-App}, \textsc{Me-Fun} and
\textsc{Me-FOp} respectively.  Because Hutchins' declarative system is not contravariant,
its algorithmic system, which we took inspiration from, does not change the type
annotation of abstractions: in his system, there is no rule
$\lambda x \myleq t. u \myrel{\leq} \lambda x \myleq t'. u$ with assumption
$t' \myrel{\leq} t$.  The main motivation for this design choice is to avoid the
undecidability of the related typechecking, as described in \cite{Pie1994}, and to avoid
complicating the meta-theory even more. Our system follows the same principle, both Rules
\textsc{Ms-Fun} and \textsc{Ms-FOp} do not change the type annotation, and in general
there is no reduction rule to change the type annotations, except from equivalence
reductions that are allowed. Handling the interaction between the stack and the context is
the cornerstone to prove the commutativity result, and to avoid the premature promotion
issues that plagued Hutchins’ approach.

To illustrate our stack mechanism, consider the following subtyping reduction derivation
(prevalidity derivations are elided for brevity):
\begin{mathpar}
  \inferrule*[right=Ms-App]
  {
    \inferrule*[right=Ms-FOp]
    {
      \inferrule*[right=Ms-Equ]
      {
        \inferrule*[]{\ldots}{\Gamma;s~\prevalid}
        \and
        \inferrule*[right=Me-Pro]
        {
          \inferrule*[]{\ldots}{\Gamma~\prevalid}
          \and
          \inferrule*[]{\ldots}{\Gamma;\nil\vdash v\myrel{\equiv}v}
        }
        {\Gamma,x\myequiv v;\nil \vdash x\myrel{\equiv}v}
      }
      {\Gamma,x\myequiv v;\nil \vdash x \myrel{\leq} v}
    }
    {\Gamma;v::\nil \vdash (\lambda x\myleq t.x)\myrel{\leq}(\lambda x\myleq t.v)}
  }
  {\Gamma;\nil \vdash (\lambda x\myleq t.x)v\myrel{\leq}(\lambda x\myleq t.v)v}
\end{mathpar}
Rule~\textsc{Ms-App} pushes the operand $v$ onto the stack. Inside the abstraction, the
variable $x$ is not promoted to its type $t$ but is instead reduced to the operand $v$
retrieved from the context, which was placed there by \textsc{Ms-FOp} after popping from
the stack. Keeping track of operands prevents the premature promotion that caused
Hutchins' commutativity proof to fail.

As in Hutchins' PSS, we define a subtyping relation $\leq$, on which we wish to prove that
transitivity is admissible in the next section, as follows, with $\mrelTwo{}$ being the
transitive closure of $\myrel{}$:
\begin{center}
  $u\leq t\quad\textup{iff}\quad\textup{exists}~v~\textup{such that}$\quad
  \begin{tikzcd}
    t \arrow[r, two heads, "\equiv"] & v \\
    & u \arrow[u, two heads, "\leq"]
  \end{tikzcd}
\end{center}

\section{Transitivity elimination}
\label{sec:commutation}

As explained in the introduction, proving transitivity elimination requires us to show
that the equivalence ($\myrel{\equiv}$) and subtyping ($\myrel{\leq}$) reductions commute,
which we do below in this section. In order to state the commutativity result, we first
define a reduction relation on extended contexts, which captures the evolution of
annotations during reduction.

\noindent
\textsf{\textbf{Reduction of extended context}}\hfill\boxed{\text{$\Gamma; s \rightarrowtail \Gamma'; s'$}}
\begin{mathpar}
  \inferrule*[right=Ct-Ann]
  {\Gamma; s \rightarrowtail \Gamma'; s' \and \Gamma; \nil \vdash t \myrel{\equiv} t'}
  {\Gamma, x \myvartriangleleft t; s \rightarrowtail \Gamma', x \myvartriangleleft t'; s'}

  \inferrule*[right=Ct-Stk]
  {\Gamma; s \rightarrowtail \Gamma'; s' \and \Gamma; \nil \vdash \alpha \myrel{\equiv} \alpha'}
  {\Gamma; \alpha :: s \rightarrowtail \Gamma'; \alpha' :: s'}
\end{mathpar}

With this definition, we can now state our main commutativity result.

\begin{lemma}[$\myrel{\leq}$ and $\myrel{\equiv}$ strongly commute]
  \label{thm:commutativity-theorem}
  Let $\Gamma; s$ be an extended context. Let $t_0$, $t_1$, and $t_2$ be terms. If
  $\Gamma;s\vdash t_0\myrel{\equiv} t_1$ and $\Gamma;s\vdash t_0\myrel{\leq} t_2$, then
  for any extended context $\Gamma'; s'$ such that $\Gamma; s\rightarrowtail \Gamma'; s'$,
  there exists a term $t_3$, such that $\Gamma; s \vdash t_2\myrel{\equiv} t_3$ and
  $\Gamma';s' \vdash t_1\myrel{\leq} t_3$.

  Diagrammatically, the existence of the solid arrows implies the existence of the dashed
  arrows:
  \begin{center}
    \begin{tikzcd}
      t_2 \arrow[r, dashed, "\equiv"]   & t_3                       \\
      t_0 \arrow[r, "\equiv"] \arrow[u, "\leq"] & t_1 \arrow[u, dashed, "\leq"]
    \end{tikzcd}
  \end{center}
\end{lemma}

The generalisation over all possible resulting contexts $\Gamma'; s'$ is crucial for the
inductive proof. Indeed, consider the following commutation diagram, with context
$\Gamma; \nil$:
\begin{center}
  \begin{tikzcd}
    \lambda x \myleq t. t \arrow[r, dashed, "\equiv"] & ?                       \\
    \lambda x \myleq t. x \arrow[r, "\equiv"] \arrow[u, "\leq"] & \lambda x \myleq t'. x \arrow[u, dashed, "\leq"]
  \end{tikzcd}
\end{center}
The only possible arrow that could exist, apart from a promotion of the whole term to
$\T$, is a promotion of the variable $x$ to its subtype annotation $t'$. Therefore, the
top right term is $\lambda x \myleq t'. t'$, and the diagram is completed as follows:
\begin{center}
  \begin{tikzcd}
    \lambda x \myleq t. t \arrow[r, "\equiv"] & \lambda x \myleq t'. t'                      \\
    \lambda x \myleq t. x \arrow[r, "\equiv"] \arrow[u, "\leq"] & \lambda x \myleq t'. x \arrow[u, "\leq"]
  \end{tikzcd}
\end{center}
However, in the proof of this theorem, everything works by induction on the terms.
Therefore, starting from the original diagram, we apply the induction hypothesis on the
following diagram, with context $\Gamma, x \myleq t; \nil$:
\begin{center}
  \begin{tikzcd}
    t \arrow[r, dashed, "\equiv"] & ?                       \\
    x \arrow[r, "\equiv"] \arrow[u, "\leq"] & x \arrow[u, dashed, "\leq"]
  \end{tikzcd}
\end{center}
And as we have seen above, we must complete this diagram with $t'$ being the top right
term, as the only possible completion of the original diagram is the term
$\lambda x \myleq t'. t'$. Therefore, the completed diagram by the induction hypothesis
must be:
\begin{center}
  \begin{tikzcd}
    t \arrow[r, "\equiv"] & t'                      \\
    x \arrow[r, "\equiv"] \arrow[u, "\leq"] & x \arrow[u, "\leq"]
  \end{tikzcd}
\end{center}
Hence the requirement that the context of the right arrow must be
$\Gamma, x \myleq t'; \nil$, and therefore the formulation of the theorem with a
generalisation over contexts.

Similarly, our equivalence reduction relation satisfies the diamond property, on which the
result above rests. A similar generalisation over contexts is required, with two different
reduced extended contexts however, as both the top and the right edge can now do variable
promotions thanks to Rule~\textsc{Me-Pro}.  The pathological case is the following, with
context $\Gamma; \nil$:
\begin{center}
  \begin{tikzcd}
    (\lambda x \myleq t. \alpha_0 \, x) \, \alpha_2 \arrow[r, "\equiv"] & (\lambda x \myleq t. \alpha_1 \, \alpha_2) \, \alpha_3                      \\
    (\lambda x \myleq t. x \, x) \, \alpha_0 \arrow[r, "\equiv"] \arrow[u, "\equiv"] & (\lambda x \myleq t. x \, \alpha_0) \, \alpha_1 \arrow[u, "\equiv"]
  \end{tikzcd}
\end{center}
To solve this pathological case in the induction, one of the sub-induction calls is the
following diagram:
\begin{center}
  \begin{tikzcd}
    \alpha_0 \, x \arrow[r, "\equiv"] & \alpha_1 \, \alpha_2                  \\
    x \, x \arrow[r, "\equiv"] \arrow[u, "\equiv"] & x \, \alpha_0 \arrow[u, "\equiv"]
  \end{tikzcd}
\end{center}
And similarly as above, the top edge is completed with extended context
$\Gamma, x \myequiv \alpha_2; \nil$, whereas the right edge is completed with extended
context $\Gamma, x \myequiv \alpha_1; \nil$.  Both these contexts satisfy the requirement
of the theorem (that is, we have
$\Gamma, x \myequiv \alpha_0; \nil \rightarrowtail \Gamma, x \myequiv \alpha_2; \nil$ and
$\Gamma, x \myequiv \alpha_0; \nil \rightarrowtail \Gamma, x \myequiv \alpha_1; \nil$).

\begin{lemma}[$\myrel{\equiv}$ has the diamond property]
  \label{lem:myrelequiv-has-diamond-property}
  Let $\Gamma_0; s_0$ be an extended context. Let $t_0$, $t_1$, and $t_2$ be terms. If
  $\Gamma_0; s_0 \vdash t_0 \myrel{\equiv} t_1$ and
  $\Gamma_0; s_0 \vdash t_0 \myrel{\equiv} t_2$, then for any extended contexts
  $\Gamma_1; s_1$ and $\Gamma_2; s_2$ such that
  $\Gamma_0; s_0 \rightarrowtail \Gamma_1; s_1$ and
  $\Gamma_0; s_0 \rightarrowtail \Gamma_2; s_2$, there exists a term $t_3$ such that
  $\Gamma_1; s_1 \vdash t_1 \myrel{\equiv} t_3$ and
  $\Gamma_2; s_2 \vdash t_2 \myrel{\equiv} t_3$.

  Moreover, for any variable $x$, if in the derivation of
  $\Gamma_0; s_0 \vdash t_0 \myrel{\equiv} t_1$ (respectively
  $\Gamma_0; s_0 \vdash t_0 \myrel{\equiv} t_2$) there isn't an application of the Rule
  \textsc{Me-Pro} that makes a promotion of variable $x$, then in the derivation
  $\Gamma_2; s_2 \vdash t_2 \myrel{\equiv} t_3$ (respectively
  $\Gamma_1; s_1 \vdash t_1 \myrel{\equiv} t_3$) there won't be an application of the Rule
  \textsc{Me-Pro} that makes a promotion of variable $x$.

  Diagrammatically:
  \begin{center}
    \begin{tikzcd}
      t_2 \arrow[r, dashed, "\equiv"]   & t_3                       \\
      t_0 \arrow[r, "\equiv"] \arrow[u, "\equiv"] & t_1 \arrow[u, dashed, "\equiv"]
    \end{tikzcd}
  \end{center}
\end{lemma}

With commutativity and the diamond property established, we can prove that our system
enjoys the transitivity elimination property.

\begin{theorem}[Transitivity is admissible]
  \label{lem:algorithmic-transitivity-elimination}
  Let $\Gamma; s$ be an extended context. Let $u$ and $v$ be terms. If
  $\Gamma; s \vdash u \leq^* v$ then $\Gamma; s \vdash u \leq v$.
\end{theorem}

In order to illustrate our result, and highlight the differences between PSS and MPSS, we
detail below an example of an application of
Theorem~\ref{lem:algorithmic-transitivity-elimination}:

\paragraph*{Example of transitivity elimination}
Consider the Church encodings for type `Unit' (the type with only one inhabitant, term
$\Unit$ below), the encoding of the constructor `single' of type Unit (term $\Single$
below), and an arbitrary well-formed subtype $A$ that encodes some interesting type. Let
the identity over the elements of subtype $A$ be the term $\IdA$ below, and consider an
arbitrary well-formed term $a$ of subtype $A$ such that some derivation
$\varepsilon; \nil \vdash a \leq^* A$ exists which may possibly be as lengthy or
complicated as it may be.
\begin{displaymath}
  \begin{array}{rcl}
    \Unit   &=& \UnitExpr\\
    \Single &=& \SingleExpr\\
    A       &=& \langle\text{arbitrary subtype}\rangle\\
    \IdA    &=& \IdExprA\\
    a       &=& \langle\text{arbitrary subtype of $A$}\rangle
  \end{array}
\end{displaymath}

Consider the following derivation:
\begin{displaymath}
  \Gamma;\nil\vdash \IdA\,a \leq \Single\,A\,a \leq^* \Unit\,A\,a
\end{displaymath}
Figure~\ref{fig:example:algorithmic-derivation} details all the elementary steps
(diagrammatically) of the above algorithmic derivation. We illustrate the proof of
Theorem~\ref{lem:algorithmic-transitivity-elimination} over this diagram.

\begin{figure}
  \begin{center}
    \begin{tikzcd}
      \Unit \,A \, a \arrow[r, "\equiv"] &
      (\lambda x \myleq A.A) \, a \arrow[r, two heads, "\equiv"] &
      \cdot \arrow[r, two heads, dashed, "\equiv"] & \cdot \arrow[r, two heads, dashed, "\equiv"] & \cdot \\

      & & & & \\

       & & \ddots
       \arrow[uu, draw=none, "\leq"] \arrow[uu, two heads]
       \arrow[r, "\leq", draw=none] \arrow[r, two heads]
       & \cdot \arrow[uu, two heads, dashed, "\leq"]
       \arrow[uul, dash, draw=none, "a \leq^* A", sloped]
       & \\

       & & & (\lambda x \myleq A. a) \, a \arrow[u, two heads, "\leq"]
       & \\

       & & & (\lambda x \myleq A. x) \, a \arrow[u, "\leq"]
       \arrow[uuuull, dashed, "\leq_H", start anchor=west, bend left=20] & \\

       & & & (\SingleExpr)\,A \, a \arrow[u, "\leq"] \arrow[r , two heads, "\equiv"]
       &
       a \arrow[uuuuu, two heads, dashed, "\leq"]
       \\

       & & & & (\lambda x \myleq A. x) \, a \arrow[u, "\leq"] \\
     \end{tikzcd}
  \end{center}
  \caption{Derivation of $\Gamma; \nil \vdash (\IdExprA) \, a \leq \Unit \, A \, a$}
  \label{fig:example:algorithmic-derivation}
\end{figure}

Theorem~\ref{lem:algorithmic-transitivity-elimination} ensures that the transitivity step
collecting the two derivations $\Gamma; \nil \vdash \IdA\,a \leq \Single\,A\,a$ and
$\Gamma; \nil \vdash \Single\,A\,a \leq^* \Unit\,A\,a$ can be removed, obtaining a
transitivity-free derivation
\begin{displaymath}
  \Gamma; \nil \vdash \IdA\,a \leq \Unit\,A\,a
\end{displaymath}

Our constructive proof produces a transitivity-free derivation of
$\Gamma; \nil \vdash \Single\,A\,a \leq^* \Unit\,A\,a$ (\ie, flattens it into
$\Gamma; \nil \vdash \Single\,A\,a \leq \Unit\,A\,a$ witnessed by the straight vertical
and horizontal paths in the middle of the diagram), and then applies the commutativity
theorem to complete the upper-right corner of the diagram, obtaining the transitivity-free
derivation $\Gamma; \nil \vdash (\IdExprA) \, a \leq \Unit\,A\,a$. This proof is possible
in our algorithmic system thanks to the operand stack and Rule \textsc{Ms-FOp} which
exposes the term $(\lambda x\myleq A.a) \, a$, and to the fact that the derivation
$\Gamma; \nil \vdash (\lambda x\myleq A.a) \, a\leq^* (\lambda x\myleq A.A) \, a$ in the
middle of the diagram preserves the possibly convoluted derivation
$\Gamma, x \myequiv a; \nil \vdash a\leq^* A$ as a sub-derivation (marked with
$\ddots$). Although preserving this convoluted derivation seems at first glance
counter-intuitive, it gives the opportunity to apply the induction hypothesis of the
Theorem~\ref{lem:algorithmic-transitivity-elimination} to the term
$(\lambda x\myleq A.a) \, a$, which will flatten the convoluted derivation. Contrary to
ours, Hutchins' attempts at proving that transitivity is admissible reduce the complexity
of the intermediate step by avoiding the convoluted derivation upfront (by directly taking
the reduction $(\lambda x\myleq A.x) \, a \myrel{\leq} (\lambda x\myleq A.A) \, a$ which
promotes the formal parameter $a$ of the redex to its annotation $A$).

But this premature promotion of a formal parameter to its type annotation is in fact
counter-productive, which manifests in the failure of Hutchins' proof of the commutativity
theorem, where the premature promotion prevents one from completing the diagram for
redexes with a transitivity-free right edge.

This result was the missing cornerstone in Hutchins' original theory, and its proof here
is necessary for establishing type safety. In the next section, we present what can be a
type-safe theory based on our MPSS.

\section{Towards a Type-Safe System}
\label{sec:type-safety}

Thanks to the above commutative system, we can now sketch the design of a type-safe
system.  The static semantics for this system is defined by a \emph{well-formedness}
judgment, presented in Figure~\ref{fig:well-formed-derivations-rules}, which is the static
condition for type safety. Rules~\textsc{Wf-PrS} and \textsc{Wf-PrE} ensure that a
variable is well-formed if and only if it occurs in a prevalid context (respectively, in a
subtype or an equivalence annotation). Rule~\textsc{Wf-Top} ensures that $\T$ is
universally well-formed. Rule~\textsc{Wf-Fun} ensures that an abstraction is well-formed
if and only if it has a well-formed annotation, and its body is well-formed in a context
enlarged with the abstraction annotation. Finally, Rule~\textsc{Wf-App} ensures that an
application is well-formed if and only if it has an operator that is a well-subtype of a
function, and an operand that is a well-subtype of the operator's subtype annotation.

Well-formedness depends on the \emph{well-subtyping} relation $\leq_\wef$, which in turn
subsumes \emph{well-equivalence} $\equiv_\wef$.  The rules for well-subtyping connect the
static system to our dynamic reduction semantics. For example, Rule~\textsc{Ws-Lf2} states
that $v$ is a well-subtype of $t$ if $v$ promotes to some $v'$ that is itself a
well-subtype of $t$.

\begin{figure}[ht!]
  \textsf{\textbf{Term well-formedness}}\hfill\boxed{\text{$\Gamma \vdash t~\wef$}}
  \begin{mathpar}
    \inferrule*[right=Wf-PrS]
    {\Gamma ~\prevalid \and x \myleq t \in \Gamma}
    {\Gamma \vdash x~\wef}
    \and
    \inferrule*[right=Wf-PrE]
    {\Gamma ~\prevalid \and x \myequiv \alpha \in \Gamma}
    {\Gamma \vdash x~\wef}
    \and
    \inferrule*[right=Wf-Top] {\Gamma ~\prevalid} {\Gamma \vdash \T~\wef}
    \and
    \inferrule*[right=Wf-Fun]
    {\Gamma, x\myleq t \vdash u~\wef \and \Gamma \vdash t ~\wef}
    {\Gamma \vdash \lambda x\myleq t.u~\wef}
    \and
    \inferrule*[right=Wf-App]
    {\Gamma \vdash u \leq^*_\wef \lambda x\myleq t.\T \and
      \Gamma \vdash v \leq^*_\wef t}
    {\Gamma \vdash u\,v~\wef}
  \end{mathpar}

  \textsf{\textbf{Transitive well-subtyping}}\hfill
  \boxed{\text{$\Gamma \vdash v \myvartriangleleft_{\wef}^* t$}}
  \begin{mathpar}
    \inferrule*[right=Ws-Sub]
    {\Gamma \vdash v ~\wef \and
      \Gamma \vdash v \myvartriangleleft_{\wef} t \and
      \Gamma \vdash t ~\wef}
    {\Gamma \vdash v \myvartriangleleft_{\wef}^* t}

    \inferrule*[right=Ws-Trs]
    {\Gamma \vdash v \myvartriangleleft_{\wef}^* u \and
      \Gamma \vdash u ~\wef \and
      \Gamma \vdash u \myvartriangleleft_{\wef}^* t}
    {\Gamma \vdash v \myvartriangleleft_{\wef}^* t}
  \end{mathpar}

  \textsf{\textbf{Well-subtyping}}\hfill\boxed{\text{$\Gamma \vdash u
      \myvartriangleleft_\wef t$}}
  \begin{mathpar}
    \inferrule*[right=Ws-Rfl]
    {\Gamma ~\prevalid}
    {\Gamma \vdash t \myvartriangleleft_{\wef} t}

    \inferrule*[right=Ws-Lf1]
    {\Gamma; \nil \vdash v \myrel{\equiv} v' \and
    \Gamma \vdash v' \myvartriangleleft_\wef t}
    {\Gamma \vdash v \myvartriangleleft_\wef t}

    \inferrule*[right=Ws-Lf2]
    {\Gamma \vdash v ~\wef \and
    \Gamma; \nil \vdash v \myrel{\leq} v' \and
    \Gamma \vdash v' ~\wef \and
    \Gamma \vdash v' \leq_\wef t}
    {\Gamma \vdash v \leq_\wef t}

    \inferrule*[right=Ws-Rgh]
    {\Gamma \vdash v \myvartriangleleft_\wef t' \and
     \Gamma; \nil \vdash t \myrel{\equiv} t'}
    {\Gamma \vdash v \myvartriangleleft_{\wef} t}
  \end{mathpar}

  \caption{Well-subtyping in MPSS}
  \label{fig:well-formed-derivations-rules}
\end{figure}

We define type safety as \emph{progress} (well-formed terms either are in normal form or
can be reduced), and \emph{preservation} (subtyping relations between terms are preserved
under reduction by the operational semantics below). We define the operational semantics
as a context-based reduction semantics \cite{Fel87}: the evaluation $\mapsto$ is defined
as the congruence closure (Rule~\textsc{Os-Con} and evaluation contexts $\Co$) of
$\beta$-reduction (Rule~\textsc{Os-Bet}).

\noindent
\textsf{\textbf{Operational Semantics}}\hfill\boxed{\text{$t\mapsto t'$}}
\begin{mathpar}
  \inferrule*[right=Os-Con]
  {t \mapsto t'}
  {\Co[t] \mapsto \Co[t']}

  \inferrule*[right=Os-Bet]
  { }
  {(\lambda x\myleq t.u)v \mapsto \cas{u}{x}{v}}

  \Co ::= \square \mid (\lambda x\myleq \Co.t) \mid (\lambda x\myleq t.\Co) \mid (\Co\,t) \mid (t\,\Co)
\end{mathpar}

Well-formedness depends on the \emph{well-subtyping} relation defined in
Figure~\ref{fig:well-formed-derivations-rules}. Well-subtyping ($\leq_\wef$) subsumes
\emph{well-equivalence} ($\equiv_\wef$). We use the meta-variable $\myvartriangleleft$ and
define ($\myvartriangleleft_\wef$), which captures both relations, and similarly for their
transitive closure ($\myvartriangleleft^*_\wef$). By a slight abuse of language, and to
simplify the text, we write ``well-subtyping'' to refer to the meta-notation
($\vartriangleleft_\wef$) in the paragraph that follows.

Figure~\ref{fig:well-formed-derivations-rules} introduces \emph{transitive
  well-subtyping}, the straightforward transitive closure of the well-subtyping relation
(Rules~\textsc{Ws-Sub} and \textsc{Ws-Trs}). Next in the figure is \emph{well-subtyping},
which is defined similarly to the diagrammatic subtyping relation $\leq$ defined in
Section~\ref{sec:machine-based-system}, but with extra requirements on the well-formedness
of the terms involved.  Rule~\textsc{Ws-Rfl} states that well-subtyping is reflexive.
Rule~\textsc{Ws-Lf1} ensures that $v$ is a well-subtype of $t$ if and only if $v'$ is a
well-subtype of $t$ and $v$ reduces to $v'$. Rule~\textsc{Ws-Lf2} ensures that $v$ is a
well-subtype of $t$ if and only if $v'$ is a well-subtype of $t$ and $v$ promotes to $v'$,
but with the extra requirements that both $v$ and $v'$ are well-formed, whose role is
explained in a minute.  Rule~\textsc{Ws-Rgh} ensures that $v$ is a well-subtype of $t$ if
and only if $v$ is a well-subtype of $t'$ and $t$ reduces to $t'$.

Thanks to these definitions, we may define type safety in our calculus as the following
progress and preservation theorems.
\begin{theorem}[Progress]
  \label{thm:progress}
  Let $t$ be a term. For every logical context $\Gamma$, if $\Gamma \vdash t~\wef$ then
  either $t$ is in normal form, or there exists a term $t'$ such that $t \mapsto t'$.
\end{theorem}
\begin{theorem}[Preservation]
  \label{thm:preservation}
  Let $\Gamma$ be a logical context. Let $t$, $t'$, and $u$ be terms. If
  $\Gamma \vdash t \leq^*_\wef u$ and $t \mapsto t'$, then
  $\Gamma \vdash t' \leq^*_\wef u$.
\end{theorem}

The proof of type preservation (Theorem~\ref{thm:preservation}) rests on showing that
reduction preserves term well-formedness. Lemma~\ref{lem:mapsto-preserves-wf} establishes
this fact.

\begin{lemma}[Evaluation preserves well-formedness]
  \label{lem:mapsto-preserves-wf}
  Let $\Gamma$ be a logical context. Let $t$ and $t'$ be terms such that $t \mapsto t'$
  and such that $\Gamma \vdash t ~\wef$. We have $\Gamma \vdash t' ~\wef$.
\end{lemma}

A crucial case in the proof of Lemma~\ref{lem:mapsto-preserves-wf} is $\beta$-reduction,
where a well-formed application $(\lambda x\myleq t.u) \, \alpha$ reduces to
$\cas{u}{x}{\alpha}$. To handle this case, we need to show that well-formedness is
preserved by substitution. The following lemma provides this guarantee, where the premise
$\Gamma \vdash \alpha \leq^*_{\wef} t$ comes from the well-formedness of the original
term.

\begin{lemma}[Substitution preserves well-formedness]
  \label{lem:substitution-preserves-wf}
  Let $\Gamma$ and $\Gamma'$ be logical contexts, and let $t$, $u$ and $\alpha$ be terms.
  Let $x$ be a variable. If $\Gamma,x\myleq t,\Gamma' \vdash u~\wef$, and
  $\Gamma \vdash \alpha \leq^*_{\wef} t$, then
  $\Gamma,\cas{\Gamma'}{x}{\alpha} \vdash\cas{u}{x}{\alpha}~\wef$.
\end{lemma}

Our subtyping relation in MPSS is context-dependent, unlike in PSS where it isn't, which
implies that Lemma~\ref{lem:substitution-preserves-wf} relies on the conjecture collected
below, which states that the well-subtyping relation is independent of the surrounding
program structure, on certain contexts named \emph{covariant contexts} defined as follows:
\begin{mathpar}
  \Po ::= \square \mid (\lambda x\myleq t.\Po) \mid (\Po\,t)
\end{mathpar}

\begin{conjecture}[Well-subtyping is context independent]
  \label{conj:stack-elimination}
  Let $\Gamma$ be a logical context and $u$ and $t$ be terms such that
  $\Gamma \vdash u \leq^*_\wef t$.  Let $\Po$ be a covariant context such that both
  $\Po[u]$ and $\Po[t]$ are well-formed in $\Gamma$. We conjecture that
  $\Gamma \vdash \Po[u] \leq^*_\wef \Po[t]$.
\end{conjecture}

We believe that the conjecture holds with the additional well-formedness requirements that
both $v$ and $v'$ are well-formed in Rule \textsc{Ws-Lf2} above.

All in all, type safety, which is the combination of progress (Theorem~\ref{thm:progress})
and preservation (Theorem~\ref{thm:preservation}), holds under the assumption that
Conjecture~\ref{conj:stack-elimination} holds. The proof that both these theorems hold
under the assumption that Conjecture~\ref{conj:stack-elimination} holds can be found in
the appendix.

Establishing this conjecture is the final step towards a complete proof of type safety.
The sketch of a complete type system provided in this section demonstrates a path forward,
built upon the solid foundation of our commutativity result. Future work could either
adopt this sketch and complete the proof of type safety, or develop a different type
system that leverages the commutativity result established in
Section~\ref{sec:commutation} to establish type safety in a different way.

\section{Related work}
\label{sec:related}

PSS naturally subsumes dependent types and higher-order subtyping, so we comment on both
ideas.

There are many works related to dependent types. For instance in Cayenne
\cite{augustsson1998cayenne}, Idris \cite{brady2013idris} or Haskell
\cite{marlow2010haskell}, dependent types are an extension of the language, in contrast to
our case where dependent types are built-in as there is no distinction between terms and
types.

Other languages mainly use dependent types for theorem proving, for instance Coq
\cite{coquand1986calculus} or Agda \cite{agda}. Our language is not strongly normalising,
and our main goal is not theorem proving as we cannot use the language as a proof system,
since applying the Curry-Howard correspondence to PSS in a consistent way is still an open
question.

To accommodate the Curry-Howard correspondence, one possible approach is to distinguish
between potentially non-terminating terms and terminating terms, either with monads or via
other means \cite{moggi1991notions, pfenning2001judgmental}. $\lambda^\theta$ is a
calculus that explores this idea
\cite{casinghino2014combiningthese,casinghino2014combining,sjoberg2015dependently}. It
combines proofs and programs into a single language that distinguishes two fragments, a
logical fragment where every term is strongly normalising that is suited to write proofs,
and a programmatic fragment that is Turing complete and lacks termination. Such systems
allow so-called \emph{freedom of speech}, which is the ability to write (terminating)
proofs about non-terminating terms.

Subtyping extensions of type systems have been considered. In 1997, Chen proposed
$\lambda_{C_\leq}$ \cite{chen1997subtyping}, an extension of the calculus of constructions
$\lambda_C$ with subtyping. However, Chen's $\lambda_{C_\leq}$ supports neither top types
nor bounded quantification as we do in PSS.

Various extensions of System F, in order to equip it with subtyping, have also been
considered in the literature. Within these we find System F$_\leq$
\cite{cardelli1994extension} and System F$_\leq^\omega$ \cite{pierce1997higher}, which
respectively add subtyping and higher-order subtyping capabilities. PSS subsumes the
version of System F$_\leq$ without contravariant arrow-types \cite{hutchinsthese}.

The origin of mixing terms and types goes back to the introduction of pure type systems
\cite{mckinna1993pure}, which generalises the seminal works on lambda calculi with types
by \cite{Bar92}. These systems have the ability to unify types and terms under the same
kind, but there is a typing relation instead of a subtyping one. Moreover, in some of
these pure type systems, the type $\bot$ is inhabited. This gives rise to the well-known
Girard's paradox (see \cite{hurkens1995simplification}). In our system, there is not a
primitive $\bot$ subtype, but we could encode such a subtype to be inhabited by itself,
and that would be unproblematic since our aim is not to use PSS as a theorem prover.

Systems where terms can be used in certain ways as types have been considered.  Aspinall
studied the combination of subtyping with singleton types
\cite{aspinall1994subtyping}. Singleton types are close to what we study, as in our
calculus, terms can be seen as singleton types of this unique value. We have for instance,
for every term $t$, $\varepsilon \vdash t \leq t$, which would be translated to
$\vdash t : \{t\}$ in a system with singleton types.

Other attempts to unify typing and subtyping exist. A more conservative attempt that tries
to unify typing and subtyping is \cite{yang2017unifying}. This work has two major
differences with our approach. First, in \cite{yang2017unifying} the subtyping relation
still tracks types (their unification of subtyping and typing is more conservative than
ours). Second, in PSS functions and function (sub)types are defined with the same syntax,
and therefore they are both instances of the same construct, which is not the case in
\cite{yang2017unifying}.

Because many interesting features described in \cite{hutchinsthese} are related to
modules, other type systems especially related to objects have been studied. Our work
shares foundational goals with the Dependent Object Types (DOT) calculus, the theoretical
basis for Scala 3 \cite{amin2016essence, rompf2016type}. Both PSS and DOT aim to create
powerful, unified type systems that support advanced modularity, and provide a formal
basis for solving challenges akin to the expression problem, the tag-elimination problem,
and others. DOT achieves this, among other features, through path-dependent types, where
the type of a member, $p.T$, depends on the term $p$. Because DOT allows objects to
contain types as members, it provides a form of dependent typing. PSS takes a more radical
approach by collapsing types and terms into a single syntactic category, where the
traditional typing relation ($:$) is completely replaced by a universal subtyping relation
($\leq$).  This design is what enables PSS's key features described in Sections 3 and 4 of
\cite{hutchinsthese}.

\section{Future work}
\label{sec:future-work}

As described in Section~\ref{sec:type-safety}, it is left as future work to define a type
system that is built upon MPSS and to prove its type safety. Another interesting area to
explore is the link between our system and Hutchins' system, to determine if both
well-formed notions are equivalent.

Once this system is defined, it is also left as future work to obtain an algorithm that
typechecks terms in MPSS. Hutchins describes an algorithm for his system in
\cite{hutchinsthese}, but one should adapt his algorithm to this new system, and study
upon which conditions it terminates. One possible approach is to use the operational
relevance of terms, for which there exist some arguments that lead them to terminate.  We
leave as future work to explore other practical algorithms which terminate on
operationally relevant terms.

Another venue of future work is to study whether PSS and MPSS are Turing-complete. To the
best of our knowledge, whether PSS is Turing-complete has not been formally proven yet.
The argument in favor of Turing-completeness rests on the existence of a (well-formed)
looping combinator in PSS, which would consist of a family of terms $L_n$ such that
$L_n \, f$ $\beta$-reduces to $f(L_{n+1} \, f)$ for any $f$. Such a looping combinator
exists for System~$\lambda *$ \cite{Howe87}, which can be encoded into Hutchins' PSS
thanks to the encoding of System~$\lambda *$ in Hutchins' PSS, provided on page~47 of
\cite{hutchinsthese}. In order to establish Turing-completeness of PSS, one should now
adapt the argument made by \cite{geuvers2009fixed} in Section 2 or the argument made in
\cite{HS08}.  There is, however, no simple known expression of a looping combinator in
System~$\lambda *$. The looping combinator from \cite{Howe87} can only be derived
mechanically, and working with it directly is proven to be hard as its length is more than
40 pages (see the discussion on the logical consistency and impredicativity of
System~$\lambdaHut$ on pages 48--51 of \cite{hutchinsthese}). It may therefore be
difficult to conclude about Turing-completeness of either PSS or MPSS, which is left as
future work.

\section{Conclusion}
\label{sec:conclusions}

This paper presents a variant of the PSS originally introduced by Hutchins in
\cite{hutchinsthese}, which we name MPSS, and proves that our reformulation enjoys a key
commutation property which was missing in Hutchins' system. We also sketch how a type
system can be derived from MPSS, which could lead to a type-safe system. To the best of
our knowledge PSS and MPSS are the first theories that unify harmoniously terms and types
by using a single kind where everything is a subtype.

Meta-theories of languages with advanced features on types are notoriously difficult to
work with. We believe our work could help in laying the foundations for a new approach
that mixes terms and types in innovative ways.

\bibliography{main.bib}

\newpage

\appendix

\section{Theorems of the text}

\textsc{Theorem~\ref{thm:commutativity-theorem} [$\myrel{\leq}$ and $\myrel{\equiv}$
  strongly commute]}
  Let $\Gamma; s$ be an extended context.
  Let $t_0$, $t_1$ and $t_2$ be terms.
  If $\Gamma;s\vdash t_0\myrel{\equiv} t_1$ and $\Gamma;s\vdash t_0\myrel{\leq} t_2$
  then for all extended context $\Gamma'; s'$ with $\Gamma; s\rightarrowtail \Gamma'; s'$,
  there exists a term $t_3$,
  such that $\Gamma; s \vdash t_2\myrel{\equiv} t_3$ and $\Gamma';s' \vdash t_1\myrel{\leq} t_3$.
  
  Diagrammatically, where the solid arrows are the conditions for the existence of the dashed arrows.
  %\begin{center}
    \begin{displaymath}
      \begin{tikzcd}
        t_2 \arrow[r, dashed, "\equiv"]   & t_3                       \\
        t_0 \arrow[r, "\equiv"] \arrow[u, "\leq"] & t_1 \arrow[u, dashed, "\leq"]
      \end{tikzcd}
    \end{displaymath}
  %\end{center}
\begin{proof}
  By induction on the term structure of $t_0$.
  
  We consider cases according to the last
  two rules that are applied respectively in the derivation of the horizontal and
  vertical edges of the diagram above, starting by the bottom-left corner.

  The cases where bottom rule is rule \textsc{Me-Top} or respectively rule \textsc{Me-TAp},
  then the missing top edge is by rule \textsc{Me-Top} or respectively \textsc{Me-TAp},
  and the missing right edge is by rule \textsc{Ms-Top}.
  
  If the vertical edge rule is \textsc{Ms-Top},
  then the missing top edge is by rule \textsc{Me-Top},
  and the missing right edge is by rule \textsc{Ms-Top}.

  If the vertical edge is \textsc{Ms-Equ}, then we have
  $\Gamma;s\vdash t_0\myrel{\equiv} t_1$ and $\Gamma;s\vdash t_0\myrel{\equiv} t_2$.
  By reflexivity of $\myrel{\equiv}$ (Proposition~\ref{prop:algorithmic-refl}), we have
  $\Gamma; s \rightarrowtail \Gamma; s$.
  Hence by Lemma~\ref{lem:myrelequiv-has-diamond-property} we have
  $\Gamma; s \vdash t_2\myrel{\equiv} t_3$ and $\Gamma';s' \vdash t_1\myrel{\equiv} t_3$.
  Thanks to rule \textsc{Ms-Equ}, from the latter derivation, we obtain $\Gamma';s' \vdash t_1\myrel{\leq} t_3$.

  We know that if the rule \textsc{Me-Pro} appears in the derivation tree $\Gamma;s\vdash t_0\myrel{\leq} t_2$,
  then we have in fact $\Gamma;s\vdash t_0\myrel{\equiv} t_2$, and as a result the results follows from
  Lemma~\ref{lem:myrelequiv-has-diamond-property} as outlined by the paragraph outlined above.
  We therefore assume that the rule \textsc{Me-Pro} doesn't appear in the derivation tree $\Gamma;s\vdash t_0\myrel{\leq} t_2$,
  an information useful only for the last case of this lemma.

  Else, for the general case we distinguish the sub-cases below:
  \begin{description}
    \item [Case \textsc{Me-Pro} with \textsc{Ms-Pro}:] We need to prove the diagram below:
  
    \begin{center}
      \begin{tikzcd}
        t \arrow[r, dashed, "\equiv"]           & . \\
        x \arrow[u, "\leq"] \arrow[r, "\equiv"] & \alpha \arrow[u, dashed, "\leq"]
      \end{tikzcd}
    \end{center}

    By assumption, we have the following reductions:
    \begin{itemize}
      \item $\Gamma; s \vdash x \myrel{\equiv} \alpha'$ with $x \myequiv \alpha \in \Gamma$
      \item $\Gamma; s \vdash x \myrel{\leq} t_2$ with $x \myleq t \in \Gamma$
    \end{itemize}
    By prevalidity of the context, we cannot have two distinct annotation for $x$ in the context,
    and therefore this case is impossible: there cannot be a subtype annotation and an equivalence
    annotation in the prevalid context $\Gamma$.
  
    \item [Case \textsc{Me-Var} with \textsc{Ms-Pro}:] We need to prove the diagram below:
    
    \begin{center}
      \begin{tikzcd}
      t \arrow[r, dashed, "\equiv"]             & . \\
      x \arrow[u, "\leq"] \arrow[r, "\equiv"]   & x \arrow[u, dashed, "\leq"]
      \end{tikzcd}
    \end{center}
    
    By assumption, we have the following reductions:
    \begin{itemize}
      \item $\Gamma; s \vdash x \myrel{\equiv} x$.
      \item $\Gamma; s \vdash x \myrel{\leq} t$ by rule \textsc{Ms-Pro}.
    \end{itemize}
    
    From rule \textsc{Ms-Pro}, we know there exists an annotation $x \myleq t \in \Gamma$.
    Let $\Gamma = \Gamma_0, x \myleq t, \Gamma_1$ for some context parts $\Gamma_0$ and $\Gamma_1$.

    Let $\Gamma'; s'$ be an extended context such that $\Gamma; s \rightarrowtail \Gamma'; s'$.
    
    By Lemma~\ref{lem:commutativity-context-weakening}, we obtain $\Gamma; \nil \rightarrowtail \Gamma'; \nil$
    from $\Gamma; s \rightarrowtail \Gamma'; s'$.
    Now, by multiple use of the rule \textsc{Ct-Ann}
    we have $\Gamma' = \Gamma'_0, x \myleq t', \Gamma'_1$ with $\Gamma_0; \nil \vdash t \myrel{\equiv} t'$.
    By weakening (Lemma~\ref{lem:context-weakening}),
    we have $\Gamma_0, x \myleq t, \Gamma_1; s \vdash t \myrel{\equiv} t'$ from $\Gamma_0; \nil \vdash t \myrel{\equiv} t'$.
    
    By rule \textsc{Ms-Pro}, we have $\Gamma'; s' \vdash x \myrel{\leq} t'$ because $x \myleq t' \in \Gamma'$.
    
    The filled diagram is therefore the following:
    \begin{center}
      \begin{tikzcd}
        t \arrow[r, "\equiv"]           & t' \\
        x \arrow[u, "\leq"] \arrow[r, "\equiv"]   & x \arrow[u, "\leq"]
      \end{tikzcd}
    \end{center}
  
    \item [Case \textsc{Me-App} with \textsc{Ms-App}:] We need to prove the diagram below:
    
    \begin{center}
      \begin{tikzcd}
      u_2 \, v_0 \arrow[r, dashed, "\equiv"]           & . \\
      u_0 \, v_0 \arrow[u, "\leq"] \arrow[r, "\equiv"] & u_1 \, v_1 \arrow[u, dashed, "\leq"]
      \end{tikzcd}
    \end{center}
    
    By assumption, we have the following reductions:
    \begin{itemize}
      \item $\Gamma; s \vdash u_0 \, v_0 \myrel{\equiv} u_1 \, v_1$ by \textsc{Me-App}, which gives us:
      \begin{itemize}
      \item $\Gamma; v_0 :: s \vdash u_0 \myrel{\equiv} u_1$
      \item $\Gamma; \nil \vdash v_0 \myrel{\equiv} v_1$
      \end{itemize}
      \item $\Gamma; s \vdash u_0 \, v_0 \myrel{\leq} u_2 \, v_0$ by \textsc{Ms-App}, which gives us:
      \begin{itemize}
      \item $\Gamma; v_0 :: s \vdash u_0 \myrel{\leq} u_2$
      \end{itemize}
    \end{itemize}

    Let $\Gamma'; s'$ be an extended context such that $\Gamma; s \rightarrowtail \Gamma'; s'$.
    By rule \textsc{Ct-Stk}, we have
    $\Gamma; v_0 :: s \rightarrowtail \Gamma'; v_1 :: s'$ from $\Gamma; \nil \vdash v_0 \myrel{\equiv} v_1$.
    
    By induction hypothesis on $u_0$, there exists a term $u_3$ such that:
    \begin{itemize}
      \item $\Gamma; v_0 :: s \vdash u_2 \myrel{\equiv} u_3$
      \item $\Gamma'; v_1 :: s' \vdash u_1 \myrel{\leq} u_3$
    \end{itemize}
    
    We now have
    \begin{itemize}
      \item $\Gamma'; s \vdash u_2 \, v_0 \myrel{\equiv} u_3 \, v_1$ by \textsc{Me-App}
      from $\Gamma; v_0 :: s \vdash u_2 \myrel{\equiv} u_3$ and $\Gamma; \nil \vdash v_0 \myrel{\equiv} v_1$
      \item $\Gamma'; s' \vdash u_1 \, v_1 \myrel{\leq} u_3 \, v_1$ by \textsc{Ms-App}
      from $\Gamma'; v_1 :: s' \vdash u_1 \myrel{\leq} u_3$
    \end{itemize}
    
    The filled diagram is therefore the following:
    \begin{center}
      \begin{tikzcd}
      u_2 \, v_0 \arrow[r, "\equiv"] & u_3 \, v_1 \\
      u_0 \, v_0 \arrow[u, "\leq"] \arrow[r, "\equiv"] & u_1 \, v_1 \arrow[u, "\leq"]
      \end{tikzcd}
    \end{center}
    
    \item [Case \textsc{Me-Fun} with \textsc{Ms-Fun}:] We need to prove the diagram below:
    
    \begin{center}
      \begin{tikzcd}
      \lambda x \myleq t_0 . u_2 \arrow[r, dashed, "\equiv"]           & . \\
      \lambda x \myleq t_0 . u_0 \arrow[u, "\leq"] \arrow[r, "\equiv"] & \lambda x \myleq t_1 . u_1 \arrow[u, dashed, "\leq"]
      \end{tikzcd}
    \end{center}
    
    By assumption, we have the following reductions:
    \begin{itemize}
      \item $\Gamma; \nil \vdash \lambda x \myleq t_0 . u_0 \myrel{\equiv} \lambda x \myleq t_1 . u_1$ by \textsc{Me-Fun}, which gives us:
      \begin{itemize}
      \item $\Gamma; \nil \vdash t_0 \myrel{\equiv} t_1$
      \item $\Gamma, x \myleq t_0; \nil \vdash u_0 \myrel{\equiv} u_1$
      \end{itemize}
      \item $\Gamma; \nil \vdash \lambda x \myleq t_0 . u_0 \myrel{\leq} \lambda x \myleq t_0 . u_2$ by \textsc{Ms-Fun}, which gives us:
      \begin{itemize}
      \item $\Gamma, x \myleq t_0; \nil \vdash u_0 \myrel{\leq} u_2$
      \end{itemize}
    \end{itemize}
    
    Let $\Gamma'; s'$ be an extended context such that $\Gamma; s \rightarrowtail \Gamma'; s'$.
    Because in our case $s = \nil$, we have $s' = \nil$.
    By rule \textsc{Ct-Ann}, we have $\Gamma, x \myleq t_0; \nil \rightarrowtail \Gamma', x \myleq t_1; \nil$
    from $\Gamma; \nil \vdash t_0 \myrel{\equiv} t_1$.
    
    By induction hypothesis on $u_0$, there exists a term $u_3$ such that:
    \begin{itemize}
      \item $\Gamma, x \myleq t_0; \nil \vdash u_2 \myrel{\equiv} u_3$
      \item $\Gamma', x \myleq t_1; \nil \vdash u_1 \myrel{\leq} u_3$
    \end{itemize}

    We now have
    \begin{itemize}
      \item $\Gamma; \nil \vdash \lambda x \myleq t_0 . u_2 \myrel{\equiv} \lambda x \myleq t_1 . u_3$ by \textsc{Me-Fun}
      from $\Gamma; \nil \vdash t_0 \myrel{\equiv} t_1$ and $\Gamma, x \myleq t_0; \nil \vdash u_2 \myrel{\equiv} u_3$
      \item $\Gamma'; \nil \vdash \lambda x \myleq t_1 . u_1 \myrel{\leq} \lambda x \myleq t_1 . u_3$ by \textsc{Ms-Fun}
      from $\Gamma', x \myleq t_1; \nil \vdash u_1 \myrel{\leq} u_3$
    \end{itemize}
    
    The filled diagram is therefore the following:
    \begin{center}
      \begin{tikzcd}
      \lambda x \myleq t_0 . u_2 \arrow[r, "\equiv"] & \lambda x \myleq t_1 . u_3 \\
      \lambda x \myleq t_0 . u_0 \arrow[u, "\leq"] \arrow[r, "\equiv"] & \lambda x \myleq t_1 . u_1 \arrow[u, "\leq"]
      \end{tikzcd}
    \end{center}
  
    \item [Case \textsc{Me-FOp} with \textsc{Ms-FOp}:] We need to prove the diagram below:

    \begin{center}
      \begin{tikzcd}
      \lambda x \myleq t_0 . u_2 \arrow[r, dashed, "\equiv"]           & . \\
      \lambda x \myleq t_0 . u_0 \arrow[u, "\leq"] \arrow[r, "\equiv"] & \lambda x \myleq t_1 . u_1 \arrow[u, dashed, "\leq"]
      \end{tikzcd}
    \end{center}
    
    By assumption, we have the following reductions:
    \begin{itemize}
      \item $\Gamma; \alpha :: s \vdash \lambda x \myleq t_0 . u_0 \myrel{\equiv} \lambda x \myleq t_1 . u_1$ by \textsc{Me-FOp}, which gives us:
      \begin{itemize}
      \item $\Gamma; \nil \vdash t_0 \myrel{\equiv} t_1$
      \item $\Gamma, x \myequiv \alpha; s \vdash u_0 \myrel{\equiv} u_1$
      \end{itemize}
      \item $\Gamma; \alpha :: s \vdash \lambda x \myleq t_0 . u_0 \myrel{\leq} \lambda x \myleq t_0 . u_2$ by \textsc{Ms-FOp}, which gives us:
      \begin{itemize}
      \item $\Gamma, x \myequiv \alpha; s \vdash u_0 \myrel{\leq} u_2$
      \end{itemize}
    \end{itemize}
    
    Let $\Gamma'; s'$ be an extended context such that $\Gamma; s \rightarrowtail \Gamma'; s'$.
    In our case, we have $s = \alpha :: s_0$ and $s' = \alpha' :: s_1$ with by rule \textsc{Ct-Stk}, $\Gamma; \nil \vdash \alpha \myrel{\equiv} \alpha'$.
    By rule \textsc{Ct-Stk}, we have $\Gamma; s_0 \rightarrowtail \Gamma'; s_1$
    and $\Gamma; \nil \vdash \alpha \myrel{\equiv} \alpha'$
    from $\Gamma; s \rightarrowtail \Gamma'; s'$.
    We now have $\Gamma, x \myequiv \alpha; s_0 \rightarrowtail \Gamma', x \myequiv \alpha'; s_1$ by rule \textsc{Ct-Ann}.
    
    By induction hypothesis on $u_0$, there exists a term $u_3$ such that:
    \begin{itemize}
      \item $\Gamma, x \myequiv \alpha; s_0 \vdash u_2 \myrel{\equiv} u_3$
      \item $\Gamma', x \myequiv \alpha'; s_1 \vdash u_1 \myrel{\leq} u_3$
    \end{itemize}
    
    We now have
    \begin{itemize}
      \item $\Gamma; \alpha :: s_0 \vdash \lambda x \myleq t_0 . u_2 \myrel{\equiv} \lambda x \myleq t_1 . u_3$ by \textsc{Me-FOp}
      from $\Gamma; \nil \vdash t_0 \myrel{\equiv} t_1$ and $\Gamma, x \myequiv \alpha; s_0 \vdash u_2 \myrel{\equiv} u_3$
      \item $\Gamma'; \alpha' :: s_1 \vdash \lambda x \myleq t_1 . u_1 \myrel{\leq} \lambda x \myleq t_1 . u_3$ by \textsc{Ms-FOp}
      from $\Gamma', x \myequiv \alpha'; s_1 \vdash u_1 \myrel{\leq} u_3$
    \end{itemize}
    
    The filled diagram is therefore the following:
    \begin{center}
      \begin{tikzcd}
      \lambda x \myleq t_0 . u_2 \arrow[r, "\equiv"] & \lambda x \myleq t_1 . u_3 \\
      \lambda x \myleq t_0 . u_0 \arrow[u, "\leq"] \arrow[r, "\equiv"] & \lambda x \myleq t_1 . u_1 \arrow[u, "\leq"]
      \end{tikzcd}
    \end{center}
    
    \item [Case \textsc{Me-Bet} with \textsc{Ms-App}:] We need to prove the diagram below:

    \begin{center}
      \begin{tikzcd}
      (\lambda x \myleq t . u_2) \, v_0 \arrow[r, dashed, "\equiv"]           & . \\
      (\lambda x \myleq t . u_0) \, v_0 \arrow[u, "\leq"] \arrow[r, "\equiv"] & \cas{u_1}{x}{v_1} \arrow[u, dashed, "\leq"]
      \end{tikzcd}
    \end{center}
    
    By assumption, we have the following reductions:
    \begin{itemize}
      \item $\Gamma; s \vdash (\lambda x \myleq t . u_0) \, v_0 \myrel{\equiv} \cas{u_1}{x}{v_1}$ by \textsc{Me-Bet}, which gives us:
      \begin{itemize}
        \item $\Gamma; s \vdash u_0 \myrel{\equiv} u_1$
        \item $\Gamma; \nil \vdash v_0 \myrel{\equiv} v_1$
      \end{itemize}
      \item $\Gamma; s \vdash (\lambda x \myleq t . u_0) \myrel{\leq} (\lambda x \myleq t . u_2) \, v_0$ by \textsc{Ms-App}, which gives us:
      \begin{itemize}
        \item $\Gamma; v_0 :: s \vdash \lambda x \myleq t . u_0 \myrel{\leq} \lambda x \myleq t . u_2$
        \item By \textsc{Ms-FOp}, this implies $\Gamma, x \myequiv v_0; s \vdash u_0 \myrel{\leq} u_2$
      \end{itemize}
    \end{itemize}
    
    Let $\Gamma'; s'$ be an extended context such that $\Gamma; s \rightarrowtail \Gamma'; s'$.

    As explained in the introduction paragraph of this lemma,
    we know that in the derivation tree of $\Gamma, x \myequiv v_0; s \vdash u_0 \myrel{\leq} u_2$
    then the rule \textsc{Me-Pro} isn't used
    (which is because otherwise we'd have $\Gamma, x \myequiv v_0; s \vdash u_0 \myrel{\equiv} u_2$
    and we conclude by Lemma~\ref{lem:myrelequiv-has-diamond-property}).
    Therefore this derivation is in fact $\Gamma; s \vdash u_0 \myrel{\leq} u_2$.
    
    By induction hypothesis on $u_0$, there exists a term $u_3$ such that:
    \begin{itemize}
      \item $\Gamma; s \vdash u_2 \myrel{\equiv} u_3$
      \item $\Gamma'; s' \vdash u_1 \myrel{\leq} u_3$
    \end{itemize}

    It is important to note that $x$ is a fresh variable created by the lambdas of this case,
    and therefore by alpha conversion there is no instance of $x$ in $s$, hence $\cas{s}{x}{v_1} = s$.
    By a similar reasoning, we know that there is no $x$ isn't in the domain of $\Gamma'$.
    As a result, the derivation $\Gamma'; s' \vdash u_1 \myrel{\leq} u_3$ do not make a promotion of $x$.
    Hence by Lemma~\ref{lem:subtyping-under-substitution-aux2}, we obtain
    $\Gamma'; s' \vdash \cas{u_1}{x}{v_1} \myrel{\leq} \cas{u_3}{x}{v_1}$.
    
    By rule \textsc{Me-Bet}, we obtain
    $\Gamma; s \vdash (\lambda x \myleq t . u_2) \, v_0 \myrel{\equiv} \cas{u_3}{x}{v_1}$ from
    $\Gamma; s \vdash u_2 \myrel{\equiv} u_3$ and $\Gamma; \nil \vdash v_0 \myrel{\equiv} v_1$.
    
    The filled diagram is therefore the following:
    \begin{center}
      \begin{tikzcd}
      (\lambda x \myleq t . u_2) \, v_0 \arrow[r, "\equiv"] & \cas{u_3}{x}{v_1} \\
      (\lambda x \myleq t . u_0) \, v_0 \arrow[u, "\leq"] \arrow[r, "\equiv"] & \cas{u_1}{x}{v_1} \arrow[u, "\leq"]
      \end{tikzcd}
    \end{center}
  \end{description}
\end{proof}

\textsc{Lemma~\ref{lem:myrelequiv-has-diamond-property} ($\myrel{\equiv}$ has the diamond property).}
  Let $\Gamma_0; s_0$ be an extended context.
  Let $t_0$, $t_1$, and $t_2$ be terms.
  If $\Gamma_0; s_0 \vdash t_0 \myrel{\equiv} t_1$ and $\Gamma_0; s_0 \vdash t_0 \myrel{\equiv} t_2$,
  then for any extended contexts
  $\Gamma_1; s_1$ and $\Gamma_2; s_2$ such that $\Gamma_0; s_0 \rightarrowtail \Gamma_1; s_1$ and $\Gamma_0; s_0 \rightarrowtail \Gamma_2; s_2$,
  there exists a term $t_3$
  such that $\Gamma_1; s_1 \vdash t_1 \myrel{\equiv} t_3$ and $\Gamma_2; s_2 \vdash t_2 \myrel{\equiv} t_3$.

  Moreover, for any variable $x$, if in the derivation of $\Gamma_0; s_0 \vdash t_0 \myrel{\equiv} t_1$
  (respectively $\Gamma_0; s_0 \vdash t_0 \myrel{\equiv} t_2$)
  there isn't an application of the Rule \textsc{Me-Pro} that makes a promotion of variable $x$,
  then in the derivation $\Gamma_2; s_2 \vdash t_2 \myrel{\equiv} t_3$
  (respectively $\Gamma_1; s_1 \vdash t_1 \myrel{\equiv} t_3$)
  there won't be an application of the Rule \textsc{Me-Pro} that makes a promotion of variable $x$.

  Diagrammatically, the existence of the solid arrows implies the existence of the dashed arrows:
  \begin{center}
    \begin{tikzcd}
      t_2 \arrow[r, dashed, "\equiv"]   & t_3                       \\
      t_0 \arrow[r, "\equiv"] \arrow[u, "\equiv"] & t_1 \arrow[u, dashed, "\equiv"]
    \end{tikzcd}
  \end{center}

\begin{proof}
  We now proceed by induction on the derivation tree of $\Gamma_0; s_0 \vdash t_0 \myrel{\equiv} t_1$.
  We consider cases according to the last
  two rules that are applied respectively in the derivation of the horizontal and
  vertical edges of the diagram above, starting by the bottom-left corner. The base case
  corresponds to when these two rules are both \textsc{Me-Top} (respectively both \textsc{Me-Var}),
  for which the lemma holds by using the rule \textsc{Me-Top} (respectively \textsc{Me-Var}) for the missing edges.
  For the general case we distinguish the subcases below:

  \begin{description}
    \item [Case \textsc{Me-Pro} and \textsc{Me-Var}:] We need to prove the diagram below:
    \begin{center}
      \begin{tikzcd}
        \alpha_1 \arrow[r, dashed, "\equiv"] & . \\
        x \arrow[u, "\equiv"] \arrow[r, "\equiv"] & x \arrow[u, dashed, "\equiv"]
      \end{tikzcd}
    \end{center}
    
    By assumption, we have the following reductions:
    \begin{itemize}
      \item $\Gamma_0; s_0 \vdash x \myrel{\equiv} x$.
      \item $\Gamma_0; s_0 \vdash x \myrel{\equiv} \alpha_1$ with $x \myequiv \alpha_0 \in \Gamma_0$ and $\Gamma_0; s_0 \vdash \alpha_0 \myrel{\equiv} \alpha_1$.
    \end{itemize}

    Let $\Gamma_1; s_1$ and $\Gamma_2; s_2$ be extended contexts such that
    $\Gamma_0; s_0 \rightarrowtail \Gamma_1; s_1$ and $\Gamma_0; s_0 \rightarrowtail \Gamma_2; s_2$.

    By Lemma~\ref{lem:commutativity-context-weakening}, \\$\Gamma_0; s_0 \rightarrowtail \Gamma_2; s_2$ implies \\$\Gamma_0; \nil \rightarrowtail \Gamma_2; \nil$.
    Because $\Gamma_0$ is of the form $\Gamma_0', x \myequiv \alpha_0, \Gamma_0''$,\\
    we have $\Gamma_2 = \Gamma_2', x \myequiv \alpha_2, \Gamma_2''$.\\
    By multiple use of the rule \textsc{Ct-Ann}, we have $\Gamma_0'; \nil \vdash \alpha_0 \myrel{\equiv} \alpha_2$.\\
    By weakening (Lemma~\ref{lem:context-weakening}), we have $\Gamma_0; s_0 \vdash \alpha_0 \myrel{\equiv} \alpha_2$.

    By induction hypothesis on $\Gamma_0; s_0 \vdash \alpha_0 \myrel{\equiv} \alpha_1$, there exists a term $\alpha_3$ such that:
    \begin{itemize}
      \item $\Gamma_1; s_1 \vdash \alpha_1 \myrel{\equiv} \alpha_3$
      \item $\Gamma_2; s_2 \vdash \alpha_2 \myrel{\equiv} \alpha_3$
    \end{itemize}

    By rule \textsc{Me-Pro}, we have $\Gamma_1; s_1 \vdash x \myrel{\equiv} \alpha_3$ from
    $x \myequiv \alpha_1 \in \Gamma_1$ and $\Gamma_1; s_1 \vdash \alpha_1 \myrel{\equiv} \alpha_3$.

    The filled diagram is therefore the following:
    \begin{center}
      \begin{tikzcd}
        \alpha_1 \arrow[r, "\equiv"] & \alpha_3 \\
        x \arrow[u, "\equiv"] \arrow[r, "\equiv"] & x \arrow[u, "\equiv"]
      \end{tikzcd}
    \end{center}

    \item [Case \textsc{Me-Pro} and \textsc{Me-Pro}:] We need to prove the diagram below:
    \begin{center}
      \begin{tikzcd}
        \alpha_2 \arrow[r, dashed, "\equiv"] & . \\
        x \arrow[u, "\equiv"] \arrow[r, "\equiv"] & \alpha_1 \arrow[u, dashed, "\equiv"]
      \end{tikzcd}
    \end{center}

    By assumption, we have the following reductions:
    \begin{itemize}
      \item $\Gamma_0; s_0 \vdash \alpha_0 \myrel{\equiv} \alpha_1$ with $x \myequiv \alpha_0 \in \Gamma_0$ and $\Gamma_0; s_0 \vdash \alpha_0 \myrel{\equiv} \alpha_1$.
      \item $\Gamma_0; s_0 \vdash \alpha_0 \myrel{\equiv} \alpha_2$ with $x \myequiv \alpha_0 \in \Gamma_0$ and $\Gamma_0; s_0 \vdash \alpha_0 \myrel{\equiv} \alpha_2$.
    \end{itemize}

    Let $\Gamma_1; s_1$ and $\Gamma_2; s_2$ be extended contexts such that
    $\Gamma_0; s_0 \rightarrowtail \Gamma_1; s_1$ and $\Gamma_0; s_0 \rightarrowtail \Gamma_2; s_2$.

    By induction hypothesis on $\Gamma_0; s_0 \vdash \alpha_0 \myrel{\equiv} \alpha_1$, there exists a term $\alpha_3$ such that:
    \begin{itemize}
      \item $\Gamma_1; s_1 \vdash \alpha_1 \myrel{\equiv} \alpha_3$
      \item $\Gamma_2; s_2 \vdash \alpha_2 \myrel{\equiv} \alpha_3$
    \end{itemize}

    The filled diagram is therefore the following:
    \begin{center}
      \begin{tikzcd}
        \alpha_2 \arrow[r, "\equiv"] & \alpha_3 \\
        x \arrow[u, "\equiv"] \arrow[r, "\equiv"] & \alpha_1 \arrow[u, "\equiv"]
      \end{tikzcd}
    \end{center}

    \item [Case \textsc{Me-App} and \textsc{Me-App}:] We need to prove the diagram below:
    \begin{center}
      \begin{tikzcd}
        u_2\,v_2 \arrow[r, dashed, "\equiv"] & . \\
        u_0\,v_0 \arrow[u, "\equiv"] \arrow[r, "\equiv"] & u_1\,v_1 \arrow[u, dashed, "\equiv"]
      \end{tikzcd}
    \end{center}
    
    By assumption, we have the following reductions:
    \begin{itemize}
      \item $\Gamma_0; s_0 \vdash u_0\,v_0 \myrel{\equiv} u_1\,v_1$ by rule \textsc{Me-App}, which gives us:
      \begin{itemize}
        \item $\Gamma_0; v_0 :: s_0 \vdash u_0 \myrel{\equiv} u_1$
        \item $\Gamma_0; \nil \vdash v_0 \myrel{\equiv} v_1$
      \end{itemize}
      \item $\Gamma_0; s_0 \vdash u_0\,v_0 \myrel{\equiv} u_2\,v_2$ by rule \textsc{Me-App}, which gives us:
      \begin{itemize}
        \item $\Gamma_0; v_0 :: s_0 \vdash u_0 \myrel{\equiv} u_2$
        \item $\Gamma_0; \nil \vdash v_0 \myrel{\equiv} v_2$
      \end{itemize}
      \item $\Gamma_0; s_0 \rightarrowtail \Gamma_1; s_1$
      \item $\Gamma_0; s_0 \rightarrowtail \Gamma_2; s_2$
    \end{itemize}

    Let $\Gamma_1; s_1$ and $\Gamma_2; s_2$ be extended contexts such that
    $\Gamma_0; s_0 \rightarrowtail \Gamma_1; s_1$ and $\Gamma_0; s_0 \rightarrowtail \Gamma_2; s_2$.
    
    By rule \textsc{Ct-Stk}, we know:
    \begin{itemize}
      \item $\Gamma_0; v_0 :: s_0 \rightarrowtail \Gamma_1; v_1 :: s_1$ from $\Gamma_0; \nil \vdash v_0 \myrel{\equiv} v_1$
      \item $\Gamma_0; v_0 :: s_0 \rightarrowtail \Gamma_2; v_2 :: s_2$ from $\Gamma_0; \nil \vdash v_0 \myrel{\equiv} v_2$
    \end{itemize}
    
    By induction hypothesis on $\Gamma_0; v_0 :: s_0 \vdash u_0 \myrel{\equiv} u_1$, there exists a term $u_3$ such that:
    \begin{itemize}
      \item $\Gamma_1; v_1 :: s_1 \vdash u_1 \myrel{\equiv} u_3$
      \item $\Gamma_2; v_2 :: s_2 \vdash u_2 \myrel{\equiv} u_3$
    \end{itemize}
    
    By Lemma~\ref{lem:commutativity-context-weakening}, respectively $\Gamma_0; s_0 \rightarrowtail \Gamma_1; s_1$ and $\Gamma_0; s_0 \rightarrowtail \Gamma_2; s_2$
    implies respectively
    $\Gamma_0; \nil \rightarrowtail \Gamma_1; \nil$ and $\Gamma_0; \nil \rightarrowtail \Gamma_2; \nil$.
    By induction hypothesis on $\Gamma_0; \nil \vdash v_0 \myrel{\equiv} v_1$, there exists a term $v_3$ such that:
    \begin{itemize}
      \item $\Gamma_1; \nil \vdash v_1 \myrel{\equiv} v_3$
      \item $\Gamma_2; \nil \vdash v_2 \myrel{\equiv} v_3$
    \end{itemize}
    
    We now have:
    \begin{itemize}
      \item $\Gamma_1; s_1 \vdash u_1\,v_1 \myrel{\equiv} u_3\,v_3$ by rule \textsc{Me-App}
      from $\Gamma_1; v_1 :: s_1 \vdash u_1 \myrel{\equiv} u_3$ and $\Gamma_1; \nil \vdash v_1 \myrel{\equiv} v_3$
      \item $\Gamma_2; s_2 \vdash u_2\,v_2 \myrel{\equiv} u_3\,v_3$ by rule \textsc{Me-App}
      from $\Gamma_2; v_2 :: s_2 \vdash u_2 \myrel{\equiv} u_3$ and $\Gamma_2; \nil \vdash v_2 \myrel{\equiv} v_3$
    \end{itemize}
    
    The filled diagram is therefore the following:
    \begin{center}
      \begin{tikzcd}
        u_2\,v_2 \arrow[r, "\equiv"] & u_3\,v_3 \\
        u_0\,v_0 \arrow[u, "\equiv"] \arrow[r, "\equiv"] & u_1\,v_1 \arrow[u, "\equiv"]
      \end{tikzcd}
    \end{center}

    \item [Case \textsc{Me-App} with \textsc{Me-Bet}:] We need to prove the diagram below:
    \begin{center}
      \begin{tikzcd}
        \cas{u_2}{x}{v_2} \arrow[r, dashed, "\equiv"] & . \\
        (\lambda x \myleq t_0. u_0)v_0 \arrow[u, "\equiv"] \arrow[r, "\equiv"] & (\lambda x \myleq t_1. u_1)v_1 \arrow[u, dashed, "\equiv"]
      \end{tikzcd}
    \end{center}
    
    By assumption, we have the following reductions:
    \begin{itemize}
      \item $\Gamma_0; s_0 \vdash (\lambda x \myleq t_0. u_0)v_0 \myrel{\equiv} (\lambda x \myleq t_1. u_1)v_1$ by rule \textsc{Me-App}, which gives us:
      \begin{itemize}
        \item $\Gamma_0; v_0 :: s_0 \vdash \lambda x \myleq t_0. u_0 \myrel{\equiv} \lambda x \myleq t_1. u_1$.
        By rule \textsc{Me-FOp}, this implies $\Gamma_0; \nil \vdash t_0 \myrel{\equiv} t_1$ and $\Gamma_0, x \myequiv v_0; s_0 \vdash u_0 \myrel{\equiv} u_1$
        \item $\Gamma_0; \nil \vdash v_0 \myrel{\equiv} v_1$
      \end{itemize}
      \item $\Gamma_0; s_0 \vdash (\lambda x \myleq t_0. u_0)v_0 \myrel{\equiv} \cas{u_2}{x}{v_2}$ by rule \textsc{Me-Bet}, which gives us:
      \begin{itemize}
        \item $\Gamma_0; s_0 \vdash u_0 \myrel{\equiv} u_2$
        \item $\Gamma_0; \nil \vdash v_0 \myrel{\equiv} v_2$
      \end{itemize}
      \item $\Gamma_0; s_0 \rightarrowtail \Gamma_1; s_1$
      \item $\Gamma_0; s_0 \rightarrowtail \Gamma_2; s_2$
    \end{itemize}

    Let $\Gamma_1; s_1$ and $\Gamma_2; s_2$ be extended contexts such that
    $\Gamma_0; s_0 \rightarrowtail \Gamma_1; s_1$ and $\Gamma_0; s_0 \rightarrowtail \Gamma_2; s_2$.

    By weakening (Lemma~\ref{lem:context-weakening}), we obtain $\Gamma_0, x \myequiv v_0; s_0 \vdash u_0 \myrel{\equiv} u_2$
    from $\Gamma_0; s_0 \vdash u_0 \myrel{\equiv} u_2$,
    and we also know that the derivation $\Gamma_0, x \myequiv v_0; s_0 \vdash u_0 \myrel{\equiv} u_2$ do not make a promotion of $x$ to $v_0$.

    We have by rule \textsc{Ct-Ann},
    $\Gamma_0, x \myequiv v_0; s_0 \rightarrowtail \Gamma_1, x \myequiv v_1; s_1$ from $\Gamma_0; \nil \vdash v_0 \myrel{\equiv} v_1$.
    Similarly we have $\Gamma_0, x \myequiv v_0; s_0 \rightarrowtail \Gamma_2, x \myequiv v_2; s_2$.
    By induction hypothesis on $\Gamma_0, x \myequiv v_0; s_0 \vdash u_0 \myrel{\equiv} u_2$, there exists a term $u_3$ such that:
    \begin{itemize}
      \item $\Gamma_1, x \myequiv v_1; s_1 \vdash u_1 \myrel{\equiv} u_3$
      \item $\Gamma_2, x \myequiv v_2; s_2 \vdash u_2 \myrel{\equiv} u_3$
    \end{itemize}

    Because the derivation $\Gamma_0, x \myequiv v_0; s_0 \vdash u_0 \myrel{\equiv} u_2$ do not make any promotion of $x$ to $v_0$,
    then the induction process ensures that the derivation $\Gamma_1, x \myequiv v_1; s_1 \vdash u_1 \myrel{\equiv} u_3$
    do not make any promotion of $x$ to $v_1$.
    As a result we have $\Gamma_1; s_1 \vdash u_1 \myrel{\equiv} u_3$.

    By Lemma~\ref{lem:commutativity-context-weakening}, respectively $\Gamma_0; s_0 \rightarrowtail \Gamma_1; s_1$ and $\Gamma_0; s_0 \rightarrowtail \Gamma_2; s_2$
    implies respectively
    $\Gamma_0; \nil \rightarrowtail \Gamma_1; \nil$ and $\Gamma_0; \nil \rightarrowtail \Gamma_2; \nil$.
    By induction hypothesis on $\Gamma_0; \nil \vdash v_0 \myrel{\equiv} v_1$, there exists a term $v_3$ such that:
    \begin{itemize}
      \item $\Gamma_1; \nil \vdash v_1 \myrel{\equiv} v_3$
      \item $\Gamma_2; \nil \vdash v_2 \myrel{\equiv} v_3$
    \end{itemize}
    
    We now have:
    \begin{itemize}
      \item $\Gamma_1; s_1 \vdash (\lambda x \myleq t_1. u_1)v_1 \myrel{\equiv} \cas{u_3}{x}{v_3}$ by rule \textsc{Me-Bet}
      from $\Gamma_1; s_1 \vdash u_1 \myrel{\equiv} u_3$ and $\Gamma_1; \nil \vdash v_1 \myrel{\equiv} v_3$
      \item $\Gamma_2; s_2 \vdash \cas{u_2}{x}{v_2} \myrel{\equiv} \cas{u_3}{x}{v_3}$ by Lemma~\ref{lem:commutativity-reduction-under-substitution}
      from $\Gamma_2, x \myequiv v_2; s_2 \vdash u_2 \myrel{\equiv} u_3$ and $\Gamma_2; \nil \vdash v_2 \myrel{\equiv} v_3$
    \end{itemize}
    
    The filled diagram is therefore the following:
    \begin{center}
      \begin{tikzcd}
        \cas{u_2}{x}{v_2} \arrow[r, "\equiv"] & \cas{u_3}{x}{v_3} \\
        (\lambda x \myleq t_0. u_0)v_0 \arrow[u, "\equiv"] \arrow[r, "\equiv"] & (\lambda x \myleq t_1. u_1)v_1 \arrow[u, "\equiv"]
      \end{tikzcd}
    \end{center}

    \item [Case \textsc{Me-App} with \textsc{Me-TAp}:] We need to prove the diagram below:
    \begin{center}
      \begin{tikzcd}
        \T \arrow[r, dashed, "\equiv"] & \T \\
        \T \, u \arrow[u, "\equiv"] \arrow[r, "\equiv"] & \T \, u' \arrow[u, dashed, "\equiv"]
      \end{tikzcd}
    \end{center}
    
    The top edge holds by rule \textsc{Me-Top}.
    The right edge holds by rule \textsc{Me-TAp}.

    \item [Case \textsc{Me-Fun} with \textsc{Me-Fun}:] We need to prove the diagram below:
    \begin{center}
      \begin{tikzcd}
        \lambda x \myleq t_2. u_2 \arrow[r, dashed, "\equiv"] & . \\
        \lambda x \myleq t_0. u_0 \arrow[u, "\equiv"] \arrow[r, "\equiv"] & \lambda x \myleq t_1. u_1 \arrow[u, dashed, "\equiv"]
      \end{tikzcd}
    \end{center}
    
    By assumption, we have the following reductions:
    \begin{itemize}
      \item $\Gamma_0; \nil \vdash \lambda x \myleq t_0. u_0 \myrel{\equiv} \lambda x \myleq t_1. u_1$ by rule \textsc{Me-Fun}, which gives us:
      \begin{itemize}
        \item $\Gamma_0; \nil \vdash t_0 \myrel{\equiv} t_1$
        \item $\Gamma_0, x \myleq t_0; \nil \vdash u_0 \myrel{\equiv} u_1$
      \end{itemize}
      \item $\Gamma_0; \nil \vdash \lambda x \myleq t_0. u_0 \myrel{\equiv} \lambda x \myleq t_2. u_2$ by rule \textsc{Me-Fun}, which gives us:
      \begin{itemize}
        \item $\Gamma_0; \nil \vdash t_0 \myrel{\equiv} t_2$
        \item $\Gamma_0, x \myleq t_0; \nil \vdash u_0 \myrel{\equiv} u_2$
      \end{itemize}
    \end{itemize}

    Let $\Gamma_1; s_1$ and $\Gamma_2; s_2$ be extended contexts such that
    $\Gamma_0; s_0 \rightarrowtail \Gamma_1; s_1$ and $\Gamma_0; s_0 \rightarrowtail \Gamma_2; s_2$.
    Because by assumption of this case we have $s_0 = \nil$, it implies that $s_1 = \nil$ and $s_2 = \nil$,
    and therefore the relations above are $\Gamma_0; \nil \rightarrowtail \Gamma_1; \nil$ and $\Gamma_0; \nil \rightarrowtail \Gamma_2; \nil$.
    
    By induction hypothesis on $\Gamma_0; \nil \vdash t_0 \myrel{\equiv} t_1$, there exists a term $t_3$ such that:
    \begin{itemize}
      \item $\Gamma_1; \nil \vdash t_1 \myrel{\equiv} t_3$
      \item $\Gamma_2; \nil \vdash t_2 \myrel{\equiv} t_3$
    \end{itemize}
    
    By rule \textsc{Ct-Ann}, we know:
    \begin{itemize}
      \item $\Gamma_0, x \myleq t_0; \nil \rightarrowtail \Gamma_1, x \myleq t_1; \nil$ from $\Gamma_0; \nil \vdash t_0 \myrel{\equiv} t_1$
      \item $\Gamma_0, x \myleq t_0; \nil \rightarrowtail \Gamma_2, x \myleq t_2; \nil$ from $\Gamma_0; \nil \vdash t_0 \myrel{\equiv} t_2$
    \end{itemize}
    
    By induction hypothesis on $\Gamma_0, x \myleq t_0; \nil \vdash u_0 \myrel{\equiv} u_1$, there exists a term $u_3$ such that:
    \begin{itemize}
      \item $\Gamma_1, x \myleq t_1; \nil \vdash u_1 \myrel{\equiv} u_3$
      \item $\Gamma_2, x \myleq t_2; \nil \vdash u_2 \myrel{\equiv} u_3$
    \end{itemize}
    
    We now have:
    \begin{itemize}
      \item $\Gamma_1; \nil \vdash \lambda x \myleq t_1. u_1 \myrel{\equiv} \lambda x \myleq t_3. u_3$ \\by rule \textsc{Me-Fun}
      from $\Gamma_1; \nil \vdash t_1 \myrel{\equiv} t_3$ and $\Gamma_1, x \myleq t_1; \nil \vdash u_1 \myrel{\equiv} u_3$
      \item $\Gamma_2; \nil \vdash \lambda x \myleq t_2. u_2 \myrel{\equiv} \lambda x \myleq t_3. u_3$ \\by rule \textsc{Me-Fun}
      from $\Gamma_2; \nil \vdash t_2 \myrel{\equiv} t_3$ and $\Gamma_2, x \myleq t_2; \nil \vdash u_2 \myrel{\equiv} u_3$
    \end{itemize}
    
    The filled diagram is therefore the following:
    \begin{center}
      \begin{tikzcd}
        \lambda x \myleq t_2. u_2 \arrow[r, "\equiv"] & \lambda x \myleq t_3. u_3 \\
        \lambda x \myleq t_0. u_0 \arrow[u, "\equiv"] \arrow[r, "\equiv"] & \lambda x \myleq t_1. u_1 \arrow[u, "\equiv"]
      \end{tikzcd}
    \end{center}

    \item [Case \textsc{Me-FOp} with \textsc{Me-FOp}:] We need to prove the diagram below:
    \begin{center}
      \begin{tikzcd}
        \lambda x \myleq t_2. u_2 \arrow[r, dashed, "\equiv"] & . \\
        \lambda x \myleq t_0. u_0 \arrow[u, "\equiv"] \arrow[r, "\equiv"] & \lambda x \myleq t_1. u_1 \arrow[u, dashed, "\equiv"]
      \end{tikzcd}
    \end{center}
    
    By assumption, we have the following reductions:
    \begin{itemize}
      \item $\Gamma_0; \alpha_0 :: s_0 \vdash \lambda x \myleq t_0. u_0 \myrel{\equiv} \lambda x \myleq t_1. u_1$ by rule \textsc{Me-FOp}, which gives us:
      \begin{itemize}
      \item $\Gamma_0; \nil \vdash t_0 \myrel{\equiv} t_1$
      \item $\Gamma_0, x \myequiv \alpha_0; s_0 \vdash u_0 \myrel{\equiv} u_1$
      \end{itemize}
      \item $\Gamma_0; \alpha_0 :: s_0 \vdash \lambda x \myleq t_0. u_0 \myrel{\equiv} \lambda x \myleq t_2. u_2$ by rule \textsc{Me-FOp}, which gives us:
      \begin{itemize}
      \item $\Gamma_0; \nil \vdash t_0 \myrel{\equiv} t_2$
      \item $\Gamma_0, x \myequiv \alpha_0; s_0 \vdash u_0 \myrel{\equiv} u_2$
      \end{itemize}
      \item The stack of our case is $s = \alpha_0 :: s_0$
    \end{itemize}

    Let $\Gamma_1; \alpha_1 :: s_1$ and $\Gamma_2; \alpha_2 :: s_2$ be extended contexts such that
    $\Gamma_0; \alpha_0 :: s_0 \rightarrowtail \Gamma_1; \alpha_1 :: s_1$
    and $\Gamma_0; \alpha_0 :: s_0 \rightarrowtail \Gamma_2; \alpha_2 :: s_2$.
    By rule \textsc{Ct-Stk}, those relations respectively implies
    $\Gamma_0; \nil \vdash \alpha_0 \myrel{\equiv} \alpha_1$ and $\Gamma_0; \nil \vdash \alpha_0 \myrel{\equiv} \alpha_2$.
    
    By Lemma~\ref{lem:commutativity-context-weakening}, respectively $\Gamma_0; \alpha_0 :: s_0 \rightarrowtail \Gamma_1; \alpha_1 :: s_1$
    and $\Gamma_0; \alpha_0 :: s_0 \rightarrowtail \Gamma_2; \alpha_2 :: s_2$
    implies respectively
    $\Gamma_0; \nil \rightarrowtail \Gamma_1; \nil$ and $\Gamma_0; \nil \rightarrowtail \Gamma_2; \nil$.
    By induction hypothesis on $\Gamma_0; \nil \vdash t_0 \myrel{\equiv} t_1$, there exists a term $t_3$ such that:
    \begin{itemize}
      \item $\Gamma_1; \nil \vdash t_1 \myrel{\equiv} t_3$
      \item $\Gamma_2; \nil \vdash t_2 \myrel{\equiv} t_3$
    \end{itemize}
    
    By rule \textsc{Ct-Ann}, we know:
    \begin{itemize}
      \item $\Gamma_0, x \myequiv \alpha_0; s_0 \rightarrowtail \Gamma_1, x \myequiv \alpha_1; s_1$ from $\Gamma_0; \nil \vdash \alpha_0 \myrel{\equiv} \alpha_1$
      \item $\Gamma_0, x \myequiv \alpha_0; s_0 \rightarrowtail \Gamma_2, x \myequiv \alpha_2; s_2$ from $\Gamma_0; \nil \vdash \alpha_0 \myrel{\equiv} \alpha_2$
    \end{itemize}
    
    By induction hypothesis on $Gamma_0, x \myequiv \alpha_0; s_0 \vdash u_0 \myrel{\equiv} u_1$, there exists a term $u_3$ such that:
    \begin{itemize}
      \item $\Gamma_1, x \myequiv \alpha_1; s_1 \vdash u_1 \myrel{\equiv} u_3$
      \item $\Gamma_2, x \myequiv \alpha_2; s_2 \vdash u_2 \myrel{\equiv} u_3$
    \end{itemize}
    
    We now have:
    \begin{itemize}
      \item $\Gamma_1; \alpha_1 :: s_1 \vdash \lambda x \myleq t_1. u_1 \myrel{\equiv} \lambda x \myleq t_3. u_3$ by rule \textsc{Me-FOp}
      from $\Gamma_1; \nil \vdash t_1 \myrel{\equiv} t_3$ and $\Gamma_1, x \myequiv \alpha_1; s_1 \vdash u_1 \myrel{\equiv} u_3$
      \item $\Gamma_2; \alpha_2 :: s_2 \vdash \lambda x \myleq t_2. u_2 \myrel{\equiv} \lambda x \myleq t_3. u_3$ by rule \textsc{Me-FOp}
      from $\Gamma_2; \nil \vdash t_2 \myrel{\equiv} t_3$ and $\Gamma_2, x \myequiv \alpha_2; s_2 \vdash u_2 \myrel{\equiv} u_3$
    \end{itemize}
    
    The filled diagram is therefore the following:
    \begin{center}
      \begin{tikzcd}
        \lambda x \myleq t_2. u_2 \arrow[r, "\equiv"] & \lambda x \myleq t_3. u_3 \\
        \lambda x \myleq t_0. u_0 \arrow[u, "\equiv"] \arrow[r, "\equiv"] & \lambda x \myleq t_1. u_1 \arrow[u, "\equiv"]
      \end{tikzcd}
    \end{center}

    \item [Case \textsc{Me-Bet} with \textsc{Me-Bet}:] We need to prove the diagram below:
    \begin{center}
      \begin{tikzcd}
        \cas{u_2}{x}{v_2}\arrow[r, dashed, "\equiv"] & . \\
        (\lambda x \myleq t. u_0) v_0 \arrow[u, "\equiv"] \arrow[r, "\equiv"] & \cas{u_1}{x}{v_1}\arrow[u, dashed, "\equiv"]
      \end{tikzcd}
    \end{center}
    
    By assumption, we have the following reductions:
    \begin{itemize}
      \item $\Gamma_0; s_0 \vdash (\lambda x \myleq t. u_0) v_0 \myrel{\equiv} \cas{u_1}{x}{v_1}$ by rule \textsc{Me-Bet}, which gives us:
      \begin{itemize}
        \item $\Gamma_0; s_0 \vdash u_0 \myrel{\equiv} u_1$
        \item $\Gamma_0; \nil \vdash v_0 \myrel{\equiv} v_1$
      \end{itemize}
      \item $\Gamma_0; s_0 \vdash (\lambda x \myleq t. u_0) v_0 \myrel{\equiv} \cas{u_2}{x}{v_2}$ by rule \textsc{Me-Bet}, which gives us:
      \begin{itemize}
        \item $\Gamma_0; s_0 \vdash u_0 \myrel{\equiv} u_2$
        \item $\Gamma_0; \nil \vdash v_0 \myrel{\equiv} v_2$
      \end{itemize}
      \item $\Gamma_0; s_0 \rightarrowtail \Gamma_1; s_1$
      \item $\Gamma_0; s_0 \rightarrowtail \Gamma_2; s_2$
    \end{itemize}

    Let $\Gamma_1; s_1$ and $\Gamma_2; s_2$ be extended contexts such that
    $\Gamma_0; s_0 \rightarrowtail \Gamma_1; s_1$ and $\Gamma_0; s_0 \rightarrowtail \Gamma_2; s_2$.
    
    By induction hypothesis on $\Gamma_0; s_0 \vdash u_0 \myrel{\equiv} u_1$, there exists a term $u_3$ such that:
    \begin{itemize}
      \item $\Gamma_1; s_1 \vdash u_1 \myrel{\equiv} u_3$
      \item $\Gamma_2; s_2 \vdash u_2 \myrel{\equiv} u_3$
    \end{itemize}
    
    By Lemma~\ref{lem:commutativity-context-weakening}, respectively $\Gamma_0; s_0 \rightarrowtail \Gamma_1; s_1$ and $\Gamma_0; s_0 \rightarrowtail \Gamma_2; s_2$
    implies respectively
    $\Gamma_0; \nil \rightarrowtail \Gamma_1; \nil$ and $\Gamma_0; \nil \rightarrowtail \Gamma_2; \nil$.
    By induction hypothesis on $\Gamma_0; \nil \vdash v_0 \myrel{\equiv} v_1$, there exists a term $v_3$ such that:
    \begin{itemize}
      \item $\Gamma_1; \nil \vdash v_1 \myrel{\equiv} v_3$
      \item $\Gamma_2; \nil \vdash v_2 \myrel{\equiv} v_3$
    \end{itemize}

    By weakening (Lemma~\ref{lem:context-weakening}) we have:
    \begin{itemize}
      \item $\Gamma_1, x \myequiv v_1; s_1 \vdash u_1 \myrel{\equiv} u_3$
      \item $\Gamma_2, x \myequiv v_2; s_2 \vdash u_2 \myrel{\equiv} u_3$
    \end{itemize}
    
    By Lemma~\ref{lem:commutativity-reduction-under-substitution}, we have:
    \begin{itemize}
      \item $\Gamma_1; s_1 \vdash \cas{u_1}{x}{v_1} \myrel{\equiv} \cas{u_3}{x}{v_3}$
      from $\Gamma_1, x \myequiv v_1; s_1 \vdash u_1 \myrel{\equiv} u_3$ and $\Gamma_1; \nil \vdash v_1 \myrel{\equiv} v_3$
      \item $\Gamma_2; s_2 \vdash \cas{u_2}{x}{v_2} \myrel{\equiv} \cas{u_3}{x}{v_3}$
      from $\Gamma_2, x \myequiv v_2; s_2 \vdash u_2 \myrel{\equiv} u_3$ and $\Gamma_2; \nil \vdash v_2 \myrel{\equiv} v_3$
    \end{itemize}
    
    The filled diagram is therefore the following:
    \begin{center}
      \begin{tikzcd}
        \cas{u_2}{x}{v_2} \arrow[r, "\equiv"] & \cas{u_3}{x}{v_3} \\
        (\lambda x \myleq t. u_0) v_0 \arrow[u, "\equiv"] \arrow[r, "\equiv"] & \cas{u_1}{x}{v_1} \arrow[u, "\equiv"]
      \end{tikzcd}
    \end{center}

    \item [Case \textsc{Me-TAp} with \textsc{Me-TAp}:] The following diagram holds by rule \textsc{Me-Top}:
    \begin{center}
      \begin{tikzcd}
        \T \arrow[r, "\equiv"] & \T \\
        \T \, u \arrow[u, "\equiv"] \arrow[r, "\equiv"] & \T \arrow[u, "\equiv"]
      \end{tikzcd}
    \end{center}
    
  \end{description}
\end{proof}

\textsc{Theorem~\ref{lem:algorithmic-transitivity-elimination} (Transivity is admissible).}
Let $\Gamma; s$ be an extended context.
Let $u$ and $v$ be terms.
If $\Gamma; s \vdash u \leq^* v$ then $\Gamma; s \vdash u \leq v$.
\begin{proof}
  By induction on the number of intermediary steps of the derivation $\Gamma;s\vdash u\leq^* v$.
  The base, which is either 0 or 1 derivation yields the result $\Gamma; s \vdash u \leq v$.
  Else, it suffices to prove that for every
  terms $t$, $u$, and $v$, if $\Gamma;s\vdash t\leq u$ and $\Gamma;s\vdash u\leq v$, then
  $\Gamma;s\vdash t\leq v$. From $\Gamma;s\vdash t\leq u$ and $\Gamma;s\vdash u\leq v$ we
  obtain the diagram below
  \begin{center}
    \begin{tikzcd}
      v \arrow[r, two heads, "\equiv"] & .
      \arrow[r, two heads, dashed, "\equiv"] & . \\
      & u \arrow[u, two heads, "\leq"]
      \arrow[r, two heads, "\equiv"]  & .
      \arrow[u, two heads, dashed, "\leq"] \\
      & & t \arrow[u, two heads, "\leq"]
    \end{tikzcd}
  \end{center}
  where the solid edges correspond to the definition of the subtyping relation. The result
  $\Gamma;s\vdash t\leq v$ follows by completing the diagram with the dashed arrows, which
  can be done by commutativity of $\myrel{\equiv}$ and $\myrel{\leq}$ established in
  Theorem~\ref{thm:commutativity-theorem} below.
\end{proof}

\textsc{Theorem~\ref{thm:progress} (Progress).}
Let $t$ be a term. For every logical context $\Gamma$, if $\Gamma \vdash t~\wef$ then either $t$ is
in normal form, or there exists a term $t'$ such that $t \mapsto t'$.
\begin{proof}
  By structural induction on the term $t$. If $t$ is $\T$ or a variable $x$, then the result holds
  since $t$ is already in normal form. Therefore it suffices to consider the following cases:
  \begin{description}
  \item[Case $t =\lambda x \myleq u. v$ with $u\not\in\NF$:] By assumption we have
    $\Gamma \vdash \lambda x \myleq u. v~\wef$.  By rule \textsc{Wf-Fun},
    we obtain $\Gamma \vdash u ~\wef$.  By the induction hypothesis
    on $u$, there exists $u'$ such that $u \mapsto u'$, and as a result
    $t \mapsto \lambda x \myleq u'.v$ by rule \textsc{Os-Con}.
  \item[Case $t =\lambda x \myleq u_n.v$ with $v\not\in\NF$:] By assumption we have
    $\Gamma \vdash \lambda x \myleq u_n. v~\wef$. By rule \textsc{Wf-Fun},
    we have that $\Gamma,x\myleq u_n \vdash v ~\wef$.
    By using the induction hypothesis on $v$, there exists $v'$
    such that $v \mapsto v'$, and as a result $t \mapsto \lambda x \myleq u_n. v'$ by rule
    \textsc{Os-Con}.
  \item[Case $t =u\,v$ with $u\not\in\NF$:] By assumption we have
    $\Gamma \vdash u\,v~\wef$. By rule \textsc{Wf-App} and Proposition~\ref{prop:transitivity-preserves-well-formedness} we have
    $\Gamma \vdash u~\wef$. Since $u$ is not a normal form and by the induction
    hypothesis, there exists a term $u'$ such that $u \mapsto u'$, and therefore
    $t \mapsto u' \, v$ by rule \textsc{Os-Con}.
  \item[Case $t =u_n\,v$ with $v\not\in\NF$.] By assumption we have
    $\Gamma \vdash u_n\,v~\wef$.
    By rule \textsc{Wf-App} and Proposition~\ref{prop:transitivity-preserves-well-formedness} we have
    $\Gamma \vdash v~\wef$. Since $v$ is not a normal form and by the induction
    hypothesis, there exists a term $v'$ such that $v \mapsto v'$, and therefore
    $t \mapsto u_n \, v'$ by rule \textsc{Os-Con}.
  \item[Case $t=u_n\,v_n$ with $u_n\in\NF$ and $v_n\in\NF$.] Term $u_n$ is either $\T$, a
    function, or a neutral applied to multiple normal forms. If $u_n$ is a function, then $t$ is a redex can be reduced by
    rule \textsc{Os-Bet}.
    The case where $u_n$ is $\T$ is impossible, indeed
    from the fact that $t$ is well-formed, by rule \textsc{W-App} we have $u_n$ a subtype of a function $\lambda x \myleq t'. \T$.
    If $u_n$ is $\T$, it cannot have a function supertype by Lemma~\ref{lem:no-supertype-top}.
    Else, if $u_n$ is a neutral applied to multiple normal forms,
    then $t = u_n\,v_n$, because $v_n$ is in normal form, is also a neutral applied to multiple normal forms,
    and therefore $t$ is in normal form. \qedhere
  \end{description}
\end{proof}

\textsc{Theorem~\ref{thm:preservation} (Preservation).}
Let $\Gamma$ be a logical context.
Let $t$, $t'$, and $u$ be terms. If
$\Gamma \vdash t \leq^*_\wef u$ and $t \mapsto t'$, then
$\Gamma \vdash t' \leq^*_\wef u$.
\begin{proof}
  To prove $\Gamma \vdash t' \leq^*_\wef u$ we use rule \textsc{Ws-Trs}
  and prove $\Gamma \vdash t' \leq^*_\wef t$ and $\Gamma \vdash t \leq^*_\wef u$.
  We already have $\Gamma \vdash t \leq^*_\wef u$, so we only prove $\Gamma \vdash t' \leq^*_\wef t$.

  From the assumption $\Gamma \vdash t \leq^*_\wef u$ and by
  Proposition~\ref{prop:transitivity-preserves-well-formedness} we know that
  $\Gamma \vdash t~\wef$. We obtain $\Gamma \vdash t'~\wef$ from the $\Gamma \vdash t~\wef$ established
  earlier and by Lemma~\ref{lem:mapsto-preserves-wf} and $t \mapsto t'$.

  By rule \textsc{Ws-Rfl} and \textsc{Ws-Sub}, we have
  $\Gamma \vdash t' \leq^*_\wef t'$,
  hence $\Gamma \vdash t' \leq^*_\wef t$ by rule \textsc{Ws-Rgh}.
\end{proof}

\textsc{Lemma~\ref{lem:mapsto-preserves-wf} (Evaluation preserves well-formedness).}
  Let $\Gamma$ be a logical context.
  Let $t$ and $t'$ be terms such that $\Gamma \vdash t ~\wef$
  and $t \mapsto t'$.
  We have $\Gamma \vdash t' ~\wef$.
\begin{proof}
  We prove that $\Gamma \vdash t'~\wef$ by induction on $t \mapsto t'$.
  We distinguish based on the last rule used in the derivation $t \mapsto t'$.

  \begin{description}
    \item[Case $t=(\lambda x\myleq u.v)w$ and $t' = \cas{v}{x}{w}$ by rule \textsc{Os-Bet}:]
    By the assumption \\$\Gamma \vdash (\lambda x\leq u.v)w~\wef$
    and by rule \textsc{Wf-App}, we have
    $\Gamma \vdash \lambda x\leq u.v \leq^*_\wef \lambda x\leq z.\T$ and
    $\Gamma \vdash w\leq^*_\wef z$.
    By inversion (Lemma~\ref{lem:inversion-lemma}),
    we obtain $\Gamma \vdash u \equiv_\wef z$.
    By Lemma~\ref{lem:equiv-commutative}, we obtain $\Gamma \vdash z \equiv_\wef u$,
    hence by Lemma~\ref{lem:equiv-to-subtyping} $\Gamma \vdash z \leq_\wef u$.
    Because both $z$ and $u$ are well-formed in $\Gamma$, we have $\Gamma \vdash z \leq^*_\wef u$ by rule \textsc{Ws-Sub}.
    By transitivity, we now have $\Gamma \vdash w\leq^*_\wef u$.
    
    From $\Gamma \vdash \lambda x\leq u.v \leq^*_\wef \lambda x\leq z.\T$
    and Proposition~\ref{prop:transitivity-preserves-well-formedness},
    we obtain $\Gamma \vdash \lambda x\leq u.v ~\wef$.
    Hence by rule \textsc{Wf-Fun}, we obtain $\Gamma,x\myleq u\vdash v~\wef$.

    Now, by Lemma~\ref{lem:substitution-preserves-wf} on
    $\Gamma,x\myleq u\vdash v~\wef$ and $\Gamma \vdash w\leq^*_\wef u$,
    we obtain $\Gamma\vdash\cas{v}{x}{w}~\wef$.

  \item[Case $t = \lambda x \myleq u.v$ and $t' = \lambda x\myleq u'.v$ by rule \textsc{Os-Con}:]
    By assumption, $\Gamma\vdash\lambda x\myleq u.v~\wef$ holds, and we have $u \mapsto u'$.

    By assumption, we have $\Gamma\vdash\lambda x\myleq u.v~\wef$. By rule
    \textsc{Wf-Fun} we obtain $\Gamma,x\myleq u\vdash v~\wef$ and
    $\Gamma \vdash u ~\wef$.
    By using the induction hypothesis on $u \mapsto u'$ we obtain$\Gamma\vdash u'~\wef$.
    From Proposition~\ref{prop:mapsto-inclusion}, we have $\Gamma; \nil \vdash u \myrel{\equiv} u'$.
    By Lemma~\ref{lem:narrowing-context-wf}, we obtain $\Gamma,x\myleq u'\vdash v~\wef$. Finally, we
    obtain $\Gamma\vdash\lambda x\myleq u'.v~\wef$ by application of rule
    \textsc{Wf-Fun}.

  \item[Case $t = \lambda x \myleq u.v$ and $t' = \lambda x\myleq u.v'$ by rule \textsc{Os-Con}:]
    By assumption, $\Gamma\vdash\lambda x\myleq u.v~\wef$ holds, and we have $v \mapsto v'$.

    By rule \textsc{Wf-Fun} we have that $\Gamma,x\myleq u\vdash v~\wef$, and by the induction
    hypothesis on $v \mapsto v'$ we obtain $\Gamma,x\myleq u\vdash v'~\wef$ since
    $\Gamma,x\myleq u\vdash v~\wef$.
    Finally, we obtain $\Gamma\vdash\lambda x\myleq u.v'~\wef$ by application of rule \textsc{Wf-Fun}.
    
  \item[Case $t = u \, v$ and $t' = u' \, v$ by rule \textsc{Os-Con}:]
    By assumption,
    $\Gamma\vdash u\,v~\wef$ holds. By rule \textsc{Wf-App} and Proposition
   ~\ref{prop:transitivity-preserves-well-formedness}, we obtain that $\Gamma\vdash u~\wef$.
    By applying the induction hypothesis on $u \mapsto u'$,
    we obtain $\Gamma\vdash u'~\wef$.
    By rule \textsc{Wf-App}
    we obtain $\Gamma\vdash u\leq^*_{\wef} \lambda x\myleq w.\T$ and $\Gamma\vdash v\leq^*_{\wef} w$.
    Hence $\Gamma\vdash u'\leq^*_{\wef} \lambda x\myleq w.\T$,
    By Lemma~\ref{lem:composability-reverse-asleft}
    
    Finally, we obtain $\Gamma \vdash u'\,v~\wef$ by application of rules \textsc{Wf-App} from
    $\Gamma\vdash u'\leq^*_{\wef}\lambda x\myleq w.\T$ and $\Gamma\vdash v\leq^*_{\wef} w$.

  \item[Case $t = u \, v$ and $t' = u \, v'$ by rule \textsc{Os-Con}:]
    By assumption, $\Gamma\vdash u\,v~\wef$ holds. By rule \textsc{Wf-App},
    Propositions~\ref{prop:transitivity-preserves-well-formedness} and
     we obtain $\Gamma\vdash u~\wef$, $\Gamma \vdash u\leq^*_{\wef} \lambda x\myleq w.\T$, $\Gamma\vdash v~\wef$,
    $\Gamma\vdash v\leq^*_{\wef} w$ and $\Gamma\vdash w~\wef$. By the induction hypothesis on $v \mapsto v'$, we obtain
    $\Gamma\vdash v'~\wef$.
    By Lemma~\ref{lem:composability-reverse-asleft}, we also obtain $\Gamma\vdash v'\leq^*_{\wef} w$.

    The result finally
    $\Gamma\vdash u\,v'~\wef$ holds by application of rules \textsc{Wf-App}
    $\Gamma \vdash u\leq^*_{\wef} \lambda x\myleq w.\T$, and
    $\Gamma \vdash v'\leq^*_{\wef} w$.\qedhere
  \end{description}
\end{proof}

\textsc{Lemma~\ref{lem:substitution-preserves-wf} (Substitution preserves well-formedness).}
  Let $\Gamma$ and $\Gamma'$ be logical contexts, and let $t$, $u$ and $\alpha$ be terms.
  Let $x$ be a variable.
  If
  $\Gamma,x\myleq t,\Gamma' \vdash u~\wef$, and
  $\Gamma \vdash \alpha \leq^*_{\wef} t$,
  then
  $\Gamma,\cas{\Gamma'}{x}{\alpha} \vdash \cas{u}{x}{\alpha}~\wef$.
\begin{proof}
  First, we establish that $\Gamma,\cas{\Gamma'}{x}{\alpha}$ is prevalid by
  Lemma~\ref{lem:substitution-preserves-prevalidity}, which we'll use in most of the cases
  below. Now we proceed by induction on the derivation tree of
  $\Gamma,x\myleq t,\Gamma'\vdash u~\wef$. We distinguish cases on the last rule used:
  \begin{description}
  \item[Rule \textsc{Wf-PrS} or \textsc{Wf-PrE}, with $u = y$ and $y \myvartriangleleft t' \in \Gamma$:]
    By assumption, we have $y \myvartriangleleft t' \in \Gamma$ hence
    $y \myvartriangleleft t' \in \Gamma \cup \cas{\Gamma'}{x}{\alpha}$.
    As proven at the start of this lemma, we have $\Gamma,\cas{\Gamma'}{x}{\alpha}$ prevalid,
    hence $\Gamma,\cas{\Gamma'}{x}{\alpha} \vdash y~\wef$ by rule \textsc{Wf-PrS} or \textsc{Wf-PrE}.
  \item[Rule \textsc{Wf-PrS} or \textsc{Wf-PrE}, with $u = y$ and $y \myvartriangleleft t' \in \Gamma'$:]
    From the fact that $y \myvartriangleleft t'$ is a subtype annotation of $\Gamma'$, we have that
    $y \myvartriangleleft \cas{t'}{x}{\alpha}$ is a subtype annotation of $\cas{\Gamma'}{x}{\alpha}$ by
    definition of the substitution of extended context which doesn't change the domain of the logical context.
    We also have $y \myvartriangleleft \cas{t'}{x}{\alpha} \in \Gamma \cup \cas{\Gamma'}{x}{\alpha}$
    from $y \myvartriangleleft \cas{t'}{x}{\alpha} \in \cas{\Gamma'}{x}{\alpha}$.
    As proven at the start of this lemma, we have $\Gamma,\cas{\Gamma'}{x}{\alpha}$ prevalid,
    we can now conclude that
    $\Gamma,\cas{\Gamma'}{x}{\alpha} \vdash y~\wef$ by rule \textsc{Wf-PrS} or \textsc{Wf-PrE}.

  \item[Rule \textsc{Wf-PrS} or \textsc{Wf-PrE}, with $u = x$:]
    By assumption of our lemma, we have $\Gamma \vdash \alpha ~\wef$,
    hence by weakening (Lemma~\ref{lem:context-weakening}), $\Gamma,\cas{\Gamma'}{x}{\alpha} \vdash \alpha ~\wef$,
    which is the desired result of this case.

  \item[Rule \textsc{Wf-Top}, with $u = \T$:] We have
    $\Gamma,\cas{\Gamma'}{x}{\alpha} \vdash \T ~\wef$ since
    $\Gamma,\cas{\Gamma'}{x}{\alpha}$ is prevalid and by rule \textsc{Wf-Top}.

  \item[Rule \textsc{Wf-Fun}, with $u = \lambda y\myleq v.u'$:]
    Without loss of generality we can assume that $y\not=x$ because we can always alpha convert the variable $y$ to a different name.
    By assumption and by the rule used, we
    have $\Gamma,x\myleq t,\Gamma' \vdash v~\wef$ and
    $\Gamma,x\myleq t,\Gamma',y\myleq v \vdash u'~\wef$. By the induction hypothesis, we obtain
    $\Gamma,\cas{\Gamma'}{x}{\alpha} \vdash\cas{v}{x}{\alpha}~\wef$ and
    $\Gamma,\cas{\Gamma'}{x}{\alpha},y\myleq\cas{v}{x}{\alpha} \vdash\cas{u'}{x}{\alpha}~\wef$. 
    Now, we obtain 
    $\Gamma,\cas{\Gamma'}{x}{\alpha} \vdash\lambda y\myleq\cas{v}{x}{\alpha}.\cas{u'}{x}{\alpha}~\wef$ by rule \textsc{Wf-Fun},
    which equals $\Gamma,\cas{\Gamma'}{x}{\alpha} \vdash\cas{(\lambda y\myleq v.u')}{x}{\alpha}~\wef$ 
    since $y\not=x$ by definition of the substitution.

  \item[Rule \textsc{Wf-App} with $u = v\,u'$:] By \textsc{Wf-App} on our assumption, we have
    $\Gamma, x \myleq t, \Gamma' \vdash u' \leq^*_\wef z$ and
    $\Gamma, x \myleq t, \Gamma' \vdash v \leq^*_\wef \lambda y \myleq z. \T$.

    Let's prove by induction on
    $\Gamma, x \myleq t, \Gamma' \vdash u' \leq^*_\wef z$
    that we have
    $\Gamma, \cas{\Gamma'}{x}{\alpha} \vdash \cas{u'}{x}{\alpha} \leq^*_\wef \cas{z}{x}{\alpha}$.
    We distinguish cases based on the last rule used in the derivation $\Gamma, x \myleq t, \Gamma' \vdash u' \leq^*_\wef z$.
    \begin{description}
      \item[Rule \textsc{Ws-Sub}:] 
      In this case our derivation is $\Gamma, x \myleq t, \Gamma' \vdash u' \leq^*_\wef z$
      with premises $\Gamma, x \myleq t, \Gamma' \vdash u' \leq_\wef z$, $\Gamma, x \myleq t, \Gamma' \vdash u' ~\wef$
      and $\Gamma, x \myleq t, \Gamma' \vdash z ~\wef$.
      In that case, by the result we will prove below, we have
      $\Gamma, \cas{\Gamma'}{x}{\alpha} \vdash \cas{u'}{x}{\alpha} \leq^*_\wef \cas{z}{x}{\alpha}$.
      \item[Rule \textsc{Ws-Trs}:]
      In this case our derivation is $\Gamma, x \myleq t, \Gamma' \vdash u' \leq^*_\wef z$
      with premises $\Gamma, x \myleq t, \Gamma' \vdash u' \leq^*_\wef v$ and
      $\Gamma, x \myleq t, \Gamma' \vdash v \leq^*_\wef v$.
      By induction on these two derivations, we have
      $\Gamma, \cas{\Gamma'}{x}{\alpha} \vdash \cas{u'}{x}{\alpha} \leq^*_\wef \cas{v}{x}{\alpha}$ and
      $\Gamma, \cas{\Gamma'}{x}{\alpha} \vdash \cas{v}{x}{\alpha} \leq^*_\wef \cas{z}{x}{\alpha}$.
      By rule \textsc{Ws-Trs}, we now have $\Gamma, \cas{\Gamma'}{x}{\alpha} \vdash \cas{u'}{x}{\alpha} \leq^*_\wef \cas{z}{x}{\alpha}$.
    \end{description}

    We now prove that for any $\Gamma, x \myleq t, \Gamma' \vdash a \leq_\wef b$,
    subtree of $\Gamma, x \myleq t, \Gamma' \vdash u' \leq^*_\wef z$,
    with both $a$ and $b$ well-formed in $\Gamma, x \myleq t, \Gamma'$,
    then we have $\Gamma, \cas{\Gamma'}{x}{\alpha} \vdash \cas{a}{x}{\alpha} \leq^*_\wef \cas{b}{x}{\alpha}$.
    We first prove an auxiliary result by induction, and then conclude our desired derivation
    $\Gamma, \cas{\Gamma'}{x}{\alpha} \vdash \cas{a}{x}{\alpha} \leq^*_\wef \cas{b}{x}{\alpha}$.

    We now prove that for any $\Gamma, x \myleq t, \Gamma' \vdash a \leq_\wef b$ being a
    subtree of $\Gamma, x \myleq t, \Gamma' \vdash u' \leq^*_\wef z$,
    then there exists three terms $a'$, $b'$ and $c$ such that
    \begin{itemize}
      \item
        $\Gamma, \cas{\Gamma'}{x}{\alpha}; \nil \vdash \cas{b}{x}{\alpha} \mrelTwo{\equiv} \cas{c}{x}{\alpha}$,
      \item
        $\Gamma, \cas{\Gamma'}{x}{\alpha}; \nil \vdash \cas{b'}{x}{\alpha} \mrelTwo{\equiv} \cas{c}{x}{\alpha}$,
      \item
        $\Gamma, \cas{\Gamma'}{x}{\alpha}; \nil \vdash \cas{a}{x}{\alpha} \mrelTwo{\equiv} \cas{a'}{x}{\alpha}$,
      \item
        $\Gamma, \cas{\Gamma'}{x}{\alpha} \vdash \cas{a'}{x}{\alpha} \leq^*_\wef \cas{b'}{x}{\alpha}$ only if $a' \neq b'$.
    \end{itemize}
    It is crucial to work on a subtrees of $\Gamma, x \myleq t, \Gamma' \vdash u' \leq^*_\wef z$
    because we use the main induction hypothesis on well-formed derivation trees.
    Here is a diagram which illustrates the above conditions:
    \begin{center}
      \begin{tikzcd}
        \cas{b}{x}{\alpha} \arrow[r, two heads, "\equiv"]
        & \cas{c}{x}{\alpha} & \\
        & \cas{b'}{x}{\alpha} \arrow[u, two heads, "\equiv"] & \\
        & & \cas{a'}{x}{\alpha} \arrow[ul, dash, dashed, "\leq^*_\wef"] \\
        & & \cas{a}{x}{\alpha} \arrow[u, two heads, "\equiv"]
      \end{tikzcd}
    \end{center}

    We prove this result by induction on $\Gamma, x \myleq t, \Gamma' \vdash a \leq_\wef b$.
    We distinguish based on the last rule used in the derivation of $\Gamma, x \myleq t, \Gamma' \vdash a \leq_\wef b$:
    \begin{description}
      \item[Case \textsc{Ws-Rfl}:]
        In our case, we have $a = b$, and with $a' = b' = c = a$, we have the desired result
        because $\mrelTwo{\equiv}$ is reflexive.

      \item[Case \textsc{Ws-Lf1}:]
        In this case, we have $\Gamma, x \myleq t, \Gamma' \vdash a \leq_\wef b$ with premises
        $\Gamma, x \myleq t, \Gamma' \vdash a_0 \leq_\wef b$, and
        $\Gamma, x \myleq t, \Gamma'; \nil \vdash a \myrel{\equiv} a_0$.
        By using the subinduction, we know there exists three terms $a'$, $b'$ and $c$ such that
        $\Gamma, \cas{\Gamma'}{x}{\alpha}; \nil \vdash \cas{b}{x}{\alpha} \mrelTwo{\equiv} \cas{c}{x}{\alpha}$, and
        $\Gamma, \cas{\Gamma'}{x}{\alpha}; \nil \vdash \cas{b'}{x}{\alpha} \mrelTwo{\equiv} \cas{c}{x}{\alpha}$, and
        $\Gamma, \cas{\Gamma'}{x}{\alpha}; \nil \vdash \cas{a_0}{x}{\alpha} \mrelTwo{\equiv} \cas{a'}{x}{\alpha}$, and
        $\Gamma, \cas{\Gamma'}{x}{\alpha} \vdash \cas{a'}{x}{\alpha} \leq^*_\wef \cas{b'}{x}{\alpha}$ only if $a' \neq b'$.
        By Lemma~\ref{lem:reduction-under-substitution} on $\Gamma, x \myleq t, \Gamma'; \nil \vdash a \myrel{\equiv} a_0$,
        we obtain $\Gamma, \cas{\Gamma'}{x}{\alpha}; \nil \vdash \cas{a}{x}{\alpha} \myrel{\equiv} \cas{a_0}{x}{\alpha}$.
        By transitivity of $\mrelTwo{\equiv}$, we have $\Gamma, \cas{\Gamma'}{x}{\alpha}; \nil \vdash \cas{a}{x}{\alpha} \myrel{\equiv} \cas{a'}{x}{\alpha}$,
        which concludes this case.

      \item[Case \textsc{Ws-Lf2}:]
        In this case, we have \\$\Gamma, x \myleq t, \Gamma' \vdash a \leq_\wef b$ with premises
        $\Gamma, x \myleq t, \Gamma' \vdash a_0 \leq_\wef b$, and
        $\Gamma, x \myleq t, \Gamma'; \nil \vdash a \myrel{\leq} a_0$, and
        $\Gamma, x \myleq t, \Gamma' \vdash a ~\wef$, and
        $\Gamma, x \myleq t, \Gamma' \vdash a_0 ~\wef$.
        By applying the induction hypothesis on $\Gamma, x \myleq t, \Gamma' \vdash a ~\wef$, we obtain $\Gamma, \cas{\Gamma'}{x}{\alpha} \vdash \cas{a}{x}{\alpha} ~\wef$.
        By applying the induction hypothesis on $\Gamma, x \myleq t, \Gamma' \vdash a_0 ~\wef$, we obtain $\Gamma, \cas{\Gamma'}{x}{\alpha} \vdash \cas{a_0}{x}{\alpha} ~\wef$.
        By using the subinduction, we know there exists three terms $a'$, $b'$ and $c$ such that
        $\Gamma, \cas{\Gamma'}{x}{\alpha}; \nil \vdash \cas{b}{x}{\alpha} \mrelTwo{\equiv} \cas{c}{x}{\alpha}$, and
        $\Gamma, \cas{\Gamma'}{x}{\alpha}; \nil \vdash \cas{b'}{x}{\alpha} \mrelTwo{\equiv} \cas{c}{x}{\alpha}$, and
        $\Gamma, \cas{\Gamma'}{x}{\alpha}; \nil \vdash \cas{a_0}{x}{\alpha} \mrelTwo{\equiv} \cas{a'}{x}{\alpha}$, and
        $\Gamma, \cas{\Gamma'}{x}{\alpha} \vdash \cas{a'}{x}{\alpha} \leq^*_\wef \cas{b'}{x}{\alpha}$ only if $a' \neq b'$.
        By Lemma~\ref{lem:subtyping-under-substitution} on $\Gamma, x \myleq t, \Gamma'; \nil \vdash a \myrel{\leq} a_0$,
        we obtain $\Gamma, \cas{\Gamma'}{x}{\alpha} \vdash \cas{a}{x}{\alpha} \leq^*_\wef \cas{a_0}{x}{\alpha}$.

        We do a case disjunction on if $a' = b'$ or not.

        \paragraph*{If $a' = b'$:}
        If $a' = b'$, then
        the derivation of our assumption \\$\Gamma, \cas{\Gamma'}{x}{\alpha}; \nil \vdash \cas{b'}{x}{\alpha} \mrelTwo{\equiv} \cas{c}{x}{\alpha}$
        is actually $\Gamma, \cas{\Gamma'}{x}{\alpha}; \nil \vdash \cas{a'}{x}{\alpha} \mrelTwo{\equiv} \cas{c}{x}{\alpha}$.
        Therefore we have
        $\Gamma, \cas{\Gamma'}{x}{\alpha}; \nil \vdash \cas{a_0}{x}{\alpha} \mrelTwo{\equiv} \cas{c}{x}{\alpha}$
        by transitivity of $\mrelTwo{\equiv}$ from
        $\Gamma, \cas{\Gamma'}{x}{\alpha}; \nil \vdash \cas{a_0}{x}{\alpha} \mrelTwo{\equiv} \cas{a'}{x}{\alpha}$ and
        $\Gamma, \cas{\Gamma'}{x}{\alpha}; \nil \vdash \cas{a'}{x}{\alpha} \mrelTwo{\equiv} \cas{c}{x}{\alpha}$.
        then picking $a' = a$, $b' = a_0$, $c = c$ and $b = b$, we have the desired result.
        Indeed we have:
        \begin{itemize}
          \item
          $\Gamma, \cas{\Gamma'}{x}{\alpha}; \nil \vdash \cas{b}{x}{\alpha} \mrelTwo{\equiv} \cas{c}{x}{\alpha}$,
          which we obtained with our subinduction.
          \item
          $\Gamma, \cas{\Gamma'}{x}{\alpha}; \nil \vdash \cas{a_0}{x}{\alpha} \mrelTwo{\equiv} \cas{c}{x}{\alpha}$,
          which we obtained above.
          \item
          $\Gamma, \cas{\Gamma'}{x}{\alpha}; \nil \vdash \cas{a}{x}{\alpha} \mrelTwo{\equiv} \cas{a}{x}{\alpha}$.
          By reflexivity on $\mrelTwo{\equiv}$.
          \item
          $\Gamma, \cas{\Gamma'}{x}{\alpha} \vdash \cas{a}{x}{\alpha} \leq^*_\wef \cas{a_0}{x}{\alpha}$ only if $a \neq a_0$.
          We prove above that we have $\Gamma, \cas{\Gamma'}{x}{\alpha} \vdash \cas{a}{x}{\alpha} \leq^*_\wef \cas{a_0}{x}{\alpha}$.
        \end{itemize}

        \paragraph*{Else,}
        Now, we will transform \\$\Gamma, \cas{\Gamma'}{x}{\alpha}; \nil \vdash \cas{a_0}{x}{\alpha} \mrelTwo{\equiv} \cas{a'}{x}{\alpha}$
        into $\Gamma, \cas{\Gamma'}{x}{\alpha} \vdash \cas{a_0}{x}{\alpha} \leq^*_\wef \cas{a'}{x}{\alpha}$.
        First, by rule \textsc{Ws-Rfl}, we have
        $\Gamma, \cas{\Gamma'}{x}{\alpha} \vdash \cas{a'}{x}{\alpha} \leq_\wef \cas{a'}{x}{\alpha}$.
        By induction on the number of derivation of $\Gamma, \cas{\Gamma'}{x}{\alpha}; \nil \vdash \cas{a_0}{x}{\alpha} \mrelTwo{\equiv} \cas{a'}{x}{\alpha}$,
        we apply the rule \textsc{Ws-Lf1} at each step to transform
        $\Gamma, \cas{\Gamma'}{x}{\alpha} \vdash \cas{a'}{x}{\alpha} \leq_\wef \cas{a'}{x}{\alpha}$ into
        $\Gamma, \cas{\Gamma'}{x}{\alpha} \vdash \cas{a_0}{x}{\alpha} \leq_\wef \cas{a'}{x}{\alpha}$.
        By Lemma~\ref{prop:transitivity-preserves-well-formedness}, we know that
        both $a_0$ and $a'$ are well-formed in context $\Gamma, \cas{\Gamma'}{x}{\alpha}$ from respectively
        $\Gamma, \cas{\Gamma'}{x}{\alpha} \vdash \cas{a}{x}{\alpha} \leq^*_\wef \cas{a_0}{x}{\alpha}$ and
        $\Gamma, \cas{\Gamma'}{x}{\alpha} \vdash \cas{a'}{x}{\alpha} \leq^*_\wef \cas{b'}{x}{\alpha}$.
        Hence, by rule \textsc{Ws-Sub} $\Gamma, \cas{\Gamma'}{x}{\alpha} \vdash \cas{a_0}{x}{\alpha} \leq^*_\wef \cas{a'}{x}{\alpha}$.

        By transitivity (rule \textsc{Ws-Trs}), we know have $\Gamma, \cas{\Gamma'}{x}{\alpha} \vdash \cas{a}{x}{\alpha} \leq^*_\wef \cas{a'}{x}{\alpha}$,
        and then by transitivity again, we have $\Gamma, \cas{\Gamma'}{x}{\alpha} \vdash \cas{a}{x}{\alpha} \leq^*_\wef \cas{b'}{x}{\alpha}$.
        
        This effectively concludes the case, because we have\\
        $\Gamma, \cas{\Gamma'}{x}{\alpha}; \nil \vdash \cas{a}{x}{\alpha} \mrelTwo{\equiv} \cas{a}{x}{\alpha}$ by reflexivity,
        $\Gamma, \cas{\Gamma'}{x}{\alpha} \vdash \cas{a}{x}{\alpha} \leq^*_\wef \cas{b}{x}{\alpha}$, and
        $\Gamma, \cas{\Gamma'}{x}{\alpha}; \nil \vdash \cas{b}{x}{\alpha} \mrelTwo{\equiv} \cas{b'}{x}{\alpha}$.

      \item[Case \textsc{Ws-Rgh}:]
        In this case, we have $\Gamma, x \myleq t, \Gamma' \vdash a \leq_\wef b$ with premises
        $\Gamma, x \myleq t, \Gamma' \vdash a \leq_\wef b_0$, and
        $\Gamma, x \myleq t, \Gamma'; \nil \vdash b \myrel{\equiv} b_0$.
        By using the subinduction, we know there exists three terms $a'$, $b'$ and $c$ such that
        $\Gamma, \cas{\Gamma'}{x}{\alpha}; \nil \vdash \cas{b_0}{x}{\alpha} \mrelTwo{\equiv} \cas{c}{x}{\alpha}$, and
        $\Gamma, \cas{\Gamma'}{x}{\alpha}; \nil \vdash \cas{b'}{x}{\alpha} \mrelTwo{\equiv} \cas{c}{x}{\alpha}$, and
        $\Gamma, \cas{\Gamma'}{x}{\alpha}; \nil \vdash \cas{a}{x}{\alpha} \mrelTwo{\equiv} \cas{a'}{x}{\alpha}$, and
        $\Gamma, \cas{\Gamma'}{x}{\alpha} \vdash \cas{a'}{x}{\alpha} \leq^*_\wef \cas{b'}{x}{\alpha}$ only if $a' \neq b'$.
        By Lemma~\ref{lem:reduction-under-substitution} on $\Gamma, x \myleq t, \Gamma'; \nil \vdash b \myrel{\equiv} b_0$
        we obtain $\Gamma, \cas{\Gamma'}{x}{\alpha}; \nil \vdash \cas{b}{x}{\alpha} \myrel{\equiv} \cas{b_0}{x}{\alpha}$.
        Hence by transitivity of $\mrelTwo{\equiv}$, we have $\Gamma, \cas{\Gamma'}{x}{\alpha}; \nil \vdash \cas{b}{x}{\alpha} \myrel{\equiv} \cas{c}{x}{\alpha}$.
        This concludes this case.
    \end{description}
    Therefore, we obtained that for any sub-derivation $\Gamma, x \myleq t, \Gamma' \vdash a \leq_\wef b$
    of $\Gamma, x \myleq t, \Gamma' \vdash u' \leq^*_\wef z$,
    there exists three terms $a'$, $b'$ and $c$ such that
    $\Gamma, \cas{\Gamma'}{x}{\alpha}; \nil \vdash \cas{b}{x}{\alpha} \mrelTwo{\equiv} \cas{c}{x}{\alpha}$, and
    $\Gamma, \cas{\Gamma'}{x}{\alpha}; \nil \vdash \cas{b'}{x}{\alpha} \mrelTwo{\equiv} \cas{c}{x}{\alpha}$, and
    $\Gamma, \cas{\Gamma'}{x}{\alpha}; \nil \vdash \cas{a}{x}{\alpha} \mrelTwo{\equiv} \cas{a'}{x}{\alpha}$, and
    $\Gamma, \cas{\Gamma'}{x}{\alpha}; \nil \vdash \cas{a'}{x}{\alpha} \leq^*_\wef \cas{b'}{x}{\alpha}$.
    We also know that both $a$ and $b$ are well-formed in $\Gamma, x \myleq t, \Gamma'$.
    Let's finish by proving that we have $\Gamma, \cas{\Gamma'}{x}{\alpha} \vdash \cas{a}{x}{\alpha} \leq^*_\wef \cas{b}{x}{\alpha}$.

    By using the induction hypothesis, we obtain that both $\cas{a}{x}{\alpha}$ and $\cas{b}{x}{\alpha}$
    are well-formed in context $\Gamma, \cas{\Gamma'}{x}{\alpha}$.
    By Proposition~\ref{prop:transitivity-preserves-well-formedness}
    applied to $\Gamma, \cas{\Gamma'}{x}{\alpha} \vdash \cas{a'}{x}{\alpha} \leq^*_\wef \cas{b'}{x}{\alpha}$
    yields that both $\cas{a'}{x}{\alpha}$ and $\cas{b'}{x}{\alpha}$ are well-formed in context $\Gamma, \cas{\Gamma'}{x}{\alpha}$.

    We now build \\$\Gamma, \cas{\Gamma'}{x}{\alpha} \vdash \cas{a}{x}{\alpha} \leq^*_\wef \cas{a'}{x}{\alpha}$.
    By rule \textsc{Ws-Rfl}, we have $\Gamma, \cas{\Gamma'}{x}{\alpha} \vdash \cas{a'}{x}{\alpha} \leq_\wef \cas{a'}{x}{\alpha}$.
    By induction on the number of derivation of $\Gamma, \cas{\Gamma'}{x}{\alpha}; \nil \vdash \cas{a}{x}{\alpha} \mrelTwo{\equiv} \cas{a'}{x}{\alpha}$,
    we apply the rule \textsc{Ws-Lf1} at each step to transform
    $\Gamma, \cas{\Gamma'}{x}{\alpha} \vdash \cas{a'}{x}{\alpha} \leq_\wef \cas{a'}{x}{\alpha}$ into
    $\Gamma, \cas{\Gamma'}{x}{\alpha} \vdash \cas{a}{x}{\alpha} \leq_\wef \cas{a'}{x}{\alpha}$.
    Because both terms are well-formed, by rule \textsc{Ws-Sub}, we have
    $\Gamma, \cas{\Gamma'}{x}{\alpha} \vdash \cas{a}{x}{\alpha} \leq^*_\wef \cas{a'}{x}{\alpha}$.

    By a similar reasoning, and because both $\cas{b}{x}{\alpha}$ and $\cas{b'}{x}{\alpha}$
    are well-formed in $\Gamma, \cas{\Gamma'}{x}{\alpha}$ we have by a similar reasoning
    $\Gamma, \cas{\Gamma'}{x}{\alpha} \vdash \cas{b'}{x}{\alpha} \leq^*_\wef \cas{b}{x}{\alpha}$.

    Finally, from all above, by transitivity (rule \textsc{Ws-Trs}), we have
    $\Gamma, \cas{\Gamma'}{x}{\alpha} \vdash \cas{a}{x}{\alpha} \leq^*_\wef \cas{b}{x}{\alpha}$ the desired result.

    Let's now prove that $\Gamma, \cas{\Gamma'}{x}{\alpha} \vdash \cas{v}{x}{\alpha} \leq^*_\wef \cas{(\lambda y \myleq z. \T)}{x}{\alpha}$ from
    $\Gamma, x \myleq t, \Gamma' \vdash v \leq^*_\wef \lambda y \myleq z. \T$.
    As above, we first proceed by induction on the number of transitivity step in $\leq^*_\wef$,
    and prove each step by a subinduction as above which we omit it here,
    because the previous reasoning do not depend on the structure of the terms $u'$ and $z$, the same
    reasoning can be made on $\Gamma, x \myleq t, \Gamma' \vdash v \leq^*_\wef \lambda y \myleq z. \T$.
    So we obtain
    Therefore, we have $\Gamma, \cas{\Gamma'}{x}{\alpha} \vdash \cas{v}{x}{\alpha} \leq^*_\wef \cas{(\lambda y \myleq z. \T)}{x}{\alpha}$.
    
    Finally, we have
    $\Gamma, \cas{\Gamma'}{x}{\alpha} \vdash \cas{v}{x}{\alpha} \, \cas{u'}{x}{\alpha} ~\wef$ by rule \textsc{Wf-App}.
  \qedhere
  \end{description}
\end{proof}

% Fin de la liste des theoremes du texte

\begin{lemma}
  \label{lem:subtyping-under-substitution}
  Let $\Gamma; s$ be an extended context, and $\Gamma'$ be a logical context,
  such that $\Gamma, \cas{\Gamma'}{x}{\alpha}; \nil$ is prevalid.
  Let $u$, $v$, $t$ and $\alpha$ be terms, such that
  both $\cas{u}{x}{\alpha}$ and $\cas{v}{x}{\alpha}$ are well-formed in context $\Gamma, \cas{\Gamma'}{x}{\alpha}$.
  
  If $\Gamma, x \myleq t, \Gamma'; \nil \vdash u \myrel{\leq} v$,
  and $\Gamma \vdash \alpha \leq^*_{\wef} t$.
  
  Then $\Gamma, \cas{\Gamma'}{x}{\alpha} \vdash \cas{u}{x}{\alpha} \leq^*_\wef \cas{v}{x}{\alpha}$.
\end{lemma}
\begin{proof}
  In this lemma, we have two cases, either the derivation $\Gamma, x \myleq t, \Gamma'; \nil \vdash u \myrel{\leq} v$
  depending on if our derivation is of the form $\Gamma, x \myleq t, \Gamma'; \nil \vdash \Po[x] \myrel{\leq} \Po[t]$ for some covariant context $\Po$.

  \begin{description}
    \item[If it is:]
  % We apply Lemma~\ref{lem:subtyping-under-substitution-aux1} to obtain
  % $\Gamma, x \myleq t, \cas{\Gamma'}{x}{\alpha}; \nil \vdash \Po_{\cas{}{x}{\alpha}}[x] \myrel{\leq} \Po_{\cas{}{x}{\alpha}}[t]$.
  By assumption, we have both $\cas{u}{x}{\alpha}$ and $\cas{v}{x}{\alpha}$ well-formed in context $\Gamma, \cas{\Gamma'}{x}{\alpha}$,
  because the derivation $\Gamma, x \myleq t, \Gamma'; \nil \vdash u \myrel{\leq} v$
  is $\Gamma, x \myleq t, \Gamma'; \nil \vdash \Po[x] \myrel{\leq} \Po[t]$,
  these terms are $\Po_{\cas{}{x}{\alpha}}[\alpha]$ and $\Po_{\cas{}{x}{\alpha}}[t]$.
  From $\Gamma \vdash \alpha \leq^*_{\wef} t$ we have by weakening (Lemma~\ref{lem:context-weakening})
  $\Gamma, x \myleq t, \cas{\Gamma'}{x}{\alpha} \vdash \alpha \leq^*_\wef t$.
  We therefore have $\Gamma, x \myleq t, \cas{\Gamma'}{x}{\alpha} \vdash \alpha \leq^*_\wef t$,
  and both $\Po_{\cas{}{x}{\alpha}}[\alpha]$ and $\Po_{\cas{}{x}{\alpha}}[t]$ well-formed in $\Gamma, x \myleq t, \cas{\Gamma'}{x}{\alpha}$.
  Hence by Conjecture~\ref{conj:stack-elimination}, we obtain
  $\Gamma, x \myleq t, \cas{\Gamma'}{x}{\alpha} \vdash \Po_{\cas{}{x}{\alpha}}[\alpha] \leq^*_\wef \Po_{\cas{}{x}{\alpha}}[t]$.
  Because there is no instance of the variable $x$ in neither $\alpha$ nor in $t$,
  this last derivation is
  $\Gamma, x \myleq t, \cas{\Gamma'}{x}{\alpha} \vdash \cas{u}{x}{\alpha} \leq^*_\wef \cas{v}{x}{\alpha}$.
  By prevalidity of $\Gamma, x \myleq t, \cas{\Gamma'}{x}{\alpha}$, we know that there is no instance of the variable $x$
  in $\Gamma$. As a result, there is no instance of this variable anywhere.
  As a result, we have
  $\Gamma, \cas{\Gamma'}{x}{\alpha} \vdash \cas{u}{x}{\alpha} \leq^*_\wef \cas{v}{x}{\alpha}$ the desired result.

    \item[Else:]
  We apply Lemma~\ref{lem:subtyping-under-substitution-aux2} to obtain
  $\Gamma, \cas{\Gamma'}{x}{\alpha}; \nil \vdash \cas{u}{x}{\alpha} \myrel{\leq} \cas{v}{x}{\alpha}$.
  By rule \textsc{Ws-Rfl}, we have
  $\Gamma, \cas{\Gamma'}{x}{\alpha} \vdash \cas{v}{x}{\alpha} \leq_\wef \cas{v}{x}{\alpha}$.
  Hence, by rule \textsc{Ws-Lft},
  $\Gamma, \cas{\Gamma'}{x}{\alpha} \vdash \cas{u}{x}{\alpha} \leq_\wef \cas{v}{x}{\alpha}$.
  Because both terms are well-formed by assumption, we have by rule \textsc{Ws-Sub},
  $\Gamma, \cas{\Gamma'}{x}{\alpha} \vdash \cas{u}{x}{\alpha} \leq^*_\wef \cas{v}{x}{\alpha}$
  the desired result.
  \end{description}
\end{proof}

\begin{lemma}[Inversion lemma.]
  \label{lem:inversion-lemma}
  Let $\Gamma$ be a logical context.
  Let $\lambda x\myleq t.u$ and $\lambda x\myleq t'.u'$ be terms.
  If $\Gamma \vdash(\lambda x\myleq t.u)\leq^*_\wef(\lambda x\myleq t'.u')$ then
  $\Gamma \vdash t \equiv_\wef t'$.
\end{lemma}
\begin{proof}
  By induction on the number of transitivity step of
  $\Gamma \vdash(\lambda x\myleq t.u)\leq^*_\wef(\lambda x\myleq t'.u')$,
  we obtain $\Gamma \vdash(\lambda x\myleq t.u)\leq^*(\lambda x\myleq t'.u')$
  by applying Proposition~\ref{prop:trans-well-to-trans} on each of these transitivity steps.
  By transitivity elimination (Lemma~\ref{lem:algorithmic-transitivity-elimination}), we then obtain
  $\Gamma; \nil \vdash (\lambda x\myleq t.u)\leq (\lambda x\myleq t'.u')$.
  By definition, there exists a term $z$ such that
  $\Gamma; \nil \vdash (\lambda x\myleq t.u)\mrelTwo{\equiv} z$ and
  $\Gamma; \nil \vdash (\lambda x\myleq t'.u') \mrelTwo{\leq} z$.
  On the first derivation, based on the structure of the term $\lambda x\myleq t.u$,
  only the rule \textsc{Me-Fun} can be applied, and therefore there exists a function
  $\lambda x\myleq t_z.u_z$ such that $z = \lambda x\myleq t_z.u_z$.
  Because only the rule \textsc{Me-Fun} has been applied, we also have
  $\Gamma; \nil \vdash t \mrelTwo{\equiv} t_z$.
  We apply the same reasoning on the derivation $\Gamma; \nil \vdash (\lambda x\myleq t'.u') \mrelTwo{\leq} z$,
  which consists only of use of the rule \textsc{Ms-Fun} or \textsc{Ms-Equ} with \textsc{Me-Fun},
  and therefore we have as well $\Gamma; \nil \vdash t' \mrelTwo{\equiv} t_z$.
  
  By rule \textsc{Ws-Rfl}, we have $\Gamma \vdash t_z \equiv_\wef t_z$.
  We now show that we have $\Gamma \vdash t \equiv_\wef t_z$ by induction on the number of derivation steps of
  $\Gamma; \nil \vdash t \mrelTwo{\equiv} t_z$. If there is no derivation step, then we have $t = t_z$ and the result holds.
  Otherwise, there exists a term $t'$ such that $\Gamma; \nil \vdash t \myrel{\equiv} t'$
  and $\Gamma; \nil \vdash t' \mrelTwo{\equiv} t_z$. By the induction hypothesis on the number of derivation steps of
  $\Gamma; \nil \vdash t' \mrelTwo{\equiv} t_z$, we obtain $\Gamma \vdash t' \equiv_\wef t_z$.
  By rule \textsc{Ws-Lf1} on $\Gamma; \nil \vdash t \myrel{\equiv} t'$ and $\Gamma \vdash t' \equiv_\wef t_z$,
  we obtain $\Gamma \vdash t \equiv_\wef t_z$.
  By the same reasoning on the number of derivation steps of $\Gamma; \nil \vdash t' \mrelTwo{\equiv} t_z$,
  we obtain $\Gamma \vdash t \equiv t'$ from $\Gamma \vdash t \equiv_\wef t_z$ and the rule \textsc{Ws-Rgh}.
\end{proof}

\begin{theorem}[No supertype of Top.]
  \label{lem:no-supertype-top}
  Let $\Gamma;s$ be an extended context.
  Let $\lambda x\myleq t.u$ be a term. We cannot
  have $\Gamma;s\vdash\T\leq^*\lambda x\myleq t.u$.
\end{theorem}
\begin{proof}
  First we transform $\Gamma;s\vdash\T\leq^*\lambda x\myleq t.u$
  into $\Gamma;s\vdash\T\leq\lambda x\myleq t.u$ by transitivity elimination (Lemma~\ref{lem:algorithmic-transitivity-elimination}).
  Assume towards a contradiction that
  $\Gamma;s\vdash\T\leq\lambda x\myleq t.u$. By definition of the subtyping relation, there exists $w$
  such that $\Gamma;s\vdash\T\mrelTwo{\leq} w$ and $\Gamma;s\vdash\lambda x\myleq t.u\mrelTwo{\equiv}w$.
  From $\Gamma;s\vdash\T\mrelTwo{\leq} w$, we obtain that $w=\T$
  since $\T$ can only promote to itself, with rules \textsc{Ms-Equ} and
  \textsc{Ms-Top}. Now from $\Gamma;s\vdash\lambda x\myleq t.u\mrelTwo{\equiv}w$, we know that $w$ is an
  abstraction since the only rule for equivalence reduction that can be applied to
  $\lambda x\myleq t.u$ is rule \textsc{Ms-Fun} or \textsc{Ms-FOp}. But this contradicts the $w=\T$ established
  earlier and we are done.
\end{proof}

%%% Local Variables:
%%% fill-column: 90
%%% require-final-newline: t
%%% mode-require-final-newline: t
%%% next-line-add-newlines: nil
%%% show-trailing-whitespaces: t
%%% indent-tabs-mode: nil
%%% ispell-dictionary: "british"
%%% mode: latex
%%% TeX-master: "main"
%%% TeX-PDF-mode: t
%%% End:

\section{Rest of the appendix}

% Preservation

\begin{proposition}[Well-formedness extraction]
  \label{prop:transitivity-preserves-well-formedness}
  Let $\Gamma$ be an extended context and $u$ and $v$ be terms. If
  $\Gamma\vdash u\leq^*_{\wef}v$ then both derivation $\Gamma\vdash u~\wef$ and
  $\Gamma\vdash v~\wef$ exist in the derivation tree $\Gamma\vdash u\leq^*_{\wef}v$.
\end{proposition}
\begin{proof}
  Note: the requirement that the derivation trees ($\Gamma \vdash u~\wef$ and $\Gamma \vdash v~\wef$)
  exist in the original derivation tree is used on other lemmas because we do induction on derivation trees.
  It is not enough to know that $\Gamma \vdash u~\wef$ and $\Gamma \vdash v~\wef$ exist,
  they need to be a subtree of our original tree. Hence this formulation.

  By induction on the derivation tree of $\Gamma \vdash u\leq^*_{\wef}v$. If the last
  rule applied is \textsc{Ws-Sub}, then the result holds from the assumption of the rule.
  Else, the last rule applied is \textsc{Wf-Trans}, then there exists a term $t$ such that
  $\Gamma \vdash u\leq^*_{\wef}t$ and $\Gamma \vdash t\leq^*_{\wef}v$, and the result
  holds by the induction hypothesis.    
\end{proof}

\begin{proposition}[From well-subtyping to subtyping.]
  \label{prop:trans-well-to-trans}
  Let $\Gamma$ be a logical context and $u$ and $v$ be terms. If
  $\Gamma \vdash u\leq_{\wef} v$ then $\Gamma; \nil \vdash u\leq v$.
\end{proposition}
\begin{proof}
  By induction on the derivation tree of $\Gamma \vdash u\leq_{\wef} v$.
  We distinguish cases based on the last rule used in the derivation tree.
  \begin{description}
    \item[Rule \textsc{Ws-Rfl}:]
      In this case, we have $u = v = t$.
      By Proposition~\ref{prop:algorithmic-refl}, we have $\Gamma; \nil \vdash t \myrel{\equiv} t$.
      By rule \textsc{Ms-Equ}, we have $\Gamma; \nil \vdash t \myrel{\leq} t$.
      By definition of the subtype relation $\leq$, we have $\Gamma; \nil \vdash t \leq t$.

    \item[Rule \textsc{Ws-Lf1} or \textsc{Ws-Lf2}:]
      In this case, we have as premises $\Gamma \vdash v' \leq_{\wef} t$,
      $\Gamma; \nil \vdash v \myrel{\vartriangleleft} v'$.
      By induction hypothesis, we obtain $\Gamma; \nil \vdash v' \leq t$.
      By definition of $\leq$, there exists $z$ such that $\Gamma; \nil \vdash v' \mrelTwo{\leq} z$ and $\Gamma; \nil \vdash t \mrelTwo{\equiv} z$.
      The reduction $\Gamma; \nil \vdash v \myrel{\vartriangleleft} v'$ is
      either $\Gamma; \nil \vdash v \myrel{\equiv} v'$ or $\Gamma; \nil \vdash v \myrel{\leq} v'$.
      In the former case, by rule \textsc{Ms-Equ}, we have $\Gamma; \nil \vdash v \myrel{\leq} v'$.
      Hence we have $\Gamma; \nil \vdash v \myrel{\leq} v'$.
      Hence $\Gamma; \nil \vdash v \mrelTwo{\leq} v'$.
      We now have $\Gamma; \nil \vdash v \mrelTwo{\leq} v'$ and $\Gamma; \nil \vdash v' \mrelTwo{\leq} z$, hence $\Gamma; \nil \vdash v \mrelTwo{\leq} z$.
      By definition of $\leq$, we have $\Gamma; \nil \vdash v \leq t$.

    \item[Rule \textsc{Ws-Rgh}:]
      In this case, we have as premises $\Gamma \vdash v \leq_{\wef} t'$,
      $\Gamma; \nil \vdash t \myrel{\equiv} t'$.
      By induction hypothesis, we obtain $\Gamma; \nil \vdash v \leq t'$.
      By definition of $\leq$, there exists $z$ such that $\Gamma; \nil \vdash v \mrelTwo{\leq} z$ and $\Gamma; \nil \vdash t' \mrelTwo{\equiv} z$.
      By definition of $\myrel{\equiv}$, we have $\Gamma; \nil \vdash t \myrel{\leq} t'$ and $\Gamma; \nil \vdash t' \myrel{\leq} t$.
      Hence $\Gamma; \nil \vdash t \mrelTwo{\leq} t'$.
      We now have $\Gamma; \nil \vdash t' \mrelTwo{\equiv} z$, hence $\Gamma; \nil \vdash t' \mrelTwo{\leq} z$ and  $\Gamma; \nil \vdash z \mrelTwo{\leq} t'$.
      We have $\Gamma; \nil \vdash t \mrelTwo{\leq} t'$ and $\Gamma; \nil \vdash t' \mrelTwo{\leq} z$, hence $\Gamma; \nil \vdash t \mrelTwo{\leq} z$.
      By definition of $\leq$, we have $\Gamma; \nil \vdash v \leq t$.
  \end{description}
\end{proof}

\begin{proposition}[From well-equivalence to equivalence.]
  \label{prop:transequiv-well-to-trans}
  Let $\Gamma$ be a logical context and $u$ and $v$ be terms. If
  $\Gamma \vdash u\equiv_{\wef} v$ then $\Gamma; \nil \vdash u\equiv v$.
\end{proposition}
\begin{proof}
  By induction on the derivation tree of $\Gamma \vdash u\equiv_{\wef} v$.
  We distinguish cases based on the last rule used in the derivation tree.
  \begin{description}
    \item[Rule \textsc{Ws-Rfl}:]
      In this case, we have $u = v = t$.
      By Proposition~\ref{prop:algorithmic-refl}, we have $\Gamma; \nil \vdash t \myrel{\equiv} t$.
      By rule \textsc{Ms-Equ}, we have $\Gamma; \nil \vdash t \myrel{\equiv} t$.
      By definition of the subtype relation $\equiv$, we have $\Gamma; \nil \vdash t \equiv t$.

    \item[Rule \textsc{Ws-Lf1}:]
      In this case, we have as premises $\Gamma \vdash v' \equiv_{\wef} t$,
      $\Gamma; \nil \vdash v \myrel{\equiv} v'$.
      By induction hypothesis, we obtain $\Gamma; \nil \vdash v' \equiv t$.
      By definition of $\equiv$, there exists $z$ such that $\Gamma; \nil \vdash v' \mrelTwo{\equiv} z$ and $\Gamma; \nil \vdash t \mrelTwo{\equiv} z$.
      From $\Gamma; \nil \vdash v \myrel{\equiv} v'$ we have $\Gamma; \nil \vdash v \mrelTwo{\equiv} v'$.
      We now have $\Gamma; \nil \vdash v \mrelTwo{\equiv} v'$ and $\Gamma; \nil \vdash v' \mrelTwo{\equiv} z$, hence $\Gamma; \nil \vdash v \mrelTwo{\equiv} z$.
      By definition of $\equiv$, we have $\Gamma; \nil \vdash v \equiv t$.

    \item[Rule \textsc{Ws-Rgh}:]
      In this case, we have as premises $\Gamma \vdash v \equiv_{\wef} t'$,
      $\Gamma; \nil \vdash t \myrel{\equiv} t'$.
      By induction hypothesis, we obtain $\Gamma; \nil \vdash v \equiv t'$.
      By definition of $\equiv$, there exists $z$ such that $\Gamma; \nil \vdash v \mrelTwo{\equiv} z$ and $\Gamma; \nil \vdash t' \mrelTwo{\equiv} z$.
      By definition of $\myrel{\equiv}$, we have $\Gamma; \nil \vdash t \myrel{\equiv} t'$ and $\Gamma; \nil \vdash t' \myrel{\equiv} t$.
      Hence $\Gamma; \nil \vdash t \mrelTwo{\equiv} t'$.
      We now have $\Gamma; \nil \vdash t' \mrelTwo{\equiv} z$, hence $\Gamma; \nil \vdash t' \mrelTwo{\equiv} z$ and  $\Gamma; \nil \vdash z \mrelTwo{\equiv} t'$.
      We have $\Gamma; \nil \vdash t \mrelTwo{\equiv} t'$ and $\Gamma; \nil \vdash t' \mrelTwo{\equiv} z$, hence $\Gamma; \nil \vdash t \mrelTwo{\equiv} z$.
      By definition of $\equiv$, we have $\Gamma; \nil \vdash v \equiv t$.
  \end{description}
\end{proof}

\begin{lemma}
  \label{lem:equiv-commutative}
  Let $\Gamma$ be a logical context and $u$ and $v$ be terms. If
  $\Gamma \vdash u \equiv_{\wef} v$ then $\Gamma \vdash v \equiv_{\wef} u$.
\end{lemma}
\begin{proof}
  By induction on the derivation tree of $\Gamma \vdash u \equiv_{\wef} v$.
  We distinguish cases based on the last rule used in the derivation tree.
  \begin{description}
    \item[Rule \textsc{Ws-Rfl}:]
      In this case, we have $u = v$.
      By rule \textsc{Ws-Rfl}, we have $\Gamma \vdash u \equiv_{\wef} u$.

    \item[Rule \textsc{Ws-Lf1}:]
      In this case, we have as premises $\Gamma \vdash u' \equiv_{\wef} v$,
      $\Gamma; \nil \vdash u \myrel{\equiv} u'$.
      By induction hypothesis, we obtain $\Gamma \vdash v \equiv_{\wef} u'$.
      By rule \textsc{Ws-Rgh}, we obtain $\Gamma \vdash v \equiv_{\wef} u$.

    \item[Rule \textsc{Ws-Rgh}:]
      In this case, we have as premises $\Gamma \vdash u \equiv_{\wef} v'$,
      $\Gamma; \nil \vdash v \myrel{\equiv} v'$.
      By induction hypothesis, we obtain $\Gamma \vdash v' \equiv_{\wef} u$.
      By rule \textsc{Ws-Lf1}, we obtain $\Gamma \vdash v \equiv_{\wef} u$.
  \end{description}
\end{proof}

\begin{lemma}
  \label{lem:equiv-to-subtyping}
  Let $\Gamma$ be a logical context and $u$ and $v$ be terms. If
  $\Gamma \vdash u \equiv_{\wef} v$ then $\Gamma \vdash u \leq_{\wef} v$.
\end{lemma}
\begin{proof}
  By induction on the derivation tree of $\Gamma \vdash u \equiv_{\wef} v$.
  We distinguish cases based on the last rule used in the derivation tree.
  \begin{description}
    \item[Rule \textsc{Ws-Rfl}:]
      In this case, we have $u = v$.
      By rule \textsc{Ws-Rfl}, we have $\Gamma \vdash u \leq_{\wef} u$.

    \item[Rule \textsc{Ws-Lf1}:]
      In this case, we have as premises $\Gamma \vdash u' \equiv_{\wef} v$,
      $\Gamma; \nil \vdash u \myrel{\equiv} u'$.
      By induction hypothesis, we obtain $\Gamma \vdash u' \leq_{\wef} v$.
      Hence by rule \textsc{Ws-Lf1}, we obtain $\Gamma \vdash u \leq_{\wef} v$.

    \item[Rule \textsc{Ws-Rgh}:]
      In this case, we have as premises $\Gamma \vdash u \equiv_{\wef} v'$,
      $\Gamma; \nil \vdash v \myrel{\equiv} v'$.
      By induction hypothesis, we obtain $\Gamma \vdash u \leq_{\wef} v'$.
      Hence by rule \textsc{Ws-Rgh}, we obtain $\Gamma \vdash u \leq_{\wef} v$.
  \end{description}
\end{proof}

\begin{proposition}[From reduction semantics to equivalence reduction.]
  \label{prop:mapsto-inclusion}
  Let $u$ and $v$ be terms.
  If $u \mapsto v$ then $\Gamma; s \vdash u \myrel{\equiv} v$ for all extended context $\Gamma; s$.
\end{proposition}
\begin{proof}
  By induction on the derivation tree of $u \mapsto v$.  The base case is when the last
  rule used is \textsc{Os-Bet}, and the result holds by rule \textsc{Me-Bet} and
  reflexivity of $\myrel{\equiv}$ (Proposition~\ref{prop:algorithmic-refl}).

  The induction cases is when the last rule used is \textsc{Os-Con}, and the result holds
  by induction and rules \textsc{Me-Fun}, \textsc{Me-FOp} and \textsc{Me-App}.
\end{proof}

\begin{proposition}[Reflexivity.]
  \label{prop:algorithmic-refl}
  Let $\Gamma; s$ be an extended context, and $u$ be a term,
  we have $\Gamma; s \vdash u \myrel{\leq} u$, and $\Gamma; s \vdash u \myrel{\equiv} u$.
\end{proposition}
\begin{proof}
  By straightforward induction on the derivation trees of
  $\Gamma;s\vdash u\myrel{\equiv}u$, noting that the
  only applicable rules at the leaves are \textsc{Me-Top} and \textsc{Me-Var},
  and that at the rest of the nodes only rules
  \textsc{Me-App}, \textsc{Me-Fun}, and \textsc{Me-FOp} occur.
  To now prove $\Gamma;s\vdash u\myrel{\leq}u$, we use rule \textsc{Ms-Equ}.
\end{proof}

% Reduction preserves wf

\begin{lemma}[Weakening of context]
  \label{lem:context-weakening}
  Let $\Gamma; s$ and $\Gamma'; s'$ be two extended contexts such that $\Gamma, \Gamma'; s @ s'$ is prevalid.
  Let $u$ and $v$ be terms.
  We have the following relations:
  
  If $\Gamma \vdash u ~\wef$ then $\Gamma, \Gamma' \vdash u ~\wef$.

  If $\Gamma \vdash u \leq_{\wef} v$ then $\Gamma, \Gamma' \vdash u \leq_{\wef} v$.

  If $\Gamma \vdash u \leq^*_{\wef} v$ then $\Gamma, \Gamma' \vdash u \leq^*_{\wef} v$.

  If $\Gamma; s \vdash u \leq v$ then $\Gamma, \Gamma'; s \vdash u \leq v$.

  If $\Gamma; s \vdash u \leq^* v$ then $\Gamma, \Gamma'; s \vdash u \leq^* v$.

  If $\Gamma; s \vdash u \myrel{\leq} v$ then $\Gamma, \Gamma'; s \vdash u \myrel{\leq} v$.

  If $\Gamma; s \vdash u \mrelTwo{\leq} v$ then $\Gamma, \Gamma'; s \vdash u \mrelTwo{\leq} v$.
  
  If $\Gamma; s \vdash u \myrel{\equiv} v$ then $\Gamma, \Gamma'; s @ s' \vdash u \myrel{\equiv} v$.
  
  If $\Gamma; s \vdash u \mrelTwo{\equiv} v$ then $\Gamma, \Gamma'; s @ s' \vdash u \mrelTwo{\equiv} v$.
\end{lemma}
\begin{proof}
  Note: this lemma is a group of different lemmas which are proved separately.
  This proof dispatches the proof of each of these lemmas to their respective proof.

  We do a case disjunction based on the case we are in:
    \begin{description}
      \item[Case $\Gamma \vdash u ~\wef$ or $\Gamma \vdash u \leq_{\wef} v$ or $\Gamma \vdash u \leq^*_{\wef} v$:]
        Done in Lemma~\ref{lem:context-weakening-aux1}.

      \item[Case $\Gamma; s \vdash u \leq v$:]
        We need to prove that if $\Gamma; s \vdash u \leq v$ then $\Gamma, \Gamma'; s \vdash u \leq v$.
        By definition of the relation $\leq$, there exists a term $z$ such that
        $\Gamma; s \vdash u \mrelTwo{\leq} z$ and $\Gamma; s \vdash v \mrelTwo{\equiv} z$.
        By using respectively Lemma~\ref{lem:context-weakening-subtyping-reduction} and Lemma~\ref{lem:context-weakening-equivalence-reduction}
        on respectively $\Gamma; s \vdash u \mrelTwo{\leq} z$ and $\Gamma; s \vdash v \mrelTwo{\equiv} z$,
        we respectively obtain $\Gamma, \Gamma'; s \vdash u \mrelTwo{\leq} z$ and $\Gamma, \Gamma'; s @ s' \vdash v \mrelTwo{\equiv} z$.
        By definition of the relation $\leq$, we obtain $\Gamma, \Gamma'; s \vdash u \leq v$.

      \item[Case $\Gamma; s \vdash u \leq^* v$:]
        We need to prove that if $\Gamma; s \vdash u \leq^* v$ then $\Gamma, \Gamma'; s \vdash u \leq^* v$.
        By definition of the relation $\leq^*$,
        there exists $u_1, \dots, u_n$ such that
        $\Gamma; s \vdash u \leq u_1$, $\Gamma; s \vdash u_1 \leq u_2, \cdots, \Gamma; s \vdash u_n \leq v$.
        By using the reasoning of the case above on each $\Gamma; s \vdash u_i \leq u_{i+1}$, we obtain
        $\Gamma, \Gamma'; s \vdash u \leq u_1$, $\Gamma, \Gamma'; s \vdash u_1 \leq u_2, \cdots, \Gamma, \Gamma'; s \vdash u_n \leq v$.
        By definition of the relation $\leq^*$, we obtain $\Gamma, \Gamma'; s \vdash u \leq^* v$.

      \item[Case $\Gamma; s \vdash u \myrel{\leq} v$:]
        We need to prove that if $\Gamma; s \vdash u \myrel{\leq} v$ then $\Gamma, \Gamma'; s \vdash u \myrel{\leq} v$.
        We proceed by induction on the derivation tree of $\Gamma; s \vdash u \myrel{\leq} v$.
        The proof can be found in Lemma~\ref{lem:context-weakening-subtyping-reduction}.

      \item[Case $\Gamma; s \vdash u \mrelTwo{\leq} v$:]
        We need to prove that if $\Gamma; s \vdash u \mrelTwo{\leq} v$ then $\Gamma, \Gamma'; s \vdash u \mrelTwo{\leq} v$.
        By definition of the relation $\mrelTwo{\leq}$, there exists $u_1, \dots, u_n$ such that
        $\Gamma; s \vdash u \myrel{\leq} u_1$, $\Gamma; s \vdash u_1 \myrel{\leq} u_2, \cdots, \Gamma; s \vdash u_n \myrel{\leq} v$.
        By using Lemma~\ref{lem:context-weakening-subtyping-reduction} multiple times, we obtain
        $\Gamma, \Gamma'; s \vdash u \myrel{\leq} u_1$, $\Gamma, \Gamma'; s \vdash u_1 \myrel{\leq} u_2, \cdots, \Gamma, \Gamma'; s \vdash u_n \myrel{\leq} v$.
        By definition of the relation $\mrelTwo{\leq}$, we obtain $\Gamma, \Gamma'; s \vdash u \mrelTwo{\leq} v$.

      \item[Case $\Gamma; s \vdash u \myrel{\equiv} v$:]
        We need to prove that if $\Gamma; s \vdash u \myrel{\equiv} v$ then $\Gamma, \Gamma'; s @ s' \vdash u \myrel{\equiv} v$.
        We proceed by induction on the derivation tree of $\Gamma; s \vdash u \myrel{\equiv} v$.
        The proof can be found in Lemma~\ref{lem:context-weakening-equivalence-reduction}.

      \item[Case $\Gamma; s \vdash u \mrelTwo{\equiv} v$:]
        We need to prove that if $\Gamma; s \vdash u \mrelTwo{\equiv} v$ then $\Gamma, \Gamma'; s @ s' \vdash u \mrelTwo{\equiv} v$.
        By definition of the relation $\mrelTwo{\equiv}$, there exists $u_1, \dots, u_n$ such that
        $\Gamma; s \vdash u \myrel{\equiv} u_1$, $\Gamma; s \vdash u_1 \myrel{\equiv} u_2, \cdots, \Gamma; s \vdash u_n \myrel{\equiv} v$.
        By using Lemma~\ref{lem:context-weakening-equivalence-reduction} multiple times, we obtain
        $\Gamma, \Gamma'; s @ s' \vdash u \myrel{\equiv} u_1$, $\Gamma, \Gamma'; s @ s' \vdash u_1 \myrel{\equiv} u_2, \cdots, \Gamma, \Gamma'; s @ s' \vdash u_n \myrel{\equiv} v$.
        By definition of the relation $\mrelTwo{\equiv}$, we obtain $\Gamma, \Gamma'; s @ s' \vdash u \mrelTwo{\equiv} v$.
    \end{description}
\end{proof}

\begin{lemma}[Weakening of context - Aux]
  \label{lem:context-weakening-aux1}
  Let $\Gamma$ and $\Gamma'$ be two logical contexts such that $\Gamma, \Gamma'$ is prevalid.
  Let $u$ and $v$ be terms.
  We have the following relations:

  \begin{enumerate}
    \item 
    If $\Gamma \vdash u ~\wef$ then $\Gamma, \Gamma' \vdash u ~\wef$.

    \item 
    If $\Gamma \vdash u \leq_{\wef} v$ then $\Gamma, \Gamma' \vdash u \leq_{\wef} v$.
    
    \item 
    If $\Gamma \vdash u \leq^*_{\wef} v$ then $\Gamma, \Gamma' \vdash u \leq^*_{\wef} v$.
  \end{enumerate}
\end{lemma}
\begin{proof}
  We need to prove all of these relations together, because any derivation tree of any of these relations
  will contain a derivation tree of the other relations.
  Consider for instance $\Gamma \vdash (\lambda x \myleq \T. x) \, \T ~\wef$.
  One of the premise of this derivation tree is $\Gamma \vdash \T \leq^*_{\wef} \T$,
  which will ultimately have as premise $\Gamma \vdash \T \leq_{\wef} \T$ and $\Gamma \vdash \T ~\wef$.
  The proof therefore consider all of these relations together.
  The proof works by induction on the derivation tree,
  for the example above for instance, $\Gamma \vdash \T ~\wef$ is considered a base case, and therefore
  we have $\Gamma, \Gamma' \vdash \T ~\wef$, from which we will be able to conclude $\Gamma, \Gamma' \vdash \T \leq_{\wef} \T$,
  and so on.

  Assume any of the cases of this lemma.
  We do an induction on the size of the derivation tree of either
  $(\Gamma \vdash u ~\wef)_1$, $(\Gamma \vdash u \leq_{\wef} v)_2$ or $(\Gamma \vdash u \leq^*_{\wef} v)_3$
  which we index by $i$ for $i \in \{1, 2, 3\}$ as a indexed family,
  one of the three possible cases of this lemma.
  The induction hypothesis can be called on any element of this family, as long as the size of its derivation tree
  is smaller.
  For instance, consider a derivation tree ending by the rule \textsc{Wf-App}
  (for instance$(\Gamma \vdash u \, v ~\wef)_1$), its premises are
  derivations of the form $(\Gamma \vdash v \leq^*_{\wef} w)_3$, on which we can apply the induction hypothesis
  because their size are smaller despite not being in the same index.

  Suppose we are now in the case $1$, that is $(\Gamma \vdash u ~\wef)_1$, let's prove $(\Gamma, \Gamma' \vdash u ~\wef)_1$.
  We distinguish cases based on the last rule used in the derivation tree $(\Gamma \vdash u ~\wef)_1$.
  \begin{description}
    \item[Rule \textsc{Wf-PrS} or rule \textsc{Wf-PrE} with $u = x$:]
      In this case, we have $x \myvartriangleleft t \in \Gamma$,
      and therefore $x \myvartriangleleft t \in \Gamma \cup \Gamma'$.
      By assumption, we have $\Gamma, \Gamma'; s @ s'$ prevalid, hence $\Gamma, \Gamma'$ prevalid.
      Hence by rule \textsc{Wf-PrS} or rule \textsc{Wf-PrE}, we obtain $\Gamma, \Gamma' \vdash x ~\wef$.
    \item[Rule \textsc{Wf-Top}:]
      In this case, we have as premises $\Gamma ~\prevalid$.
      By assumption, we have $\Gamma, \Gamma'; s @ s'$ prevalid, hence $\Gamma, \Gamma'$ prevalid.
      We now have $\Gamma, \Gamma' \vdash \T ~\wef$ by rule \textsc{Wf-Top}.
    \item[Rule \textsc{Wf-Fun} with $u = \lambda x \myleq a. b$:]
      In this case, we have as premises $(\Gamma \vdash a ~\wef)_1$ and $(\Gamma, x \myleq a \vdash b ~\wef)_1$.
      Because $x$ is a fresh variable, we can always alpha convert it to another name that does not exist already,
      and as such the variable $x$ is not used in $\Gamma'$: no annotation of $x$ exist in $\Gamma'$,
      and all the annotations of $\Gamma'$ does not reference $x$.
      By induction hypothesis on the subtree $(\Gamma \vdash a ~\wef)_1$ of our assumption $(\Gamma \vdash u ~\wef)_1$,
      we obtain $(\Gamma, \Gamma' \vdash a ~\wef)_1$.
      By induction hypothesis on the subtree $(\Gamma, x \myleq a \vdash b ~\wef)_1$ of our assumption $(\Gamma \vdash u ~\wef)_1$,
      we obtain $(\Gamma, x \myleq a, \Gamma' \vdash b ~\wef)_1$.
      Because $x$ does not exist in $\Gamma'$, we have $(\Gamma, \Gamma', x \myleq a \vdash b ~\wef)_1$.
      By rule \textsc{Wf-Fun}, we obtain $\Gamma, \Gamma' \vdash \lambda x \myleq a. b ~\wef$.
    \item[Rule \textsc{Wf-App} with $u = a \, b$:]
      In this case, we have as premises $(\Gamma \vdash a \leq^*_{\wef} \lambda x \myleq w. \T)_3$ and $(\Gamma \vdash b \leq^*_{\wef} w)_3$.
      By induction hypothesis on the subtree $(\Gamma \vdash a \leq^*_{\wef} \lambda x \myleq w. \T)_3$ of our assumption $(\Gamma \vdash u ~\wef)_1$,
      we obtain $(\Gamma, \Gamma' \vdash a \leq^*_{\wef} \lambda x \myleq w. \T)_3$.
      By induction hypothesis on the subtree $(\Gamma \vdash b \leq^*_{\wef} w)_3$ of our assumption $(\Gamma \vdash u ~\wef)_1$,
      we obtain $(\Gamma, \Gamma' \vdash b \leq^*_{\wef} w)_3$.
      By rule \textsc{Wf-App}, we obtain $\Gamma, \Gamma' \vdash a \, b ~\wef$.
  \end{description}

  \medskip

  Suppose we are now in the case $2$, that is $(\Gamma \vdash u \leq_{\wef} v)_2$, let's prove $(\Gamma, \Gamma' \vdash u \leq_{\wef} v)_2$.
  We distinguish cases based on the last rule used in the derivation tree $(\Gamma \vdash u \leq_{\wef} v)_2$.
  \begin{description}
    \item[Rule \textsc{Ws-Rfl}:]
      In this case, we have $u = v = t$.
      By using the induction hypothesis on our assumption $(\Gamma \vdash u \leq_{\wef} v)_2$, we obtain $\Gamma, \Gamma' \vdash t \leq_{\wef} t$.

    \item[Rule \textsc{Ws-Lf1}:]
      In this case, we have as premises $(\Gamma \vdash v' \vartriangleleft_{\wef} t)_2$,
      $(\Gamma; \nil \vdash v \myrel{\equiv} v')_i$.
      By induction hypothesis on the subtree $(\Gamma \vdash v' \vartriangleleft_{\wef} t)_2$ of our assumption $(\Gamma \vdash u \leq_{\wef} v)_2$,
      we obtain $(\Gamma, \Gamma' \vdash v' \vartriangleleft_{\wef} t)_2$.
      By using Lemma~\ref{lem:context-weakening-equivalence-reduction}
      we obtain $\Gamma, \Gamma'; \nil \vdash v \myrel{\equiv} v'$
      from $\Gamma; \nil \vdash v \myrel{\equiv} v'$.
      By rule \textsc{Ws-Lf1}, we obtain $\Gamma, \Gamma' \vdash v \vartriangleleft_{\wef} t$.

    \item[Rule \textsc{Ws-Lf2}:]
      In this case, we have as premises $(\Gamma \vdash v' \leq_{\wef} t)_2$,
      $(\Gamma; \nil \vdash v \myrel{\leq} v')_i$, $(\Gamma \vdash v ~\wef)_1$ and $(\Gamma \vdash v' ~\wef)_1$.
      By induction hypothesis on the subtree $(\Gamma \vdash v' \leq_{\wef} t)_2$ of our assumption $(\Gamma \vdash u \leq_{\wef} v)_2$,
      we obtain $(\Gamma, \Gamma' \vdash v' \leq_{\wef} t)_2$.
      By using the induction hypothesis on the subtree $(\Gamma \vdash v ~\wef)_1$ of our assumption $(\Gamma \vdash u \leq_{\wef} v)_2$,
      we obtain $(\Gamma, \Gamma' \vdash v ~\wef)_1$.
      By using the induction hypothesis on the subtree $(\Gamma \vdash v' ~\wef)_1$ of our assumption $(\Gamma \vdash u \leq_{\wef} v)_2$,
      we obtain $(\Gamma, \Gamma' \vdash v' ~\wef)_1$.
      By using Lemma~\ref{lem:context-weakening-subtyping-reduction},
      we obtain $\Gamma, \Gamma'; \nil \vdash v \myrel{\leq} v'$
      from $\Gamma; \nil \vdash v \myrel{\leq} v'$.
      By rule \textsc{Ws-Lf2}, we obtain $\Gamma, \Gamma' \vdash v \leq_{\wef} t$.

    \item[Rule \textsc{Ws-Rgh}:]
      In this case, we have as premises $(\Gamma \vdash v \vartriangleleft_{\wef} t')_2$,
      $(\Gamma; \nil \vdash t \myrel{\equiv} t')_i$.
      By induction hypothesis on the subtree $(\Gamma \vdash v \vartriangleleft_{\wef} t')_2$ of our assumption $(\Gamma \vdash u \vartriangleleft_{\wef} v)_2$,
      we obtain $(\Gamma, \Gamma' \vdash v \vartriangleleft_{\wef} t')_2$.
      By using Lemma~\ref{lem:context-weakening-equivalence-reduction}, we obtain $\Gamma, \Gamma'; \nil \vdash t \myrel{\equiv} t'$.
      By rule \textsc{Ws-Rgh}, we obtain $\Gamma, \Gamma' \vdash v \vartriangleleft_{\wef} t$.
  \end{description}

  \medskip

  Suppose we are now in the case $3$, that is $(\Gamma \vdash u \leq^*_{\wef} v)_3$, let's prove $(\Gamma, \Gamma' \vdash u \leq^*_{\wef} v)_3$.
  We distinguish cases based on the last rule used in the derivation tree $(\Gamma \vdash u \leq^*_{\wef} v)_3$.
  \begin{description}
    \item[Rule \textsc{Ws-Sub}:]
      In this case, we have as premises $(\Gamma \vdash u \leq_{\wef} v)_2$.
      By using the induction hypothesis on the subtree $(\Gamma \vdash u \leq_{\wef} v)_2$ of our assumption $(\Gamma \vdash u \leq^*_{\wef} v)_3$,
      we obtain $(\Gamma, \Gamma' \vdash u \leq_{\wef} v)_2$.
      By rule \textsc{Ws-Sub}, we obtain $\Gamma, \Gamma' \vdash u \leq^*_{\wef} v$.

    \item[Rule \textsc{Ws-Trs}:]
      In this case, we have as premises $(\Gamma \vdash u \leq^*_{\wef} w)_3$ and $(\Gamma \vdash w \leq^*_{\wef} v)_3$.
      By induction hypothesis on the subtree $(\Gamma \vdash u \leq^*_{\wef} w)_3$ of our assumption $(\Gamma \vdash u \leq^*_{\wef} v)_3$,
      we obtain $(\Gamma, \Gamma' \vdash u \leq^*_{\wef} w)_3$.
      By induction hypothesis on the subtree $(\Gamma \vdash w \leq^*_{\wef} v)_3$ of our assumption $(\Gamma \vdash u \leq^*_{\wef} v)_3$,
      we obtain $(\Gamma, \Gamma' \vdash w \leq^*_{\wef} v)_3$.
      By rule \textsc{Ws-Trs}, we obtain $\Gamma, \Gamma' \vdash u \leq^*_{\wef} v$.
  \end{description}
\end{proof}

\begin{lemma}[Weakening of context - Subtyping reduction]
  \label{lem:context-weakening-subtyping-reduction}
  Let $\Gamma; s$ and $\Gamma'$ be two contexts such that $\Gamma, \Gamma'; s$ is prevalid.
  Let $u$ and $v$ be terms.
  We have the following relations:
  
  If $\Gamma; s \vdash u \myrel{\leq} v$ then $\Gamma, \Gamma'; s \vdash u \myrel{\leq} v$.
\end{lemma}
\begin{proof}
  By induction on the derivation tree of $\Gamma; s \vdash u \myrel{\leq} v$.
  We distinguish based on the last rule used.
  \begin{description}
    \item[Rule \textsc{Ms-Pro}:]
      In this case, we have $u = x$ and $v = t$ with $x \myleq t \in \Gamma$.
      Because $x \myleq t \in \Gamma$, we have $x \myleq t \in \Gamma \cup \Gamma'$.
      By assumption, we have $\Gamma, \Gamma'; s @ s'$ prevalid, hence $\Gamma, \Gamma'; s$ prevalid.
      Hence by rule \textsc{Ms-Pro}, we obtain $\Gamma, \Gamma'; s \vdash x \myrel{\leq} t$.

    \item[Rule \textsc{Ms-Top}:]
      In this case, we have $\Gamma; s \vdash u \myrel{\leq} \T$ and
      the derivation has no premises other than the prevalidity of the context.
      By assumption, we know that $\Gamma, \Gamma'; s$ is prevalid.
      Hence, by rule \textsc{Ms-Top}, we obtain $\Gamma, \Gamma'; s \vdash u \myrel{\leq} \T$.

    \item[Rule \textsc{Ms-Equ}:]
      In this case, we have $\Gamma; s \vdash u \myrel{\leq} v$ with premise
      $\Gamma; s \vdash u \myrel{\equiv} v$.
      By Lemma~\ref{lem:context-weakening-equivalence-reduction}, we obtain $\Gamma, \Gamma'; s \vdash u \myrel{\equiv} v$.
      Hence, by rule \textsc{Ms-Equ}, we conclude that $\Gamma, \Gamma'; s \vdash u \myrel{\leq} v$.

    \item[Rule \textsc{Ms-App}:]
      In this case, we have $u = a\,b$ and $v = a'\,b$ with premise
      $\Gamma; b :: s \vdash a \myrel{\leq} a'$.
      By the induction hypothesis, we obtain $\Gamma, \Gamma'; b :: s \vdash a \myrel{\leq} a'$.
      Hence, by rule \textsc{Ms-App}, we conclude that $\Gamma, \Gamma'; s \vdash a\,b \myrel{\leq} a'\,b$.

    \item[Rule \textsc{Ms-FOp}:]
      In this case, $u = \lambda x \myleq t . a$, $v = \lambda x \myleq t . a'$ and $s = \alpha :: s_0$
      with premise $\Gamma, x \myequiv \alpha; s_0 \vdash a \myrel{\leq} a'$.
      By the induction hypothesis we obtain $\Gamma, \Gamma', x \myequiv \alpha; s_0 \vdash a \myrel{\leq} a'$.
      From the fact that $x$ is fresh
      (hence can be renamed so that it never appears in $\Gamma'$), we obtain
      $\Gamma, \Gamma', x \myequiv \alpha; s_0 \vdash a \myrel{\leq} a'$.
      Therefore, by rule \textsc{Ms-FOp}, we deduce that
      $\Gamma, \Gamma'; \alpha :: s_0 \vdash \lambda x \myleq t . a \myrel{\leq} \lambda x \myleq t . a'$.

    \item[Rule \textsc{Ms-Fun}:]
      In this case, $u = \lambda x \myleq t . a$, $v = \lambda x \myleq t . a'$ and $s = \nil$
      with premise $\Gamma, x \myleq t; \nil \vdash a \myrel{\leq} a'$.
      By the induction hypothesis we obtain $\Gamma, \Gamma', x \myleq t; \nil \vdash a \myrel{\leq} a'$.
      From the fact that $x$ is fresh
      (hence can be renamed so that it never appears in $\Gamma'$), we obtain
      $\Gamma, \Gamma', x \myleq t; \nil \vdash a \myrel{\leq} a'$.
      Therefore, by rule \textsc{Ms-Fun}, we deduce that
      $\Gamma, \Gamma'; \nil \vdash \lambda x \myleq t . a \myrel{\leq} \lambda x \myleq t . a'$.
  \end{description}
\end{proof}

\begin{lemma}[Weakening of context - Equivalence reduction]
  \label{lem:context-weakening-equivalence-reduction}
  Let $\Gamma; s$ and $\Gamma'; s'$ be two extended contexts such that $\Gamma, \Gamma'; s @ s'$ is prevalid.
  Let $u$ and $v$ be terms.
  We have the following relations:

  If $\Gamma; s \vdash u \myrel{\equiv} v$ then $\Gamma, \Gamma'; s @ s' \vdash u \myrel{\equiv} v$.
\end{lemma}
\begin{proof}
  By induction on the derivation tree of $\Gamma; s \vdash u \myrel{\equiv} v$.
  We distinguish based on the last rule used.
  \begin{description}
    \item[Rule \textsc{Me-Var}:]  
      In this case, we have $u = x$ and $v = x$.
      By assumption, we have $\Gamma, \Gamma'; s @ s'$ prevalid.
      Hence by rule \textsc{Me-Var}, we obtain $\Gamma, \Gamma'; s @ s' \vdash x \myrel{\equiv} x$.

    \item[Rule \textsc{Me-Pro}:]
      In this case, we have $u = x$ and $v = \alpha'$ with $x \myequiv \alpha' \in \Gamma$
      and $\Gamma; s \vdash \alpha \myrel{\equiv} \alpha'$.
      By induction, we have $\Gamma, \Gamma'; s @ s' \vdash \alpha \myrel{\equiv} \alpha'$.
      Because $x \myequiv \alpha \in \Gamma$, we have $x \myequiv \alpha \in \Gamma \cup \Gamma'$.
      By assumption, we have $\Gamma, \Gamma'; s @ s'$ prevalid.
      Hence by rule \textsc{Me-Pro}, we obtain $\Gamma, \Gamma'; s @ s' \vdash x \myrel{\equiv} \alpha'$.

    \item[Rule \textsc{Me-App}:]
      In this case, we have $u = a\,b$ and $v = a'\,b'$ with premises
      $\Gamma; b::s \vdash a \myrel{\equiv} a'$ and $\Gamma; \nil \vdash b \myrel{\equiv} b'$.
      By the induction hypothesis, we obtain
      $\Gamma, \Gamma'; b :: s @ s' \vdash a \myrel{\equiv} a'$ and
      $\Gamma, \Gamma'; \nil \vdash b \myrel{\equiv} b'$.
      Hence, by rule \textsc{Me-App}, we conclude that
      $\Gamma, \Gamma'; s @ s' \vdash a\,b \myrel{\equiv} a'\,b'$.

    \item[Rule \textsc{Me-Top}:]  
      In this case, we have $\Gamma; s \vdash \T \myrel{\equiv} \T$ and 
      the derivation has no premises other than the prevalidity of the context.
      By assumption, we know that $\Gamma, \Gamma'; s @ s'$ is prevalid.
      Hence, by rule \textsc{Me-Top}, we obtain $\Gamma, \Gamma'; s @ s' \vdash \T \myrel{\equiv} \T$.
      
    \item[Rule \textsc{Me-Fun}:]
      In this case, $u = \lambda x \myleq a. b$, $v = \lambda x \myleq a'. b'$ and $s = \nil$
      with premises $\Gamma; \nil \vdash a \myrel{\equiv} a'$ and
      $\Gamma, x \myleq a; \nil \vdash b \myrel{\equiv} b'$. 
      By the induction hypothesis we obtain
      $\Gamma, \Gamma'; \nil \vdash a \myrel{\equiv} a'$ and
      $\Gamma, x \myleq a, \Gamma'; \nil \vdash b \myrel{\equiv} b'$.
      From the fact that $x$ is fresh
      (hence can be renamed so that it never appears in $\Gamma'$), we obtain
      $\Gamma, \Gamma', x \myleq a; \nil \vdash b \myrel{\equiv} b'$.
      Therefore, by rule \textsc{Me-Fun}, we deduce that
      $\Gamma, \Gamma'; \nil \vdash \lambda x \myleq a. b \myrel{\equiv} \lambda x \myleq a'. b'$.
      
    \item[Rule \textsc{Me-FOp}:]
      In this case, $u = \lambda x \myleq a. b$, $v = \lambda x \myleq a'. b'$ and $s = \alpha :: s_0$
      with premises $\Gamma; \nil \vdash a \myrel{\equiv} a'$ and
      $\Gamma; x \myequiv \alpha; s_0 \vdash b \myrel{\equiv} b'$.
      By the induction hypothesis we obtain
      $\Gamma, \Gamma'; \nil \vdash a \myrel{\equiv} a'$ and
      $\Gamma, x \myequiv \alpha, \Gamma'; s_0 @ s' \vdash b \myrel{\equiv} b'$.
      From the fact that $x$ is fresh
      (hence can be renamed so that it never appears in $\Gamma'$), we obtain
      $\Gamma, \Gamma', x \myequiv \alpha; s_0 @ s' \vdash b \myrel{\equiv} b'$.
      Therefore, by rule \textsc{Me-FOp}, we deduce that
      $\Gamma, \Gamma'; \alpha :: s_0 @ s' \vdash \lambda x \myleq a. b \myrel{\equiv} \lambda x \myleq a'. b'$.

    \item[Rule \textsc{Me-Bet}:]
      In this case, $u = (\lambda x \myleq a. b) \, c$ and $v = \cas{b'}{x}{c'}$
      with premises $\Gamma; s \vdash b \myrel{\equiv} b'$ and
      $\Gamma; \nil \vdash c \myrel{\equiv} c'$.
      By the induction hypothesis we obtain
      $\Gamma, \Gamma'; s @ s' \vdash b \myrel{\equiv} b'$ and
      $\Gamma, \Gamma'; \nil \vdash c \myrel{\equiv} c'$.
      Therefore, by rule \textsc{Me-Bet}, we deduce that
      $\Gamma, \Gamma'; s @ s' \vdash (\lambda x \myleq a. b) \, c \myrel{\equiv} \cas{b'}{x}{c'}$.

    \item[Rule \textsc{Me-TAp}:]
      In this case, $u = \T \, b$ and $v = \T$ with premise $\Gamma; s ~\prevalid$.
      By assumption, we know that $\Gamma, \Gamma'; s @ s'$ is prevalid.
      Hence, by rule \textsc{Me-TAp}, we obtain $\Gamma, \Gamma'; s @ s' \vdash \T \, b \myrel{\equiv} \T$.
  \end{description}
\end{proof}

%% NARROWING OF CONTEXT

\begin{lemma}[Narrowing of context in well-formedness derivation]
  \label{lem:narrowing-context-wf}
  Let $\Gamma$ and $\Gamma'$ be two contexts.
  Let $t$, $t'$, and $u$ be terms.
  If $\Gamma, x \myleq t, \Gamma' \vdash u~\wef$ and $\Gamma; \nil \vdash t \myrel{\equiv} t'$ and $\Gamma \vdash t' ~\wef$,
  then $\Gamma, x \myleq t', \Gamma' \vdash u~\wef$.
\end{lemma}
\begin{proof}
  By induction on the derivation tree of $\Gamma, x \myleq t, \Gamma' \vdash u~\wef$.
  We distinguish cases on the last rule used in the derivation tree.
  \begin{description}
    \item[Rule \textsc{Wf-PrS} or rule \textsc{Wf-PrE} with $u = y$:]
    By assumption, we have $y \myvartriangleleft t_y \in \Gamma, x \myleq t, \Gamma'$ for some $t_y$,
    and $\Gamma, x \myleq t, \Gamma' ~\prevalid$.
    By Lemma~\ref{lem:narrowing-prevalidity}, $\Gamma, x \myleq t', \Gamma' ~\prevalid$.
    Either this subtype annotation is $x \myleq t$, or it's another one.
    In the former case, we have $y = x$,
    and $x \myleq t' \in \Gamma, x \myleq t', \Gamma'$, hence $\Gamma, x \myleq t', \Gamma' \vdash u~\wef$
    by rule \textsc{Wf-PrS}.
    In the latter case, we have $y \myvartriangleleft t_y \in \Gamma \cup \Gamma'$,
    hence $y \myvartriangleleft t_y \in \Gamma, x \myleq t', \Gamma'$
    hence $\Gamma, x \myleq t', \Gamma' \vdash u~\wef$ by rule \textsc{Wf-PrS} or \textsc{Wf-PrE}.

    \item[Rule \textsc{Wf-Top}:]
    By assumption, $u = \T$ and \\$\Gamma, x \myleq t, \Gamma' ~\prevalid$.
    By Lemma~\ref{lem:narrowing-prevalidity}, $\Gamma, x \myleq t', \Gamma' ~\prevalid$.
    Hence, $\Gamma, x \myleq t', \Gamma' \vdash \T ~\wef$ by rule \textsc{Wf-Top}.

    \item[Rule \textsc{Wf-Fun} with $u = \lambda y \myleq v. w$:]
    We have $\Gamma, x \myleq t, \Gamma' \vdash \lambda y \myleq v. w~\wef$ with premise \\$\Gamma, x \myleq t, \Gamma', y \myleq v \vdash w~\wef$.
    By the induction hypothesis, $\Gamma, x \myleq t', \Gamma', y \myleq v \vdash w~\wef$.
    Hence, $\Gamma, x \myleq t', \Gamma' \vdash \lambda y \myleq v. w~\wef$ by \textsc{Wf-Fun}.

    \item[Rule \textsc{Wf-App} with $u = a\,b$:]
    \item[] 
    By assumption, there exist some term $z$ such that
    $\Gamma, x \myleq t, \Gamma' \vdash a \leq^*_\wef \lambda y \myleq z. \T$ and $\Gamma, x \myleq t, \Gamma' \vdash b \leq^*_\wef z$.
    We prove the case $\Gamma, x \myleq t, \Gamma' \vdash a \leq^*_\wef \lambda y \myleq z. \T$, the other one is done by the exact same reasoning.
    We wish to prove $\Gamma, x \myleq t', \Gamma' \vdash a \leq^*_\wef \lambda y \myleq z. \T$
    We first decompose $\Gamma, x \myleq t, \Gamma' \vdash a \leq^*_\wef \lambda y \myleq z. \T$
    into $\Gamma, x \myleq t, \Gamma' \vdash a \leq_\wef a_0 \leq_\wef \cdots \leq_\wef a_n \leq_\wef \lambda y \myleq z. \T$,
    to then prove we can transform each of these step of this derivation
    $\Gamma, x \myleq t, \Gamma' \vdash b \leq_\wef c$ into
    $\Gamma, x \myleq t', \Gamma' \vdash b \leq^*_\wef c$.
    To do so, we do the same reasoning as in Lemma~\ref{lem:substitution-preserves-wf} by using instead Lemmas
    \ref{lem:narrowing-context-subtyping-reduction} and \ref{lem:narrowing-context-equivalence-reduction}.
    Hence by using the rule \textsc{Ws-Trs}, we later obtain $\Gamma, x \myleq t', \Gamma' \vdash a \leq^*_\wef \lambda y \myleq z. \T$.
    By using the same reasoning on $\Gamma, x \myleq t, \Gamma' \vdash b \leq^*_\wef z$,
    we obtain $\Gamma, x \myleq t', \Gamma' \vdash b \leq^*_\wef z$.
    Hence $\Gamma, x \myleq t', \Gamma' \vdash a \, b ~\wef$ by rule \textsc{Wf-App}.
  \end{description}
\end{proof}

\begin{lemma}[Narrowing of context in subtyping reductions]
  \label{lem:narrowing-context-subtyping-reduction}
  Let $\Gamma; s$ be an extended context, let $\Gamma'$ be an additional context.
  Let $u$, $v$, $t$ and $t'$ be terms.
  If $\Gamma, x \myleq t, \Gamma'; \nil \vdash u \myrel{\leq} v$ and $\Gamma; \nil \vdash t \myrel{\equiv} t'$.
  Assume that both $u$, $v$ are well-formed in context $\Gamma, x \myleq t', \Gamma'$.
  Assume also $\Gamma \vdash t' ~\wef$.
  Then there exists a term $v'$ such that we have $\Gamma, x \myleq t', \Gamma'; \nil \vdash u \myrel{\leq} v'$,
  and $\Gamma, x \myleq t', \Gamma'; \nil \vdash v \myrel{\leq} v'$,
  and $\Gamma, x \myleq t', \Gamma' \vdash v' ~\wef$.
\end{lemma}
\begin{proof}
  If the derivation $\Gamma, x \myleq t, \Gamma'; \nil \vdash u \myrel{\leq} v$
  does not make a promotion of $x$ to $t$, then we have
  $\Gamma, x \myleq t', \Gamma'; \nil \vdash u \myrel{\leq} v$
  and $v' = v$ satisfy the lemma.

  Else, if the derivation $\Gamma, x \myleq t, \Gamma'; \nil \vdash u \myrel{\leq} v$
  make a promotion of $x$ to $t$, then we have
  $\Gamma, x \myleq t', \Gamma'; \nil \vdash \Po[x] \myrel{\leq} \Po[t']$ with
  $\Gamma, x \myleq t', \Gamma'; \nil \vdash \Po[t'] \myrel{\equiv} \Po[t]$
  by weakening to adapt the context and obtain the last derivation from $\Gamma; \nil \vdash t \myrel{\equiv} t'$.
  We now need to show that $v'$ is well-formed in context $\Gamma, x \myleq t', \Gamma'$
  by induction on the structure of $v'$.
  \begin{description}
    \item[Case $v' = x$:] A variable is well-formed by rule \textsc{Wf-PrS} or \textsc{Wf-PrE}.

    \item[Case $v' = \T$:]
    We have $\Gamma, x \myleq t', \Gamma' \vdash \T ~\wef$ by rule \textsc{Wf-Top}.

    \item[Case $v' = \lambda y \myleq t'' . w$:]
    We have $\Gamma, x \myleq t', \Gamma' \vdash \lambda y \myleq t'' . w ~\wef$
    by rule \textsc{Wf-Fun} with premise $\Gamma, x \myleq t', \Gamma', y \myleq t'' \vdash w ~\wef$.
    By the induction hypothesis, we have $\Gamma, x \myleq t', \Gamma', y \myleq t'' \vdash w ~\wef$.

    \item[Case $v' = a\,b$:]
    Because $v'$ and $u$ are the same term, except for a promotion of $x$ to $t$ in head position,
    we know in our induction that $u = a'\,b'$ with $b' = b$.
    Because $u$ is well-formed, we have $\Gamma, x \myleq t', \Gamma' \vdash a' \leq^*_\wef \lambda y \myleq w . \T$
    and $\Gamma, x \myleq t', \Gamma' \vdash b' \leq^*_\wef w$.
    We now prove $\Gamma, x \myleq t', \Gamma' \vdash a \leq^*_\wef \lambda y \myleq w . \T$,
    which will allow us to conclude that $v'$ is well-formed by rule \textsc{Wf-App}.
    We have $\Gamma, x \myleq t', \Gamma' \vdash a' \leq^*_\wef \lambda y \myleq w . \T$.
    By using the induction hypothesis on $a$, we know that it's well-formed in context $\Gamma, x \myleq t', \Gamma'$.
    Because $a'$ is $a$ but with a promotion of $x$ to $t$ in head position, we have
    $\Gamma, x \myleq t', \Gamma' \vdash a \leq_\wef a'$ from rule \textsc{Ws-Lf2}.
    Hence we can conclude that $\Gamma, x \myleq t', \Gamma' \vdash a \leq^*_\wef \lambda y \myleq w . \T$.
  \end{description}
\end{proof}

\begin{lemma}[Narrowing of context in equivalence reductions]
    \label{lem:narrowing-context-equivalence-reduction}
    Let $\Gamma; s$ be an extended context, let $\Gamma'$ be an additional context.
    Let $u$, $v$, $t$ and $t'$ be terms.
    If $\Gamma, x \myleq t, \Gamma'; s \vdash u \myrel{\equiv} v$ and $\Gamma; \nil \vdash t \myrel{\equiv} t'$.
    Then we have $\Gamma, x \myleq t', \Gamma' \vdash u \myrel{\equiv} v$.
\end{lemma}
\begin{proof}
  We do an induction on the derivation tree of $\Gamma, x \myleq t, \Gamma'; s \vdash u \myrel{\equiv} v$.
  We distinguish cases on the last rule used.
  \begin{description}
    \item[Rule \textsc{Me-Var}:]
    We have $u = v = y$ some variable.
    The derivation $\Gamma, x \myleq t, \Gamma'; s \vdash u \myrel{\equiv} v$ hence holds by rule \textsc{Me-Var}.
  
    \item[Rule \textsc{Me-Pro}:]
    We have $u = y$ and $v = \alpha'$ with $y \myequiv \alpha \in \Gamma, x \myleq t, \Gamma'$ and $\Gamma, x \myleq t, \Gamma' \vdash \alpha \myrel{\equiv} \alpha'$.
    By prevalidity of the context, we know that $y \neq x$, and as a result we have $y \myequiv \alpha \in \Gamma, \Gamma'$,
    hence $y \myequiv \alpha \in \Gamma, x \myleq t', \Gamma'$.
    By rule \textsc{Me-Pro}, we therefore have $\Gamma, x \myleq t', \Gamma' \vdash y \myrel{\equiv} \alpha'$.

    \item[Rule \textsc{Me-Top}:]
    By assumption, we have \\$\Gamma, x \myleq t, \Gamma' \vdash \T \myrel{\equiv} \T$ with premise\\ $\Gamma, x \myleq t, \Gamma' ~\prevalid$.
    By Lemma~\ref{lem:narrowing-prevalidity}, $\Gamma, x \myleq t', \Gamma'; s ~\prevalid$.\\
    By rule \textsc{Ms-Top}, we have $\Gamma, x \myleq t', \Gamma'; s \vdash \T \myrel{\equiv} \T$.
  
    \item[Rule \textsc{Me-App}:]
    We have $\Gamma, x \myleq t, \Gamma'; s \vdash a\,b \myrel{\equiv} a'\,b'$ with premises $\Gamma, x \myleq t, \Gamma'; b::s \vdash a \myrel{\equiv} a'$ and $\Gamma, x \myleq t, \Gamma'; \nil \vdash b \myrel{\equiv} b'$.
    By the induction hypothesis, $\Gamma, x \myleq t', \Gamma'; b::s \vdash a \myrel{\equiv} a'$ and $\Gamma, x \myleq t', \Gamma'; \nil \vdash b \myrel{\equiv} b'$.
    Hence $\Gamma, x \myleq t', \Gamma' \vdash a\,b \myrel{\equiv} a'\,b'$ by \textsc{Ms-App}.
  
    \item[Rule \textsc{Me-Fun}:]
    We have $\Gamma, x \myleq t, \Gamma'; \nil \vdash \lambda y \myleq w. u \myrel{\equiv} \lambda y \myleq w'. u'$
    with $\Gamma, x \myleq t, \Gamma', y \myleq w; \nil \vdash u \myrel{\equiv} u'$ and $\Gamma, x \myleq t, \Gamma'; \nil \vdash w \myrel{\equiv} w'$.
    By the induction hypothesis, $\Gamma, x \myleq t', \Gamma', y \myleq w; \nil \vdash u \myrel{\equiv} u'$
    and $\Gamma, x \myleq t', \Gamma'; \nil \vdash w \myrel{\equiv} w'$.
    Thus, $\Gamma, x \myleq t', \Gamma' \vdash \lambda y \myleq w. u \myrel{\equiv} \lambda y \myleq w'. u'$ by \textsc{Me-Fun}.
  
    \item[Rule \textsc{Me-FOp}:]
    We have \\$\Gamma, x \myleq t, \Gamma'; \alpha :: s \vdash \lambda y \myleq w. u \myrel{\equiv} \lambda y \myleq w'. u'$
    with premise $\Gamma, x \myleq t, \Gamma', y \equiv \alpha; s \vdash u \myrel{\equiv} u'$ and $\Gamma, x \myleq t, \Gamma'; \nil \vdash w \myrel{\equiv} w'$.
    By the induction hypothesis, $\Gamma, x \myleq t', \Gamma', y \equiv \alpha; s \vdash u \myrel{\equiv} u'$ and $\Gamma, x \myleq t', \Gamma'; \nil \vdash w \myrel{\equiv} w'$.
    Hence $\Gamma, x \myleq t', \Gamma' \vdash \lambda y \myleq w. u \myrel{\equiv} \lambda y \myleq w'. u'$ by \textsc{Me-FOp}.
  
    \item[Rule \textsc{Me-Bet}:]
    We have $\Gamma, x \myleq t, \Gamma'; s \vdash (\lambda y \myleq a. b) \, c \myrel{\equiv} \cas{b'}{y}{c'}$ with premises
    $\Gamma, x \myleq t, \Gamma'; s \vdash b \myrel{\equiv} b'$ and $\Gamma, x \myleq t, \Gamma'; \nil \vdash c \myrel{\equiv} c'$.
    By the induction hypothesis, $\Gamma, x \myleq t', \Gamma'; s \vdash b \myrel{\equiv} b'$ and $\Gamma, x \myleq t', \Gamma'; \nil \vdash c \myrel{\equiv} c'$.
    Hence $\Gamma, x \myleq t', \Gamma' \vdash (\lambda y \myleq a. b) \, c \myrel{\equiv} \cas{b'}{y}{c'}$ by \textsc{Me-Bet}.
  
    \item[Rule \textsc{Me-TAp}:]
    We have $\Gamma, x \myleq t, \Gamma'; s \vdash \T \, b \myrel{\equiv} \T$ with premise $\Gamma, x \myleq t, \Gamma'; s ~\prevalid$.
    By Lemma~\ref{lem:narrowing-prevalidity}, $\Gamma, x \myleq t', \Gamma'; s ~\prevalid$.
    Hence $\Gamma, x \myleq t', \Gamma' \vdash \T \, b \myrel{\equiv} \T$ by \textsc{Me-TAp}.
  \end{description}
\end{proof}

\begin{lemma}[Narrowing prevalidity]
  \label{lem:narrowing-prevalidity}
  Let $\Gamma, x \myvartriangleleft t, \Gamma' s$ be an extended prevalid context,
  and $t'$ such that $\Gamma; s \vdash t \myrel{\equiv} t'$.
  Then $\Gamma, x \myvartriangleleft t', \Gamma'; s$ is prevalid.
\end{lemma}
\begin{proof}
  By induction on the derivation tree of $\Gamma, x \myvartriangleleft t, \Gamma' s ~\prevalid$.
  We distinguish cases on the last rule used.
  \begin{description}
    \item[Rule \textsc{Pv-Emp}:]
      In this case, we have $\Gamma, x \myvartriangleleft t, \Gamma' = \varepsilon$.
      Hence $\Gamma = \varepsilon$ and $\Gamma' = \varepsilon$.
      By assumption, we have $\Gamma; s \vdash t \myrel{\equiv} t'$.
      We need to prove that $\Gamma, x \myvartriangleleft t', \Gamma'; s$ is prevalid,
      which is equivalent to proving that $x \myvartriangleleft t'; s$ is prevalid.
      Because $\Gamma = \varepsilon$, we have $\fv(t') = \emptyset$.
      Hence $x \myvartriangleleft t'; s$ is prevalid.

    \item[Rule \textsc{Pv-Ctx}:]
      In this case, we have $\Gamma, x \myvartriangleleft t, \Gamma' = \Gamma_0, y \myleq t_y$
      with $\Gamma_0 ~\prevalid$ and $y \not\in \dom(\Gamma_0)$ and $\fv(t_y) \subseteq \dom(\Gamma_0)$.
      Either $y \myleq t_y$ is $x \myvartriangleleft t$, or it is in $\Gamma_0$ or $\Gamma'$.

      If $y \myleq t_y$ is $x \myvartriangleleft t$, then $\Gamma_0 = \Gamma$ and $\Gamma' = \varepsilon$.
      By assumption, we have $\Gamma; s \vdash t \myrel{\equiv} t'$.
      We need to prove that $\Gamma, x \myvartriangleleft t', \Gamma'; s$ is prevalid,
      which is equivalent to proving that $\Gamma, x \myvartriangleleft t'; s$ is prevalid.
      Because $\Gamma_0 = \Gamma$, we have $\Gamma ~\prevalid$.
      Because $x \not\in \dom(\Gamma)$, we have $x \not\in \dom(\Gamma)$.
      Because $\fv(t) \subseteq \dom(\Gamma)$ and $\Gamma; s \vdash t \myrel{\equiv} t'$,
      we have $\fv(t') \subseteq \dom(\Gamma)$.
      Hence $\Gamma, x \myvartriangleleft t'; s$ is prevalid.

      If $y \myleq t_y$ is in $\Gamma_0$, then $\Gamma_0 = \Gamma_1, y \myleq t_y$ for some $\Gamma_1$.
      We have $\Gamma_1 ~\prevalid$ and $y \not\in \dom(\Gamma_1)$ and $\fv(t_y) \subseteq \dom(\Gamma_1)$.
      We also have $\Gamma = \Gamma_1, y \myleq t_y$ and $\Gamma' = \varepsilon$.
      By induction hypothesis, $\Gamma_1, x \myvartriangleleft t', \Gamma'; s$ is prevalid.
      Hence $\Gamma_1, x \myvartriangleleft t', \Gamma', y \myleq t_y; s$ is prevalid.

      If $y \myleq t_y$ is in $\Gamma'$, then $\Gamma' = \Gamma_1, y \myleq t_y$ for some $\Gamma_1$.
      We have $\Gamma_1 ~\prevalid$ and $y \not\in \dom(\Gamma_1)$ and $\fv(t_y) \subseteq \dom(\Gamma_1)$.
      We also have $\Gamma = \Gamma$ and $\Gamma' = \Gamma_1, y \myleq t_y$.
      By induction hypothesis, $\Gamma, x \myvartriangleleft t', \Gamma_1; s$ is prevalid.
      Hence $\Gamma, x \myvartriangleleft t', \Gamma_1, y \myleq t_y; s$ is prevalid.

    \item[Rule \textsc{Pv-EqA}:]
      In this case, we have $\Gamma, x \myvartriangleleft t, \Gamma' = \Gamma_0, y \myequiv \alpha$
      with $\Gamma_0 ~\prevalid$ and $y \not\in \dom(\Gamma_0)$ and $\fv(\alpha) \subseteq \dom(\Gamma_0)$.
      Either $y \myequiv \alpha$ is $x \myvartriangleleft t$, or it is in $\Gamma_0$ or $\Gamma'$.

      If $y \myequiv \alpha$ is $x \myvartriangleleft t$, then $\Gamma_0 = \Gamma$ and $\Gamma' = \varepsilon$.
      By assumption, we have $\Gamma; s \vdash t \myrel{\equiv} t'$.
      We need to prove that $\Gamma, x \myvartriangleleft t', \Gamma'; s$ is prevalid,
      which is equivalent to proving that $\Gamma, x \myvartriangleleft t'; s$ is prevalid.
      Because $\Gamma_0 = \Gamma$, we have $\Gamma ~\prevalid$.
      Because $x \not\in \dom(\Gamma)$, we have $x \not\in \dom(\Gamma)$.
      Because $\fv(t) \subseteq \dom(\Gamma)$ and $\Gamma; s \vdash t \myrel{\equiv} t'$,
      we have $\fv(t') \subseteq \dom(\Gamma)$.
      Hence $\Gamma, x \myvartriangleleft t'; s$ is prevalid.

      If $y \myequiv \alpha$ is in $\Gamma_0$, then $\Gamma_0 = \Gamma_1, y \myequiv \alpha$ for some $\Gamma_1$.
      We have $\Gamma_1 ~\prevalid$ and $y \not\in \dom(\Gamma_1)$ and $\fv(\alpha) \subseteq \dom(\Gamma_1)$.
      We also have $\Gamma = \Gamma_1, y \myequiv \alpha$ and $\Gamma' = \varepsilon$.
      By induction hypothesis, $\Gamma_1, x \myvartriangleleft t', \Gamma'; s$ is prevalid.
      Hence $\Gamma_1, x \myvartriangleleft t', \Gamma', y \myequiv \alpha; s$ is prevalid.

      If $y \myequiv \alpha$ is in $\Gamma'$, then $\Gamma' = \Gamma_1, y \myequiv \alpha$ for some $\Gamma_1$.
      We have $\Gamma_1 ~\prevalid$ and $y \not\in \dom(\Gamma_1)$ and $\fv(\alpha) \subseteq \dom(\Gamma_1)$.
      We also have $\Gamma = \Gamma$ and $\Gamma' = \Gamma_1, y \myequiv \alpha$.
      By induction hypothesis, $\Gamma, x \myvartriangleleft t', \Gamma_1; s$ is prevalid.
      Hence $\Gamma, x \myvartriangleleft t', \Gamma_1, y \myequiv \alpha; s$ is prevalid.

    \item[Rule \textsc{Pv-Nil}:]
      In this case, we have $\Gamma, x \myvartriangleleft t, \Gamma'; s = \Gamma, x \myvartriangleleft t, \Gamma'; \nil$
      with $\Gamma, x \myvartriangleleft t, \Gamma' ~\prevalid$.
      By induction hypothesis, $\Gamma, x \myvartriangleleft t', \Gamma' ~\prevalid$.
      Hence $\Gamma, x \myvartriangleleft t', \Gamma'; \nil$ is prevalid.

    \item[Rule \textsc{Pv-Sta}:]
      In this case, we have \\$\Gamma, x \myvartriangleleft t, \Gamma'; s = \Gamma, x \myvartriangleleft t, \Gamma'; \alpha :: s_0$
      with $\Gamma, x \myvartriangleleft t, \Gamma'; s_0 ~\prevalid$ and $\fv(\alpha) \subseteq \dom(\Gamma, x \myvartriangleleft t, \Gamma')$.
      By induction hypothesis, $\Gamma, x \myvartriangleleft t', \Gamma'; s_0 ~\prevalid$.
      Because $\fv(\alpha) \subseteq \dom(\Gamma, x \myvartriangleleft t, \Gamma')$,
      we have $\fv(\alpha) \subseteq \dom(\Gamma, x \myvartriangleleft t', \Gamma')$.
      Hence $\Gamma, x \myvartriangleleft t', \Gamma'; \alpha :: s_0$ is prevalid.
  \end{description}
\end{proof}

%% FIN NARROWING

\begin{proposition}[Reduction preserves subtyping derivation]
  \label{lem:composability-reverse-asleft}
  Let $\Gamma; s$ be an extended context.
  Let $u$, $u'$ and $v$ be terms such that
  $\Gamma \vdash u \leq^*_{\wef} v$ and
  $u \mapsto u'$, and
  $\Gamma \vdash u' ~\wef$.

  Then $\Gamma \vdash u' \leq^*_{\wef} v$.
\end{proposition}
\begin{proof}
  By~\ref{prop:mapsto-inclusion}, we have $\Gamma; \nil \vdash u \myrel{\equiv} u'$.
  By rule \textsc{Ws-Rfl} and \textsc{Ws-Lf1}, we have $\Gamma \vdash u' \leq_{\wef} u$.
  From Proposition~\ref{prop:transitivity-preserves-well-formedness}, we have $\Gamma \vdash u ~\wef$
  from $\Gamma \vdash u \leq^*_{\wef} v$.
  Hence rule \textsc{Ws-Sub}, we have $\Gamma \vdash u' \leq^*_{\wef} u$.
  By \textsc{Ws-Trs}, we then have $\Gamma \vdash u' \leq^*_{\wef} v$.
\end{proof}

% Substitution preserves wf

% \valentin{todo relire tout ce qui suit}

\begin{lemma}[Substitution preserves prevalidity]
  \label{lem:substitution-preserves-prevalidity}
  Let $\Gamma,x \myvartriangleleft t,\Gamma'; s$ be an extended context.

  If $\Gamma,x \myvartriangleleft t,\Gamma'; s$ is prevalid, then $\Gamma, \cas{\Gamma'}{x}{t}; \cas{s}{x}{t}$ is prevalid.

  If $\Gamma,x \myvartriangleleft t,\Gamma'$ is prevalid, then $\Gamma, \cas{\Gamma'}{x}{t}$ is prevalid.
\end{lemma}
\begin{proof}
  We prove the first result, the second is similar.
  We now do an induction on the derivation tree $\Gamma,x \myvartriangleleft t,\Gamma'; s$ prevalid.
  We distinguish cases on the last rule used in the derivation tree of $\Gamma,x \myvartriangleleft t,\Gamma'; s$ prevalid.
  \begin{description}
    \item[Rule \textsc{Pv-Emp}:]
      In this case, we have $\Gamma, x \myvartriangleleft t, \Gamma' = \varepsilon$.
      Hence $\Gamma = \varepsilon$ and $\Gamma' = \varepsilon$.
      The result to prove in this case is $\Gamma, \cas{\Gamma'}{x}{t} = \varepsilon$ is prevalid,
      which holds by rule \textsc{Pv-Emp}.

    \item[Rule \textsc{Pv-Ctx}:]
      In this case, we have $\Gamma, x \myvartriangleleft t, \Gamma' = \Gamma_0, y \myleq t_y$
      with $\Gamma_0 ~\prevalid$ and $y \not\in \dom(\Gamma_0)$ and $\fv(t_y) \subseteq \dom(\Gamma_0)$.
      Either $y \myleq t_y$ is $x \myvartriangleleft t$, or it is in $\Gamma_0$ or $\Gamma'$.

      If $y \myleq t_y$ is $x \myvartriangleleft t$, then $\Gamma_0 = \Gamma$ and $\Gamma' = \varepsilon$.
      We need to prove that $\Gamma, \cas{\Gamma'}{x}{t} = \Gamma$ is prevalid,
      which holds by assumption.

      If $y \myleq t_y$ is in $\Gamma_0$, then $\Gamma_0 = \Gamma_1, y \myleq t_y$ for some $\Gamma_1$.
      We have $\Gamma_1 ~\prevalid$ and $y \not\in \dom(\Gamma_1)$ and $\fv(t_y) \subseteq \dom(\Gamma_1)$.
      We also have $\Gamma = \Gamma_1, y \myleq t_y$ and $\Gamma' = \varepsilon$.
      By induction hypothesis, $\Gamma_1, \cas{\Gamma'}{x}{t} = \Gamma_1$ is prevalid.
      Because $y \not\in \dom(\Gamma_1)$, we have $y \not\in \dom(\Gamma_1, \cas{\Gamma'}{x}{t})$.
      Because $\fv(t_y) \subseteq \dom(\Gamma_1)$, we have $\fv(\cas{t_y}{x}{t}) \subseteq \dom(\Gamma_1, \cas{\Gamma'}{x}{t})$.
      Hence $\Gamma_1, \cas{\Gamma'}{x}{t}, y \myleq \cas{t_y}{x}{t}$ is prevalid.

      If $y \myleq t_y$ is in $\Gamma'$, then $\Gamma' = \Gamma_1, y \myleq t_y$ for some $\Gamma_1$.
      We have $\Gamma_1 ~\prevalid$ and $y \not\in \dom(\Gamma_1)$ and $\fv(t_y) \subseteq \dom(\Gamma_1)$.
      We also have $\Gamma = \Gamma$ and $\Gamma' = \Gamma_1, y \myleq t_y$.
      By induction hypothesis, $\Gamma, \cas{\Gamma_1}{x}{t}$ is prevalid.
      Because $y \not\in \dom(\Gamma_1)$, we have $y \not\in \dom(\Gamma, \cas{\Gamma_1}{x}{t})$.
      Because $\fv(t_y) \subseteq \dom(\Gamma_1)$, we have $\fv(\cas{t_y}{x}{t}) \subseteq \dom(\Gamma, \cas{\Gamma_1}{x}{t})$.
      Hence $\Gamma, \cas{\Gamma_1}{x}{t}, y \myleq \cas{t_y}{x}{t}$ is prevalid.

    \item[Rule \textsc{Pv-EqA}:]
      In this case, we have $\Gamma, x \myvartriangleleft t, \Gamma' = \Gamma_0, y \myequiv \alpha$
      with $\Gamma_0 ~\prevalid$ and $y \not\in \dom(\Gamma_0)$ and $\fv(\alpha) \subseteq \dom(\Gamma_0)$.
      Either $y \myequiv \alpha$ is $x \myvartriangleleft t$, or it is in $\Gamma_0$ or $\Gamma'$.

      If $y \myequiv \alpha$ is $x \myvartriangleleft t$, then $\Gamma_0 = \Gamma$ and $\Gamma' = \varepsilon$.
      We need to prove that $\Gamma, \cas{\Gamma'}{x}{t} = \Gamma$ is prevalid,
      which holds by assumption.

      If $y \myequiv \alpha$ is in $\Gamma_0$, then $\Gamma_0 = \Gamma_1, y \myequiv \alpha$ for some $\Gamma_1$.
      We have $\Gamma_1 ~\prevalid$ and $y \not\in \dom(\Gamma_1)$ and $\fv(\alpha) \subseteq \dom(\Gamma_1)$.
      We also have $\Gamma = \Gamma_1, y \myequiv \alpha$ and $\Gamma' = \varepsilon$.
      By induction hypothesis, $\Gamma_1, \cas{\Gamma'}{x}{t} = \Gamma_1$ is prevalid.
      Because $y \not\in \dom(\Gamma_1)$, we have $y \not\in \dom(\Gamma_1, \cas{\Gamma'}{x}{t})$.
      Because $\fv(\alpha) \subseteq \dom(\Gamma_1)$, we have $\fv(\cas{\alpha}{x}{t}) \subseteq \dom(\Gamma_1, \cas{\Gamma'}{x}{t})$.
      Hence $\Gamma_1, \cas{\Gamma'}{x}{t}, y \myequiv \cas{\alpha}{x}{t}$ is prevalid.

      If $y \myequiv \alpha$ is in $\Gamma'$, then $\Gamma' = \Gamma_1, y \myequiv \alpha$ for some $\Gamma_1$.
      We have $\Gamma_1 ~\prevalid$ and $y \not\in \dom(\Gamma_1)$ and $\fv(\alpha) \subseteq \dom(\Gamma_1)$.
      We also have $\Gamma = \Gamma$ and $\Gamma' = \Gamma_1, y \myequiv \alpha$.
      By induction hypothesis, $\Gamma, \cas{\Gamma_1}{x}{t}$ is prevalid.
      Because $y \not\in \dom(\Gamma_1)$, we have $y \not\in \dom(\Gamma, \cas{\Gamma_1}{x}{t})$.
      Because $\fv(\alpha) \subseteq \dom(\Gamma_1)$, we have $\fv(\cas{\alpha}{x}{t}) \subseteq \dom(\Gamma, \cas{\Gamma_1}{x}{t})$.
      Hence $\Gamma, \cas{\Gamma_1}{x}{t}, y \myequiv \cas{\alpha}{x}{t}$ is prevalid.

    \item[Rule \textsc{Pv-Nil}:]
      In this case, we have $\Gamma, x \myvartriangleleft t, \Gamma'; s = \Gamma, x \myvartriangleleft t, \Gamma'; \nil$
      with $\Gamma, x \myvartriangleleft t, \Gamma' ~\prevalid$.
      By induction hypothesis, $\Gamma, \cas{\Gamma'}{x}{t} ~\prevalid$.
      Hence $\Gamma, \cas{\Gamma'}{x}{t}; \nil$ is prevalid.

    \item[Rule \textsc{Pv-Sta}:]
      In this case, we have \\$\Gamma, x \myvartriangleleft t, \Gamma'; s = \Gamma, x \myvartriangleleft t, \Gamma'; \alpha :: s_0$
      with $\Gamma, x \myvartriangleleft t, \Gamma'; s_0 ~\prevalid$ and $\fv(\alpha) \subseteq \dom(\Gamma, x \myvartriangleleft t, \Gamma')$.
      By induction hypothesis, $\Gamma, \cas{\Gamma'}{x}{t}; \cas{s_0}{x}{t} ~\prevalid$.
      Because $\fv(\alpha) \subseteq \dom(\Gamma, x \myvartriangleleft t, \Gamma')$,
      we have $\fv(\cas{\alpha}{x}{t}) \subseteq \dom(\Gamma, \cas{\Gamma'}{x}{t})$.
      Hence $\Gamma, \cas{\Gamma'}{x}{t}; \cas{\alpha}{x}{t} :: \cas{s_0}{x}{t}$ is prevalid.
  \end{description}
\end{proof}

\begin{lemma}[Promotion under substitution outside of covariant contexts]
  \label{lem:subtyping-under-substitution-aux1}
  Let $\Gamma; s$ be an extended context, and $\Gamma'$ be a logical context,
  such that $\Gamma, \cas{\Gamma'}{x}{\alpha}; \nil$ is prevalid.
  Let $t$ and $\alpha$ be terms.

  Let $\Gamma, x \myleq t, \Gamma'; \nil \vdash \Po[x] \myrel{\leq} \Po[t]$ be a derivation 
  for some covariant context $\Po$.

  Then $\Gamma, x \myleq t, \cas{\Gamma'}{x}{\alpha}; \nil \vdash \Po_{\cas{}{x}{\alpha}}[x] \myrel{\leq} \Po_{\cas{}{x}{\alpha}}[t]$.
\end{lemma}
\begin{proof}
  By induction on $u$. We distinguish cases based on the structure of $u$.
  \begin{description}
    \item[$u = \T$:] Impossible by assumption.
    \item[$u = x$:] Nothing to do.
    \item[$u = \lambda y \myleq z. u'$:]
      By assumption, we have $\Gamma, x \myleq t, \Gamma'; \nil \vdash \Po[\lambda y \myleq z. u'] \myrel{\leq} \Po[\lambda y \myleq z. v]$.\\
      Let $P'$ be such that $\Po[\lambda y \myleq z. u'] = \Po'[u']$.\\
      Our derivation is now $\Gamma, x \myleq t, \Gamma'; \nil \vdash \Po'[u'] \myrel{\leq} \Po'[v]$.\\
      By induction, we have $\Gamma, x \myleq t, \cas{\Gamma'}{x}{\alpha}; \nil \vdash \Po'_{\cas{}{x}{\alpha}}[u'] \myrel{\leq} \Po'_{\cas{}{x}{\alpha}}[v]$.\\
      By definition of $\Po'$, this is \\
      $\Gamma, x \myleq t, \cas{\Gamma'}{x}{\alpha}; \nil \vdash \Po_{\cas{}{x}{\alpha}}[\lambda y \myleq z. \cas{u'}{x}{\alpha}] \myrel{\leq} \Po_{\cas{}{x}{\alpha}}[\lambda y \myleq z. \cas{v}{x}{\alpha}]$.\\
      To conclude we need to prove that\\
      $\Gamma, x \myleq t, \cas{\Gamma'}{x}{\alpha}; \nil \vdash \Po_{\cas{}{x}{\alpha}}[\lambda y \myleq \cas{z}{x}{\alpha}. \cas{u'}{x}{\alpha}] \myrel{\leq} \Po_{\cas{}{x}{\alpha}}[\lambda y \myleq \cas{z}{x}{\alpha}. \cas{v}{x}{\alpha}]$
      holds.\\
      Because this derivation make a promotion of $x$ to $t$, it doesn't do a promotion of $y$ to its type annotation,
      as a result this derivation holds.\\
      By alpha conversion, the term $\lambda y \myleq \cas{z}{x}{\alpha}. \cas{u'}{x}{\alpha}$\\
      is equal to $\cas{(\lambda y \myleq z. u')}{x}{\alpha}$, and similarly for $\lambda y \myleq \cas{z}{x}{\alpha}. \cas{v}{x}{\alpha}$.\\
      Hence the result.
    \item[$u = u' \, u''$:]
      By assumption, we have \\$\Gamma, x \myleq t, \Gamma'; \nil \vdash \Po[u' \, u''] \myrel{\leq} \Po[v' \, v'']$.\\
      Let $P'$ be such that $\Po[u' \, u''] = \Po'[u']$.\\
      Our derivation is now $\Gamma, x \myleq t, \Gamma'; \nil \vdash \Po'[](u') \myrel{\leq} \Po'[](v')$.\\
      By induction, we have $\Gamma, x \myleq t, \cas{\Gamma'}{x}{\alpha}; \nil \vdash \cas{(\Po'[])}{x}{\alpha}(u') \myrel{\leq} \cas{(\Po'[])}{x}{\alpha}(v')$.\\
      By definition of $P'$, this is $\Gamma, x \myleq t, \cas{\Gamma'}{x}{\alpha}; \nil \vdash \Po_{\cas{}{x}{\alpha}}[u' \, \cas{u''}{x}{\alpha}] \myrel{\leq} \Po_{\cas{}{x}{\alpha}}[v' \, \cas{v''}{x}{\alpha}]$.\\
      To conclude we need to prove that\\
      $\Gamma, x \myleq t, \cas{\Gamma'}{x}{\alpha}; \nil \vdash \Po_{\cas{}{x}{\alpha}}[\cas{u'}{x}{\alpha} \, \cas{u''}{x}{\alpha}] \myrel{\leq} \Po_{\cas{}{x}{\alpha}}[\cas{v'}{x}{\alpha} \, \cas{v''}{x}{\alpha}]$
      holds.\\
      Because this derivation make a promotion of $x$ to $t$, it doesn't use the annotation of $u''$ in the stack,
      as a result this derivation holds.\\
      By definition of the substitution, the term $\cas{u'}{x}{\alpha} \, \cas{u''}{x}{\alpha}$\\
      is equal to $\cas{(u' \, u'')}{x}{\alpha}$, and similarly for $\cas{v'}{x}{\alpha} \, \cas{v''}{x}{\alpha}$,
      hence the result.
  \end{description}
\end{proof}

\begin{lemma}[Promotion under substitution inside of covariant contexts]
  \label{lem:subtyping-under-substitution-aux2}
  Let $\Gamma; s$ be an extended context.
  Let $\Gamma'$ be a logical context,
  such that $\Gamma, \cas{\Gamma'}{x}{\alpha}; \cas{s}{x}{\alpha}$ is prevalid.
  Let $u$, $v$, $t$ and $\alpha$ be terms.
  Let $x$ be a variable.
  If $\Gamma, x \myleq t, \Gamma'; s \vdash u \myrel{\leq} v$,
  such that this derivation is NOT of the form $\Gamma, x \myleq t, \Gamma'; s \vdash \Po[x] \myrel{\leq} \Po[t]$
  for some covariant context $P$.
  Then we have $\Gamma, \cas{\Gamma'}{x}{\alpha}; \cas{s}{x}{\alpha} \vdash \cas{u}{x}{\alpha} \myrel{\leq} \cas{v}{x}{\alpha}$.
\end{lemma}
\begin{proof}
  By induction on the derivation tree of $\Gamma, x \myleq t, \Gamma'; s \vdash u \myrel{\leq} v$.
  We distinguish cases on the last rule used.
  \begin{description}
    \item[Rule \textsc{Ms-Pro}:]
      In this case, we have $u = y$ and $v = t_y$ with $y \myleq t_y \in \Gamma, x \myleq t, \Gamma'$
      and $\Gamma, x \myleq t, \Gamma'$ prevalid.
      We know that $y \not= x$ from the assumption that our derivation does not make a promotion of $x$ to $t$.
      If $y \myleq t_y \in \Gamma$, by prevalidity of $\Gamma, x \myleq t, \Gamma'$ we know
      there is no instance of $x$ in $t_y$, and as such we have $\cas{t_y}{x}{\alpha} = t_y$.
      Now we have $\Gamma, \cas{\Gamma'}{x}{\alpha}; \cas{s}{x}{\alpha} \vdash y \myrel{\leq} t_y$
      by rule \textsc{Ms-Pro}, which is the desired result.
      If $y \myleq t_y \in \Gamma'$, we have by definition of the substitution on context,
      $y \myleq \cas{t_y}{x}{\alpha} \in \cas{\Gamma'}{x}{\alpha}$,
      and as a result we have $\Gamma, \cas{\Gamma'}{x}{\alpha}; \cas{s}{x}{\alpha} \vdash y \myrel{\leq} \cas{t_y}{x}{\alpha}$
      by rule \textsc{Ms-Pro}.

    \item[Rule \textsc{Ms-Top}:]
      In this case, we have \\$u = y$ and $v = \T$.
      We need to show $\Gamma, \cas{\Gamma'}{x}{\alpha}; \cas{s}{x}{\alpha} \vdash \cas{u}{x}{\alpha} \myrel{\leq} \cas{\T}{x}{\alpha}$.
      Since $\cas{\T}{x}{\alpha} = \T$,
      we need to show $\Gamma, \cas{\Gamma'}{x}{\alpha}; \cas{s}{x}{\alpha} \vdash \cas{u}{x}{\alpha} \myrel{\leq} \T$,
      which holds by rule \textsc{Ms-Top}.

    \item[Rule \textsc{Ms-Equ}:]
      In this case, we have \\$\Gamma, x \myleq t, \Gamma'; s \vdash u \myrel{\equiv} v$.
      By Lemma~\ref{lem:reduction-under-substitution}, $\Gamma, \cas{\Gamma'}{x}{\alpha}; \cas{s}{x}{\alpha} \vdash \cas{u}{x}{\alpha} \myrel{\equiv} \cas{v}{x}{\alpha}$.
      Hence, $\Gamma, \cas{\Gamma'}{x}{\alpha}; \cas{s}{x}{\alpha} \vdash \cas{u}{x}{\alpha} \myrel{\leq} \cas{v}{x}{\alpha}$ by \textsc{Ms-Equ}.

    \item[Rule \textsc{Ms-App}:]
      In this case, we have $u = a\,b$ and $v = a'\,b$ with premise $\Gamma, x \myleq t, \Gamma'; b::s \vdash a \myrel{\leq} a'$.
      By the induction hypothesis, $\Gamma, \cas{\Gamma'}{x}{\alpha}; \cas{b}{x}{\alpha}::\cas{s}{x}{\alpha} \vdash \cas{a}{x}{\alpha} \myrel{\leq} \cas{a'}{x}{\alpha}$.
      Hence $\Gamma, \cas{\Gamma'}{x}{\alpha}; \cas{s}{x}{\alpha} \vdash \cas{a}{x}{\alpha} \, \cas{b}{x}{\alpha} \myrel{\leq} \cas{a'}{x}{\alpha} \, \cas{b}{x}{\alpha}$ by \textsc{Ms-App},
      which is by definition of the substitution,
      equivalent to $\Gamma, \cas{\Gamma'}{x}{\alpha}; \cas{s}{x}{\alpha} \vdash \cas{(a \, b)}{x}{\alpha} \myrel{\leq} \cas{(a' \, b)}{x}{\alpha}$.

    \item[Rule \textsc{Ms-Fun}:]
      In this case, $u = \lambda y \myleq w. a$, $v = \lambda y \myleq w. a'$ with\\ $\Gamma, x \myleq t, \Gamma', y \myleq w; \nil \vdash a \myrel{\leq} a'$.
      By the induction hypothesis, $\Gamma, \cas{\Gamma'}{x}{\alpha}, y \myleq \cas{w}{x}{\alpha}; \nil \vdash \cas{a}{x}{\alpha} \myrel{\leq} \cas{a'}{x}{\alpha}$.
      Thus, $\Gamma, \cas{\Gamma'}{x}{\alpha}; \nil \vdash \lambda y \myleq \cas{w}{x}{\alpha}. \cas{a}{x}{\alpha} \myrel{\leq} \lambda y \myleq \cas{w}{x}{\alpha}. \cas{a'}{x}{\alpha}$ by \textsc{Ms-Fun},
      which is by definition of the substitution, and because by alpha conversion we have $y \not= x$,
      equivalent to $\Gamma, \cas{\Gamma'}{x}{\alpha}; \nil \vdash \cas{(\lambda y \myleq w. a)}{x}{\alpha} \myrel{\leq} \cas{(\lambda y \myleq w. a')}{x}{\alpha}$.

    \item[Rule \textsc{Ms-FOp}:]
      In this case, $u = \lambda y \myleq w. a$, $v = \lambda y \myleq w. a'$ with $\Gamma, x \myleq t, \Gamma', y \equiv \beta; s \vdash a \myrel{\leq} a'$.
      By the induction hypothesis, $\Gamma, \cas{\Gamma'}{x}{\alpha}, y \equiv \cas{\beta}{x}{\alpha}; \cas{s}{x}{\alpha} \vdash \cas{a}{x}{\alpha} \myrel{\leq} \cas{a'}{x}{\alpha}$.
      Thus, $\Gamma, \cas{\Gamma'}{x}{\alpha}; \cas{s}{x}{\alpha} \vdash \lambda y \myleq \cas{w}{x}{\alpha}. \cas{a}{x}{\alpha} \myrel{\leq} \lambda y \myleq \cas{w}{x}{\alpha}. \cas{a'}{x}{\alpha}$ by \textsc{Ms-FOp},
      which is by definition of the substitution, and because by alpha conversion we have $y \not= x$,
      equivalent to $\Gamma, \cas{\Gamma'}{x}{\alpha}; \cas{s}{x}{\alpha} \vdash \cas{(\lambda y \myleq w. a)}{x}{\alpha} \myrel{\leq} \cas{(\lambda y \myleq w. a')}{x}{\alpha}$.
  \end{description}
\end{proof}

\begin{lemma}[Reduction under substitution]
  \label{lem:reduction-under-substitution}
  Let $\Gamma; s$ be an extended context.
  Let $\Gamma'$ be a logical context,
  such that $\Gamma, \cas{\Gamma'}{x}{\alpha}; \cas{s}{x}{\alpha}$ is prevalid.
  Let $u$, $v$, $t$ and $\alpha$ be terms.
  Let $x$ be a variable.
  If $\Gamma, x \myleq t, \Gamma'; s \vdash u \myrel{\equiv} v$,
  then $\Gamma, \cas{\Gamma'}{x}{\alpha}; \cas{s}{x}{\alpha} \vdash \cas{u}{x}{\alpha} \myrel{\equiv} \cas{v}{x}{\alpha}$.
\end{lemma}
\begin{proof}
  By induction on the derivation tree of $\Gamma, x \myleq t, \Gamma'; s \vdash u \myrel{\equiv} v$.
  We distinguish cases on the last rule used.
  \begin{description}
    \item[Rule \textsc{Me-Var}:]
      In this case, we have $u = y$ and $v = y$ with $\Gamma, x \myleq t, \Gamma'$ prevalid.
      If $y = x$, then $\cas{u}{x}{\alpha} = \alpha$ and $\cas{v}{x}{\alpha} = \alpha$.
      We need to show $\Gamma, \cas{\Gamma'}{x}{\alpha}; \cas{s}{x}{\alpha} \vdash \alpha \myrel{\equiv} \alpha$,
      which holds by reflexivity (Proposition~\ref{prop:algorithmic-refl}).
      If $y \neq x$, then $\cas{u}{x}{\alpha} = y$ and $\cas{v}{x}{\alpha} = y$.
      We need to show $\Gamma, \cas{\Gamma'}{x}{\alpha}; \cas{s}{x}{\alpha} \vdash y \myrel{\equiv} y$,
      which holds by rule \textsc{Me-Var}.

    \item[Rule \textsc{Me-Pro}:]
      In this case, we have $u = y$ and $v = \alpha''$
      with $y \myequiv \alpha' \in \Gamma, x \myleq t, \Gamma'$,
      and $\Gamma, x \myleq t, \Gamma' \vdash \alpha' \myrel{\equiv} \alpha''$,
      and $\Gamma, x \myleq t, \Gamma'$ prevalid.
      Because $\Gamma, x \myleq t, \Gamma'$ is prevalid, there is a single annotation of $x$,
      and because it's $x \myleq t$ it cannot be $y \myequiv \alpha'$ and as such we have $y \not= x$.
      If $y \myequiv \alpha' \in \Gamma$, by prevalidity of $\Gamma, x \myleq t, \Gamma'$ we know
      there is no instance of $x$ in $\alpha''$, and as such we have $\cas{\alpha''}{x}{\alpha} = \alpha'$ and
      $y \myequiv \alpha' \in \Gamma, \cas{\Gamma'}{x}{\alpha}$.
      Now we have $\Gamma, \cas{\Gamma'}{x}{\alpha}; \cas{s}{x}{\alpha} \vdash y \myrel{\equiv} \alpha''$
      by rule \textsc{Me-Pro}, which is the desired result.
      If $y \myequiv \alpha' \in \Gamma'$, we have by definition of the substitution on context,
      $y \myequiv \cas{\alpha'}{x}{\alpha} \in \cas{\Gamma'}{x}{\alpha}$,
      hence $y \myequiv \cas{\alpha'}{x}{\alpha} \in \Gamma, \cas{\Gamma'}{x}{\alpha}$.
      By induction, we have $\Gamma, \cas{\Gamma'}{x}{\alpha}; \cas{s}{x}{\alpha} \vdash \cas{\alpha'}{x}{\alpha} \myrel{\equiv} \cas{\alpha''}{x}{\alpha}$,
      hence by rule \textsc{Me-Pro}, we have $\Gamma, \cas{\Gamma'}{x}{\alpha}; \cas{s}{x}{\alpha} \vdash y \myrel{\equiv} \cas{\alpha''}{x}{\alpha}$.

    \item[Rule \textsc{Me-App}:]
      In this case, we have $u = a\,b$ and\\ $v = a'\,b'$ with premises
      $\Gamma, x \myleq t, \Gamma'; b::s \vdash a \myrel{\equiv} a'$ and\\
      $\Gamma, x \myleq t, \Gamma'; \nil \vdash b \myrel{\equiv} b'$.
      By the induction hypothesis, we obtain
      $\Gamma, \cas{\Gamma'}{x}{\alpha}; \cas{b}{x}{\alpha}::\cas{s}{x}{\alpha} \vdash \cas{a}{x}{\alpha} \myrel{\equiv} \cas{a'}{x}{\alpha}$ and
      $\Gamma, \cas{\Gamma'}{x}{\alpha}; \nil \vdash \cas{b}{x}{\alpha} \myrel{\equiv} \cas{b'}{x}{\alpha}$.
      Hence, by rule \textsc{Me-App}, we conclude that
      $\Gamma, \cas{\Gamma'}{x}{\alpha}; \cas{s}{x}{\alpha} \vdash \cas{a}{x}{\alpha} \, \cas{b}{x}{\alpha} \myrel{\equiv} \cas{a'}{x}{\alpha} \,  \cas{b'}{x}{\alpha}$.

    \item[Rule \textsc{Me-Top}:]
      In this case, we have $u = \T$ and \\$v = \T$.
      We need to show $\Gamma, \cas{\Gamma'}{x}{\alpha}; \cas{s}{x}{\alpha} \vdash \cas{\T}{x}{\alpha} \myrel{\equiv} \cas{\T}{x}{\alpha}$.
      Since $\cas{\T}{x}{\alpha} = \T$,
      we need to show $\Gamma, \cas{\Gamma'}{x}{\alpha}; \cas{s}{x}{\alpha} \vdash \T \myrel{\equiv} \T$,
      which holds by rule \textsc{Me-Top}.

    \item[Rule \textsc{Me-Fun}:]
      In this case, $u = \lambda y \myleq a. b$,\\ $v = \lambda y \myleq a'. b'$ and $s = \nil$
      with premises $\Gamma, x \myleq t, \Gamma'; \nil \vdash a \myrel{\equiv} a'$ and
      $\Gamma, x \myleq t, \Gamma', y \myleq a; \nil \vdash b \myrel{\equiv} b'$.
      By the induction hypothesis we obtain
      $\Gamma, \cas{\Gamma'}{x}{\alpha}; \nil \vdash \cas{a}{x}{\alpha} \myrel{\equiv} \cas{a'}{x}{\alpha}$ and
      $\Gamma, \cas{\Gamma'}{x}{\alpha}, y \myleq \cas{a}{x}{\alpha}; \nil \vdash \cas{b}{x}{\alpha} \myrel{\equiv} \cas{b'}{x}{\alpha}$.
      Therefore, by rule \textsc{Me-Fun}, we deduce that
      $\Gamma, \cas{\Gamma'}{x}{\alpha}; \nil \vdash \lambda y \myleq \cas{a}{x}{\alpha}. \cas{b}{x}{\alpha} \myrel{\equiv} \lambda y \myleq \cas{a'}{x}{\alpha}. \cas{b'}{x}{\alpha}$,
      which is by definition of the substitution, and because by alpha conversion we have $y \not= x$,
      $\Gamma, \cas{\Gamma'}{x}{\alpha}; \nil \vdash \cas{(\lambda y \myleq a. b)}{x}{\alpha} \myrel{\equiv} \cas{(\lambda y \myleq a'. b')}{x}{\alpha}$.

    \item[Rule \textsc{Me-FOp}:]
      In this case, $u = \lambda y \myleq a. b$, \\$v = \lambda y \myleq a'. b'$ and $s = \beta :: s_0$\\
      with premises $\Gamma, x \myleq t, \Gamma'; \nil \vdash a \myrel{\equiv} a'$ and
      $\Gamma, x \myleq t, \Gamma'; y \myequiv \beta; s_0 \vdash b \myrel{\equiv} b'$.\\
      By the induction hypothesis we obtain
      $\Gamma, \cas{\Gamma'}{x}{\alpha}; \nil \vdash \cas{a}{x}{\alpha} \myrel{\equiv} \cas{a'}{x}{\alpha}$ and\\
      $\Gamma, \cas{\Gamma'}{x}{\alpha}, y \myequiv \cas{\beta}{x}{\alpha}; \cas{s_0}{x}{\alpha} \vdash \cas{b}{x}{\alpha} \myrel{\equiv} \cas{b'}{x}{\alpha}$.
      Therefore, by rule \textsc{Me-FOp}, we deduce that
      $\Gamma, \cas{\Gamma'}{x}{\alpha}; \cas{(\beta :: s_0)}{x}{\alpha} \vdash \lambda y \myleq \cas{a}{x}{\alpha}. \cas{b}{x}{\alpha} \myrel{\equiv} \lambda y \myleq \cas{a'}{x}{\alpha}. \cas{b'}{x}{\alpha}$,
      which is by definition of the substitution, and because by alpha conversion we have $y \not= x$,
      $\Gamma, \cas{\Gamma'}{x}{\alpha}; \cas{s}{x}{\alpha} \vdash \cas{(\lambda y \myleq a. b)}{x}{\alpha} \myrel{\equiv} \cas{(\lambda y \myleq a'. b')}{x}{\alpha}$.

    \item[Rule \textsc{Me-Bet}:]
      In this case, $u = (\lambda y \myleq a. b) \, c$ and \\$v = \cas{b'}{y}{c'}$
      with premises \\$\Gamma, x \myleq t, \Gamma'; s \vdash b \myrel{\equiv} b'$ and
      $\Gamma, x \myleq t, \Gamma'; \nil \vdash c \myrel{\equiv} c'$.\\
      By the induction hypothesis we obtain
      $\Gamma, \cas{\Gamma'}{x}{\alpha}; \cas{s}{x}{\alpha} \vdash \cas{b}{x}{\alpha} \myrel{\equiv} \cas{b'}{x}{\alpha}$ and\\
      $\Gamma, \cas{\Gamma'}{x}{\alpha}; \nil \vdash \cas{c}{x}{\alpha} \myrel{\equiv} \cas{c'}{x}{\alpha}$.
      Therefore, by rule \textsc{Me-Bet}, we deduce that\\
      $\Gamma, \cas{\Gamma'}{x}{\alpha}; \cas{s}{x}{\alpha} \vdash (\lambda y \myleq \cas{a}{x}{\alpha}. \cas{b}{x}{\alpha}) \, \cas{c}{x}{\alpha} \myrel{\equiv} \cas{\cas{b}{x}{\alpha}}{y}{\cas{c'}{x}{\alpha}}$.\\
      This last term, according to Barendregt's substitution lemma \cite{Bar92}, is $\cas{\cas{b'}{y}{c'}}{x}{\alpha}$,
      Hence $\Gamma, \cas{\Gamma'}{x}{\alpha}; \cas{s}{x}{\alpha} \vdash \cas{(\lambda y \myleq a. b)}{x}{\alpha} \, \cas{c}{x}{\alpha} \myrel{\equiv} \cas{\cas{b'}{y}{c'}}{x}{\alpha}$.

    \item[Rule \textsc{Me-TAp}:]
      In this case, $u = \T \, b$ and $v = \T$ with premise $\Gamma, x \myleq t, \Gamma'; s ~\prevalid$.
      We need to show $\Gamma, \cas{\Gamma'}{x}{\alpha}; \cas{s}{x}{\alpha} \vdash \cas{\T \, b}{x}{\alpha} \myrel{\equiv} \cas{\T}{x}{\alpha}$.
      Since $\cas{\T}{x}{\alpha} = \T$,
      we need to show $\Gamma, \cas{\Gamma'}{x}{\alpha}; \cas{s}{x}{\alpha} \vdash \T \, \cas{b}{x}{\alpha} \myrel{\equiv} \T$,
      which holds by rule \textsc{Me-TAp}.
  \end{description}
\end{proof}

\begin{lemma}[Reduction under substitution - auxiliary for the commutation theorem]
  \label{lem:commutativity-reduction-under-substitution}
  Let $\Gamma; s$ be an extended context.
  Let $u$, $u'$, $v$ and $v'$ be terms.
  Let $x$ be a variable.
  If $\Gamma, x \myequiv v, \Gamma'; s \vdash u \myrel{\equiv} u'$ and $\Gamma; \nil \vdash v \myrel{\equiv} v'$,
  then $\Gamma, \cas{\Gamma'}{x}{v}; \cas{s}{x}{v} \vdash \cas{u}{x}{v} \myrel{\equiv} \cas{u'}{x}{v'}$.
\end{lemma}
\begin{proof}
  By induction on the derivation tree of $\Gamma, x \myequiv v, \Gamma'; s \vdash u \myrel{\equiv} v$.
  We distinguish cases on the last rule used.
  \begin{description}
    \item[Rule \textsc{Me-Var}:]
      In this case, we have $u = y$ and $v = y$ with $\Gamma, x \myequiv v, \Gamma'$ prevalid.
      If $y = x$, then $\cas{u}{x}{v} = v$ and $\cas{v}{x}{v} = v$.
      We need to show $\Gamma, \cas{\Gamma'}{x}{v}; \cas{s}{x}{v} \vdash v \myrel{\equiv} v$,
      which holds by reflexivity (Proposition~\ref{prop:algorithmic-refl}).
      If $y \neq x$, then $\cas{u}{x}{v} = y$ and $\cas{v}{x}{v} = y$.
      We need to show $\Gamma, \cas{\Gamma'}{x}{v}; \cas{s}{x}{v} \vdash y \myrel{\equiv} y$,
      which holds by rule \textsc{Me-Var}.

    \item[Rule \textsc{Me-Pro} with $u = y \neq x$:]
      In this case, we have $u = y$ and $v = \alpha''$
      with $y \myequiv \alpha' \in \Gamma, x \myequiv v, \Gamma'$,
      and $\Gamma, x \myequiv v, \Gamma' \vdash \alpha' \myrel{\equiv} \alpha''$,
      and $\Gamma, x \myequiv v, \Gamma'$ prevalid.
      If $y \myequiv \alpha' \in \Gamma$, by prevalidity of $\Gamma, x \myequiv v, \Gamma'$ we know
      there is no instance of $x$ in $\alpha''$, and as such we have $\cas{\alpha''}{x}{v} = \alpha'$ and
      $y \myequiv \alpha' \in \Gamma, \cas{\Gamma'}{x}{v}$.
      Now we have $\Gamma, \cas{\Gamma'}{x}{v}; \cas{s}{x}{v} \vdash y \myrel{\equiv} \alpha''$
      by rule \textsc{Me-Pro}, which is the desired result.
      If $y \myequiv \alpha' \in \Gamma'$, we have by definition of the substitution on context,
      $y \myequiv \cas{\alpha'}{x}{v} \in \cas{\Gamma'}{x}{v}$,
      hence $y \myequiv \cas{\alpha'}{x}{v} \in \Gamma, \cas{\Gamma'}{x}{v}$.
      By induction, we have $\Gamma, \cas{\Gamma'}{x}{v}; \cas{s}{x}{v} \vdash \cas{\alpha'}{x}{v} \myrel{\equiv} \cas{\alpha''}{x}{v}$,
      hence by rule \textsc{Me-Pro}, we have $\Gamma, \cas{\Gamma'}{x}{v}; \cas{s}{x}{v} \vdash y \myrel{\equiv} \cas{\alpha''}{x}{v}$.

    \item[Rule \textsc{Me-Pro} with $u = x$:] 
      In this case, we have $u = x$ and $v = \alpha$ with premise $\Gamma, x \myequiv v, \Gamma'; s \vdash v \myrel{\equiv} \alpha$.
      By prevalidity, we know that there is no instance of $x$ in both $v$ and $\alpha$. As a result we have both $\cas{v}{x}{v} = v$ and $\cas{\alpha}{x}{v} = \alpha$.
      Therefore induction we have $\Gamma, \cas{\Gamma'}{x}{v}; \cas{s}{x}{v} \vdash v \myrel{\equiv} \alpha$.
      Because $\cas{u}{x}{v} = v$, this is the desired result.

    \item[Rule \textsc{Me-App}:]
      In this case, we have $u = a\,b$ and \\$v = a'\,b'$ with premises
      $\Gamma, x \myequiv v, \Gamma'; b::s \vdash a \myrel{\equiv} a'$ and\\
      $\Gamma, x \myequiv v, \Gamma'; \nil \vdash b \myrel{\equiv} b'$.
      By the induction hypothesis, we obtain
      $\Gamma, \cas{\Gamma'}{x}{v}; \cas{b}{x}{v}::\cas{s}{x}{v} \vdash \cas{a}{x}{v} \myrel{\equiv} \cas{a'}{x}{v}$ and
      $\Gamma, \cas{\Gamma'}{x}{v}; \nil \vdash \cas{b}{x}{v} \myrel{\equiv} \cas{b'}{x}{v}$.
      Hence, by rule \textsc{Me-App}, we conclude that
      $\Gamma, \cas{\Gamma'}{x}{v}; \cas{s}{x}{v} \vdash \cas{a}{x}{v} \, \cas{b}{x}{v} \myrel{\equiv} \cas{a'}{x}{v} \,  \cas{b'}{x}{v}$.

    \item[Rule \textsc{Me-Top}:]
      In this case, we have $u = \T$ and \\$v = \T$.
      We need to show $\Gamma, \cas{\Gamma'}{x}{v}; \cas{s}{x}{v} \vdash \cas{\T}{x}{v} \myrel{\equiv} \cas{\T}{x}{v}$.
      Since $\cas{\T}{x}{v} = \T$,
      we need to show $\Gamma, \cas{\Gamma'}{x}{v}; \cas{s}{x}{v} \vdash \T \myrel{\equiv} \T$,
      which holds by rule \textsc{Me-Top}.

    \item[Rule \textsc{Me-Fun}:]
      In this case, $u = \lambda y \myleq a. b$,\\ $v = \lambda y \myleq a'. b'$ and $s = \nil$
      with premises $\Gamma, x \myequiv v, \Gamma'; \nil \vdash a \myrel{\equiv} a'$ and
      $\Gamma, x \myequiv v, \Gamma', y \myleq a; \nil \vdash b \myrel{\equiv} b'$.
      By the induction hypothesis we obtain
      $\Gamma, \cas{\Gamma'}{x}{v}; \nil \vdash \cas{a}{x}{v} \myrel{\equiv} \cas{a'}{x}{v}$ and
      $\Gamma, \cas{\Gamma'}{x}{v}, y \myleq \cas{a}{x}{v}; \nil \vdash \cas{b}{x}{v} \myrel{\equiv} \cas{b'}{x}{v}$.
      Therefore, by rule \textsc{Me-Fun}, we deduce that
      $\Gamma, \cas{\Gamma'}{x}{v}; \nil \vdash \lambda y \myleq \cas{a}{x}{v}. \cas{b}{x}{v} \myrel{\equiv} \lambda y \myleq \cas{a'}{x}{v}. \cas{b'}{x}{v}$,
      which is by definition of the substitution, and because by alpha conversion we have $y \not= x$,
      $\Gamma, \cas{\Gamma'}{x}{v}; \nil \vdash \cas{(\lambda y \myleq a. b)}{x}{v} \myrel{\equiv} \cas{(\lambda y \myleq a'. b')}{x}{v}$.

    \item[Rule \textsc{Me-FOp}:]
      In this case, $u = \lambda y \myleq a. b$, \\$v = \lambda y \myleq a'. b'$ and $s = \beta :: s_0$
      with premises \\$\Gamma, x \myequiv v, \Gamma'; \nil \vdash a \myrel{\equiv} a'$ and
      $\Gamma, x \myequiv v, \Gamma'; y \myequiv \beta; s_0 \vdash b \myrel{\equiv} b'$.\\
      By the induction hypothesis we obtain
      $\Gamma, \cas{\Gamma'}{x}{v}; \nil \vdash \cas{a}{x}{v} \myrel{\equiv} \cas{a'}{x}{v}$ and\\
      $\Gamma, \cas{\Gamma'}{x}{v}, y \myequiv \cas{\beta}{x}{v}; \cas{s_0}{x}{v} \vdash \cas{b}{x}{v} \myrel{\equiv} \cas{b'}{x}{v}$.
      Therefore, by rule \textsc{Me-FOp}, we deduce that
      $\Gamma, \cas{\Gamma'}{x}{v}; \cas{(\beta :: s_0)}{x}{v} \vdash \lambda y \myleq \cas{a}{x}{v}. \cas{b}{x}{v} \myrel{\equiv} \lambda y \myleq \cas{a'}{x}{v}. \cas{b'}{x}{v}$,
      which is by definition of the substitution, and because by alpha conversion we have $y \not= x$,
      $\Gamma, \cas{\Gamma'}{x}{v}; \cas{s}{x}{v} \vdash \cas{(\lambda y \myleq a. b)}{x}{v} \myrel{\equiv} \cas{(\lambda y \myleq a'. b')}{x}{v}$.

    \item[Rule \textsc{Me-Bet}:]
      In this case, $u = (\lambda y \myleq a. b) \, c$\\ and $v = \cas{b'}{y}{c'}$
      with premises \\$\Gamma, x \myequiv v, \Gamma'; s \vdash b \myrel{\equiv} b'$ and
      $\Gamma, x \myequiv v, \Gamma'; \nil \vdash c \myrel{\equiv} c'$.\\
      By the induction hypothesis we obtain
      $\Gamma, \cas{\Gamma'}{x}{v}; \cas{s}{x}{v} \vdash \cas{b}{x}{v} \myrel{\equiv} \cas{b'}{x}{v}$ and\\
      $\Gamma, \cas{\Gamma'}{x}{v}; \nil \vdash \cas{c}{x}{v} \myrel{\equiv} \cas{c'}{x}{v}$.
      Therefore, by rule \textsc{Me-Bet}, we deduce that\\
      $\Gamma, \cas{\Gamma'}{x}{v}; \cas{s}{x}{v} \vdash (\lambda y \myleq \cas{a}{x}{v}. \cas{b}{x}{v}) \, \cas{c}{x}{v} \myrel{\equiv} \cas{\cas{b}{x}{v}}{y}{\cas{c'}{x}{v}}$.\\
      This last term, according to Barendregt's substitution lemma \cite{Bar92}, is $\cas{\cas{b'}{y}{c'}}{x}{v}$,\\
      Hence $\Gamma, \cas{\Gamma'}{x}{v}; \cas{s}{x}{v} \vdash \cas{(\lambda y \myleq a. b)}{x}{v} \, \cas{c}{x}{v} \myrel{\equiv} \cas{\cas{b'}{y}{c'}}{x}{v}$.

    \item[Rule \textsc{Me-TAp}:]
      In this case, $u = \T \, b$ and $v = \T$ with premise $\Gamma, x \myequiv v, \Gamma'; s ~\prevalid$.
      We need to show $\Gamma, \cas{\Gamma'}{x}{v}; \cas{s}{x}{v} \vdash \cas{\T \, b}{x}{v} \myrel{\equiv} \cas{\T}{x}{v}$.
      Since $\cas{\T}{x}{v} = \T$,
      we need to show $\Gamma, \cas{\Gamma'}{x}{v}; \cas{s}{x}{v} \vdash \T \, \cas{b}{x}{v} \myrel{\equiv} \T$,
      which holds by rule \textsc{Me-TAp}.
  \end{description}
\end{proof}

%% CONGRUENCE

\begin{lemma}[Congruence of $\leq$]
  \label{lem:algorithmic-congruence}
  Let $\Gamma; s$ be an extended context, and let $u$, $u'$, $v$, $t$ and $\alpha$ be terms
  We have these six congruence relations:

  If $\Gamma; v :: s \vdash u \leq u'$, then $\Gamma; s \vdash u \, v \leq u' \, v$.

  If $\Gamma; v :: s \vdash u \leq^* u'$, then $\Gamma; s \vdash u \, v \leq^* u' \, v$.

  If $\Gamma, x \myleq t; \nil \vdash u \leq u'$ then $\Gamma; \nil \vdash \lambda x \myleq t. u \leq \lambda x \myleq t. u'$.

  If $\Gamma, x \myleq t; \nil \vdash u \leq^* u'$ then $\Gamma; \nil \vdash \lambda x \myleq t. u \leq^* \lambda x \myleq t. u'$.

  If $\Gamma, x \myequiv \alpha; s \vdash u \leq u'$ then $\Gamma; \alpha :: s \vdash \lambda x \myleq t. u \leq \lambda x \myleq t. u'$.

  If $\Gamma, x \myequiv \alpha; s \vdash u \leq^* u'$ then $\Gamma; \alpha :: s \vdash \lambda x \myleq t. u \leq^* \lambda x \myleq t. u'$.
\end{lemma}
\begin{proof}
  We do the first case. Every other case is similar:
  it lifts Lemmas~\ref{lem:algorithmic-congruence-subtyping-reduction} and~\ref{lem:algorithmic-congruence-equivalence-reduction}.

  Suppose we have $\Gamma; v :: s \vdash u \leq u'$.
  By definition of the subtype relation, we have a term $t$ such that
  $\Gamma; v :: s \vdash u \mrelTwo{\equiv} t$ and $\Gamma; v :: s \vdash u' \mrelTwo{\leq} t$.
  By Lemma~\ref{lem:algorithmic-congruence-equivalence-reduction}, we have $\Gamma; s \vdash u \, v \mrelTwo{\equiv} t \, v$
  from $\Gamma; v :: s \vdash u \mrelTwo{\equiv} t$.
  By Lemma~\ref{lem:algorithmic-congruence-subtyping-reduction}, we have $\Gamma; s \vdash u' \, v \mrelTwo{\leq} t \, v$
  from $\Gamma; v :: s \vdash u' \mrelTwo{\leq} t$.
  We then conclude by definition of the subtype relation with the intermediate term $t \, v$.
\end{proof}

\begin{lemma}[Congruence of $\myrel{\leq}$]
  \label{lem:algorithmic-congruence-subtyping-reduction}
  Let $\Gamma; s$ be an extended context, and let $u$, $u'$, $v$ and $t$ be terms
  We have these congruence relations:

  If $\Gamma; v :: s \vdash u \mrelTwo{\leq} u'$, then $\Gamma; s \vdash u \, v \mrelTwo{\leq} u' \, v$.

  If $\Gamma, x \myleq t; \nil \vdash u \mrelTwo{\leq} u'$ then $\Gamma; \nil \vdash \lambda x \myleq t. u \mrelTwo{\leq} \lambda x \myleq t. u'$.

  If $\Gamma, x \myequiv \alpha; s \vdash u \mrelTwo{\leq} u'$ then $\Gamma; \alpha :: s \vdash \lambda x \myleq t. u \mrelTwo{\leq} \lambda x \myleq t. u'$.
\end{lemma}
\begin{proof}
  For the first case, take $\Gamma; v :: s \vdash u \mrelTwo{\leq} u'$.
  On each $\Gamma; v :: s \vdash a \myrel{\leq} b$ subderivation of our initial multi-step derivation,
  we apply the rule \textsc{Ms-App}, to obtain $\Gamma; s \vdash a\,v \myrel{\leq} b\,v$.
  By doing so on each subderivation, we finally obtain the derivation $\Gamma; s \vdash u \, v \mrelTwo{\leq} u' \, v$,
  which is the desired result.

  The other cases are similar, with the use of rules \textsc{Ms-Fun} and \textsc{Ms-FOp} instead.
\end{proof}

\begin{lemma}[Congruence of $\myrel{\equiv}$]
  \label{lem:algorithmic-congruence-equivalence-reduction}
  Let $\Gamma; s$ be an extended context, and let $u$, $u'$, $v$ and $t$ be terms
  We have these congruence relations:

  If $\Gamma; v :: s \vdash u \mrelTwo{\equiv} u'$, then $\Gamma; s \vdash u \, v \mrelTwo{\equiv} u' \, v$.

  If $\Gamma, x \myleq t; \nil \vdash u \mrelTwo{\equiv} u'$ then $\Gamma; \nil \vdash \lambda x \myleq t. u \mrelTwo{\equiv} \lambda x \myleq t. u'$.

  If $\Gamma, x \myequiv \alpha; s \vdash u \mrelTwo{\equiv} u'$ then $\Gamma; \alpha :: s \vdash \lambda x \myleq t. u \mrelTwo{\equiv} \lambda x \myleq t. u'$.
\end{lemma}
\begin{proof}
  Similar proof as Lemma~\ref{lem:algorithmic-congruence-subtyping-reduction}.
\end{proof}

%%% Commutativity

\begin{lemma}[Commutativity - context weakening]
  \label{lem:commutativity-context-weakening}
  Let $\Gamma; s$ and $\Gamma'; s'$ be extended contexts such that
  $\Gamma; s \rightarrowtail \Gamma'; s'$.

  Then we have $\Gamma; \nil \rightarrowtail \Gamma'; \nil$.
\end{lemma}
\begin{proof}
  By induction on $s$.
  \begin{description}
    \item[If $s = \nil$:] Then the result holds.
    \item[If $s = \alpha :: s_0$:]
    Then by assumption we have $\Gamma; \alpha :: s_0 \rightarrowtail \Gamma'; s'$.
    By rule \textsc{Ct-Stk}, we then have $\Gamma; s_0 \rightarrowtail \Gamma'; s_1$, with $s' = \alpha' :: s_1$ for some term $\alpha'$.
    By induction we finally have $\Gamma; \nil \rightarrowtail \Gamma'; \nil$.
  \end{description}
\end{proof}

\end{document}